\documentclass[10pt]{article}

\usepackage{authblk}

\usepackage{caption2}

\usepackage{graphicx}

\usepackage[numbers]{natbib}        

\usepackage[all,cmtip]{xy}

\usepackage{setspace}
\usepackage{mathtools, bbm,youngtab}
\usepackage{amscd,mathrsfs}
\usepackage{amssymb,amsmath,dsfont,amsfonts,amsthm,stackengine}
\usepackage{euscript,enumerate}
\usepackage{undertilde}              
\usepackage{mathabx}  
\usepackage{stmaryrd}

\usepackage[cbgreek]{textgreek} 

\usepackage{booktabs}
\usepackage{multirow}

\usepackage{hhline} 

\usepackage{cases} 

\usepackage[retainorgcmds]{IEEEtrantools} 

\usepackage[dvipsnames]{xcolor}

\DeclareMathAlphabet{\mathpzc}{OT1}{pzc}{m}{it}

\usepackage{epstopdf}

\allowdisplaybreaks 

\makeatletter
\newcommand{\mathleft}{\@fleqntrue\@mathmargin0pt}
\makeatother

\usepackage{float}

 \theoremstyle{plain}
 \newtheorem{thm}{Theorem}
 \newtheorem{cor}[thm]{Corollary}

 \theoremstyle{definition}
 
 \newtheorem{rem}[thm]{Remark}
  
 \numberwithin{equation}{section}

\usepackage{hyperref}
\hypersetup{colorlinks=true, linkcolor=blue}
\hypersetup{colorlinks=true,citecolor=blue}

\makeatletter
\DeclareFontFamily{OMX}{MnSymbolE}{}
\DeclareSymbolFont{MnLargeSymbols}{OMX}{MnSymbolE}{m}{n}
\SetSymbolFont{MnLargeSymbols}{bold}{OMX}{MnSymbolE}{b}{n}
\DeclareFontShape{OMX}{MnSymbolE}{m}{n}{
    <-6>  MnSymbolE5
   <6-7>  MnSymbolE6
   <7-8>  MnSymbolE7
   <8-9>  MnSymbolE8
   <9-10> MnSymbolE9
  <10-12> MnSymbolE10
  <12->   MnSymbolE12
}{}
\DeclareFontShape{OMX}{MnSymbolE}{b}{n}{
    <-6>  MnSymbolE-Bold5
   <6-7>  MnSymbolE-Bold6
   <7-8>  MnSymbolE-Bold7
   <8-9>  MnSymbolE-Bold8
   <9-10> MnSymbolE-Bold9
  <10-12> MnSymbolE-Bold10
  <12->   MnSymbolE-Bold12
}{}

\let\llangle\@undefined
\let\rrangle\@undefined
\DeclareMathDelimiter{\llangle}{\mathopen}%
                     {MnLargeSymbols}{'164}{MnLargeSymbols}{'164}
\DeclareMathDelimiter{\rrangle}{\mathclose}%
                     {MnLargeSymbols}{'171}{MnLargeSymbols}{'171}
\makeatother

\title{\Large{\textbf{Compatible-Strain Mixed Finite Element Methods \\ for 3D Compressible and Incompressible Nonlinear Elasticity}}}
\author[1]{Mostafa Faghih Shojaei}
\author[1,2]{Arash Yavari\thanks{Corresponding author, e-mail: arash.yavari@ce.gatech.edu}}
\affil[1]{\small \textit{School of Civil and Environmental Engineering, Georgia Institute of Technology, Atlanta, GA 30332, USA}}
\affil[2]{\small \textit{The George W. Woodruff School of Mechanical Engineering, Georgia Institute of Technology, Atlanta, GA 30332, USA}}

\begin{document}

\maketitle

\begin{abstract}
A new family of mixed finite element methods--\emph{compatible-strain mixed finite element methods} (\mbox{CSFEMs})--are introduced for three-dimensional compressible and incompressible nonlinear elasticity. A Hu-Washizu-type functional is extremized in order to obtain a mixed formulation for nonlinear elasticity. The independent fields of the mixed formulations are the displacement, the displacement gradient, and the first Piola-Kirchhoff stress. A pressure-like field is also introduced in the case of incompressible elasticity. We define the displacement in $H^1$, the displacement gradient in $H(curl)$, the stress in $H(div)$, and the pressure-like field in $L^2$. In this setting, for improving the stability of the proposed finite element methods without compromising their consistency, we consider some stabilizing terms in the Hu-Washizu-type functional that vanish at its critical points. Using a conforming interpolation, the solution and the test spaces are approximated with some piecewise polynomial subspaces of them. In three dimensions, this requires using the N\'{e}d\'{e}lec edge elements for the displacement gradient and the N\'{e}d\'{e}lec face elements for the stress. This approach results in mixed finite element methods that satisfy the Hadamard jump condition and the continuity of traction on all internal faces of the mesh.  This, in particular, makes \mbox{CSFEMs} quite efficient for modeling heterogeneous solids.
We assess the performance of \mbox{CSFEMs} by solving several numerical examples, and demonstrate their good performance for bending problems, for bodies with complex geometries, and in the near-incompressible and the incompressible regimes. Using \mbox{CSFEMs}, one can capture very large strains and accurately approximate stresses and the pressure field. Moreover, in our numerical examples, we do not observe any numerical artifacts such as checkerboarding of pressure, hourglass instability, or locking.  
\end{abstract}

\begin{description}
\item[Keywords:] Mixed finite element methods; finite element exterior calculus; nonlinear elasticity; incompressible elasticity; Hilbert complex. 
\end{description}
{}
\tableofcontents{}


\section{Introduction}
It is known that the standard finite elements formulated in terms of the displacement field are not effective for various problems in nonlinear elasticity such as nearly incompressible or incompressible solids, bending analyses, capturing very large strains, and accurate calculation of stress. Developing finite element methods using the mixed formulations of elasticity is one path to overcome these limitations. However, it is a challenge to develop a robust and efficient mixed finite element method for nonlinear elasticity free from numerical instabilities and artifacts. This is more pronounced for problems in 3D as there is a much wider range of deformations in dimension three and 3D problems require more expensive computations. 
We recently proposed a new family of mixed finite element methods --- compatible-strain mixed finite element methods --- for 2D compressible \citep{AnFSYa2017} and incompressible \citep{FaYa2018} nonlinear elasticity. Our observations in several numerical examples indicated that these mixed methods have excellent performance in solving various 2D problems and do not suffer from numerical instabilities and artifacts including the difficulties mentioned earlier. In this paper, we extend these mixed methods to 3D compressible and incompressible nonlinear elasticity.

Over the years different approaches have been proposed in the finite element literature to capture large deformations of solids. Here we focus on some well-know works that are based on a mixed formulation and have proved promising for 3D nonlinear problems (see also our literature review in \citep{AnFSYa2017} and \citep{FaYa2018}). Mixed formulations are based on a saddle-point variational principle such as the two-field Hellinger-Reissner principle or the three-field Hu-Washizu principle, see \citep{arnold1990mixed} and \mbox{\citep[\S1.5]{Wr2009}}. One of the most commonly used schemes in the literature has been the enhanced strain method originally introduced by \citet{Simo1990} for infinitesimal strains and later extended to 2D and 3D nonlinear elasticity by \citet{SiAr1992} and \citet{simo1993improved}. In these methods, strain is assumed to be additively decomposed into a compatible part associated with the displacement field, and an enhanced part. The problem is then written as a two-field mixed formulation in terms of the displacement and the enhanced strain, which is derived from a three-field Hu-Washizu-type mixed formulation after eliminating the stress assuming that the enhanced strain and the stress are $L^2$-orthogonal. 
See \citep{armero2000locking} for a detailed discussion of early developments of enhanced strain methods and their locking and stability. 
Using an interpolation of strain and stress different from those in the original enhanced strain method, \citet{kasper2000mixedI} proposed a new mixed method for 2D and 3D linear elasticity called mixed-enhanced strain method and extended it to nonlinear elasticity in \citep{kasper2000mixedII}. 
\citet{lamichhane2006convergence} proposed a parameter-dependent modification of the standard Hu-Washizu mixed formulation for 2D linear elasticity and studied its uniform convergence in the  incompressible limit for different interpolations. Their study also incorporates the enhanced strain methods proposed in \citep{Simo1990} and \citep{kasper2000mixedI}. \citet{chavan2007locking} extended the approach introduced in \citep{lamichhane2006convergence} to 3D nonlinear elasticity considering a Mooney-Rivlin material model. By solving several 2D and 3D problems, they demonstrated the good performance of their method in bending problems and for nearly incompressible solids. \citet{reese2000new} introduced a new reduced-integration stabilized brick element method for 3D finite elasticity whose stabilization is based on the enhanced strain method. Their scheme is numerically efficient and shows robust performance in bending of thin shells, compression tests, and the near incompressible regime. 

The well-posedness of a mixed formulation requires that certain pairs of independent variables are defined in compatible spaces. This is commonly written as an inf-sup condition also known as the LBB condition named after the works of \citet{La1969}, \citet{Bab73}, and \citet{Bre74}. At the discrete level, the satisfaction of LBB condition is a necessary condition for the stability of mixed finite element methods. 
Thus, only particular combinations of finite element spaces for the independent variables result in convergent methods. 
Because of theses difficulties, the mixed formulations of 2D problems cannot simply be extended to 3D problems. In other words, one cannot use a mixed formulation with finite element spaces that converge in 2D and only switch the 2D elements with the counterpart 3D elements to obtain a convergent method. 
In \citep{FaYa2018} we formulated $96$ different four-field mixed finite element methods for 2D incompressible nonlinear elasticity by considering different combinations of first and second-order finite element spaces to independently approximate displacement, displacement gradient, stress, and pressure. 
By examining the linearized discrete systems of those mixed methods, we showed that $75$ out of $96$ of them result in singular tangent stiffness (Jacobian) matrices for any mesh and that only the remaining $21$ cases may result in convergent schemes. In this paper, using the same approach, we show that in 3D all the $96$ possible choices of the first and second-order four-field mixed methods lead to singular tangent stiffness matrices for any mesh and regardless of its size. 
To overcome this difficulty, we add some stabilization terms to the mixed formulations without compromising the consistency of their discretization schemes. This can also help to introduce a convergent mixed method with a fewer degrees of freedom, which is greatly beneficial for computationally expensive 3D problems. An example of such modification is the work of \citet{Hughes1986} on the Stokes problem, where they introduced a stabilized mixed finite element method using an equal-order $C^0$ interpolation of both velocity and pressure. Furthermore, inspired by the work of \citet{Hughes1986}, \citet{FrHu1988} developed a mixed finite element method for nearly incompressible linear elastic solids by adding stabilization terms to the weak formulation associated with the critical point of the Hellinger-Reissner principle. \citet{klaas1999} developed a stabilized displacement-pressure mixed finite element method for 3D finite elasticity by using linear shape functions for both displacement and pressure. In these works, the combinations of the finite element spaces are unstable according to the LBB condition and result in unphysical solutions. However, adding the stabilization terms resulted in convergent mixed methods. 
 
This paper is organized as follows. In \S2, we discuss the mixed formulations that are used in the 3D CSFEMs. In \S2.1, we review some preliminaries and definitions. In \S2.2, by defining suitable Hu-Washizu-type energy functionals we derive a three-field mixed formulation for compressible elastostatics and a four-field mixed formulation for incompressible elastostatics. 
In \S3, we discuss the finite element approximations for the proposed mixed formulations. In \S3.1, we define the finite elements (shape functions and degrees of freedom) for the displacement, displacement gradient, stress, and pressure. In \S3.2, we define the finite element approximation spaces and use them to introduce the mixed finite element methods in \S3.3. Next, the matrix formulation of the mixed finite elements are discussed in \S3.4. In \S3.5, we investigate singularities of the mixed methods for some combinations of finite element spaces and explain how the stabilizing terms remove those singularities. To study the performance of the 3D CSFEMs, we present several numerical examples in \S4 for both compressible and incompressible solids in dimension three. The paper ends by some concluding remarks in \S5.

\section{A Mixed Formulation for Nonlinear Elasticity}

In this section, following \citep{FaYa2018}, we present two mixed formulations one for 3D compressible nonlinear elasticity and one for 3D incompressible nonlinear elasticity.

\subsection{Preliminaries}

Suppose $\mathbf{X} = (\mathrm{X}^{1}, \mathrm{X}^{2}, \mathrm{X}^{3})\in\mathbb{R}^{3}$ is the position of a material point in the reference configuration $\mathcal{B}\subset\mathbb{R}^{3}$ with boundary $\partial\mathcal{B}$. For any vector field $\boldsymbol{U}$ and any $\binom{2}{0}$-tensor field $\boldsymbol{T}$, one can define $\binom{2}{0}$-tensors $\mathbf{grad}\,\boldsymbol{U}$ and $\mathbf{curl}\,\boldsymbol{T}$ and a vector field $\mathbf{div}\,\boldsymbol{T}$ with components 
\begin{equation*}
(\mathbf{grad}\,\boldsymbol{U})^{IJ}=\partial {U}^{I}/\partial\mathrm{X}^{J},~~ (\mathbf{curl}\,\boldsymbol{T})^{IJ}=\varepsilon_{JKL}\partial T^{IL}/\partial\mathrm{X}^{K}, ~~(\mathbf{div}\,\boldsymbol{T})^{I}=\partial T^{IJ}/\partial\mathrm{X}^{J},
\end{equation*}
where $\varepsilon_{JKL}$ is the standard permutation symbol, and summation convention for repeated indices is assumed.
Suppose $L^2(\mathcal{B})$, $L^{2}(T\mathcal{B})$, and $L^{2}(\otimes^2 T\mathcal{B})$ are the spaces of square integrable scalar fields, vector fields, and $\binom{2}{0}$-tensor fields, respectively. Define the following spaces:
\begin{align*}
H^{1}(T\mathcal{B}) &:=\left\{ \boldsymbol{U}\in L^{2}(T\mathcal{B}):  \mathbf{grad}\,\boldsymbol{U} \in L^{2}(\otimes^2 T\mathcal{B}) \right\}, \\
H^{\mathbf{c}}(\mathcal{B}) &:=\left\{ \boldsymbol{T} \in L^{2}(\otimes^2 T\mathcal{B}): \mathbf{curl}\,\boldsymbol{T}\in L^{2}(\otimes^2 T\mathcal{B}) \right\}, \\
H^{\mathbf{d}}(\mathcal{B}) &:=\left\{ \boldsymbol{T}\in L^{2}(\otimes^2 T\mathcal{B}): \mathbf{div}\,\boldsymbol{T}\in L^{2}(T\mathcal{B}) \right\}. \nonumber
\end{align*}
In the above spaces, $\mathbf{grad}$, $\mathbf{curl}\,$, and $\mathbf{div}$ are defined in the distributional sense. 
Recalling that \mbox{$\mathbf{curl}\,(\mathbf{grad}\,\boldsymbol{Y})=\boldsymbol{0}$} and \mbox{$\mathbf{div}(\mathbf{curl}\,\boldsymbol{T})=\boldsymbol{0}$}, one writes the following differential complex
\citep{AngoshtariYavari2014I,AngoshtariYavari2014II}:

\begin{equation}
\begin{gathered}
\scalebox{.9}{\xymatrix@C=3ex{ 
&  \text{displacements} \ar[r]  \ar@{<.>}[d] & \text{disp.~gradients} \ar[r]  \ar@{<.>}[d] & \text{compatibility}  \ar@{<.>}[d] &   & \\
 \mathbf{0} \ar[r] & H^{1}(T\mathcal{B}) \ar[r]^-{\mathbf{grad}}  & H^{\mathbf{c}}(\mathcal{B}) \ar[r]^-{\mathbf{curl}}  \ar@{<.>}[d] & H^{\mathbf{d}}(\mathcal{B})  \ar[r]^-{\mathbf{div}}  \ar@{<.>}[d] & L^{2}(T\mathcal{B}) \ar[r] \ar@{<.>}[d]  & \mathbf{0}  \\
  &   & \text{stress~functions} \ar[r] & \text{first~PK~stresses} \ar[r] & \text{equilibrium}  &  } } \nonumber 
\end{gathered}
\end{equation}
where the first arrow is a trivial operator sending zero to zero, and the last arrow indicates the zero operator mapping the $L^{2}$-space to zero.
The physical interpretation of this differential complex is as follows: Let $\boldsymbol{U}(\mathbf{X}):=\varphi(\mathbf{X})-\mathbf{X}$, $\mathbf{X}\in\mathcal{B}$, be the displacement field associated with a motion $\varphi:\mathcal{B}\rightarrow\mathbb{R}^{3}$. Then, $\boldsymbol{K}:=\mathbf{grad}\,\boldsymbol{U}$ is the displacement gradient and $\mathbf{curl}\,\boldsymbol{K}=\boldsymbol{0}$ is the necessary condition for the compatibility of $\boldsymbol{K}$. Moreover, given a first Piola-Kirchhoff stress tensor $\boldsymbol{P}$, the equilibrium equation $\mathbf{div}\,\boldsymbol{P}=\boldsymbol{0}$ is the necessary condition for the existence of a stress function $\boldsymbol{\Psi}$ such that $\boldsymbol{P}=\mathbf{curl}\,\boldsymbol{\Psi}$. This holds whenever $\boldsymbol{U} \in H^{1}(T\mathcal{B})$,  $\boldsymbol{K} \in \operatorname{ker}(\mathbf{curl}\,)\subset H^{\mathbf{c}}(\mathcal{B})$, and $\boldsymbol{P} \in \operatorname{ker}(\mathbf{div})\subset H^{\mathbf{d}}(\mathcal{B})$. 
The deformation gradient is defined as $\boldsymbol{F}:=\boldsymbol{I}+\boldsymbol{K}$, where $\boldsymbol{I}$ is the identity tensor, and $J:=\operatorname{det} \boldsymbol{F}$ (in Cartesian coordinates for both the reference and current configurations). One can show that $dv=J dV$, where $dV$ and $dv$ are the volume elements of the undeformed and deformed configurations, respectively. For incompressible solids, $J=1$. To weakly impose $J-1=0$, one considers a Lagrange multiplier $p$ as an independent field variable, which physically is realized as a pressure-like variable. 
At the discrete level, the restriction of $J$ to an element is a scaler describing the change of volume of that element \citep{Yavari2008}. Hence, one can assume that discrete pressure $p$ is also defined on each element, and in general, is not continuous across the element interfaces. Therefore, as a discontinuous scalar-valued field, $p \in L^2(\mathcal{B})$. 

\subsection{Mixed Formulations}

Let $\rho_0$ be the mass density of the body $\mathcal{B}$ and $\boldsymbol{B}$ be the body force per unit mass. Assume that the boundary of the body is a disjoint union of two subsets $\partial\mathcal{B}=\Gamma_d\sqcup\Gamma_t$ and is subjected to the displacement boundary condition $\boldsymbol{U}\big|_{\Gamma_d}=\overline{\boldsymbol{U}}$ and the traction boundary condition $(\boldsymbol{P}\boldsymbol{N})\big|_{\Gamma_t}=\overline{\boldsymbol{T}}$, where $\boldsymbol{N}$ is the unit outward normal vector field of $\partial\mathcal{B}$ in the reference configuration. Also, define  $H^{1}(T\mathcal{B},\Gamma_{d},\overline{\boldsymbol{U}}) :=\left\{ \boldsymbol{U}\in H^{1}(T\mathcal{B}): \boldsymbol{U}|_{\Gamma_{d}} =\overline{\boldsymbol{U}} \right\}$ and $H^{1}(T\mathcal{B},\Gamma_{d}) := H^{1}(T\mathcal{B},\Gamma_{d},\boldsymbol{0})$, where $\overline{\boldsymbol{U}}$ is of $H^{1/2}$-class. Suppose $ \langle,\rangle$ is the standard inner product of $\mathbb{R}^{3}$ and let $\llangle,\rrangle$ stand for the $L^{2}$-inner products of scalar, vector, and tensor fields, that is, $\llangle f,g \rrangle:=\int_{\mathcal{B}}fg\,dV$, $\llangle\boldsymbol{Y},\boldsymbol{Z}\rrangle:=\int_{\mathcal{B}}Y^{I}Z^{I}dV$, and $\llangle\boldsymbol{S},\boldsymbol{T}\rrangle:=\int_{\mathcal{B}}S^{IJ}T^{IJ}dV$. Then, one can define a Hu-Washizu-type functional \mbox{$\mathcal{I}: H^{1}(T\mathcal{B},\Gamma_{d},\overline{\boldsymbol{U}})\times H^{\mathbf{c}}(\mathcal{B}) \times H^{\mathbf{d}}(\mathcal{B}) =:\EuScript{D}\rightarrow\mathbb{R}$} as 
\begin{equation}\label{HWF}
  \mathcal{I}(\boldsymbol{U},\boldsymbol{K},\boldsymbol{P})
  =\int_{\mathcal{B}}W(\mathbf{X},\boldsymbol{K})dV - \llangle \boldsymbol{P},\boldsymbol{K}
  -\mathbf{grad}\,\boldsymbol{U}\rrangle - \llangle \rho_{0}\boldsymbol{B},\boldsymbol{U}\rrangle 
  - \int_{\Gamma_{t}}\langle \overline{\boldsymbol{T}},\boldsymbol{U}\rangle dA,
\end{equation}
where $W(\mathbf{X},\boldsymbol{K})$ is the stored energy function of a hyperelastic material. For an isotropic solid, the energy function can be written as $W = \widehat{W}(\mathbf{X},I_1,I_2,I_3)$, where $I_1=\operatorname{tr} \boldsymbol{C}$, $I_2=\frac{1}{2}[(\operatorname{tr} \boldsymbol{C})^2 -  \operatorname{tr} \boldsymbol{C}^2]$, and $I_3=\operatorname{det} \boldsymbol{C}$ are the invariants of the right Cauchy-Green deformation tensor $\boldsymbol{C}=\boldsymbol{F}^{\mathsf{T}}\boldsymbol{F}$. For incompressible solids, $J=\sqrt{I_3}=1$, and one modifies (\ref{HWF}) by defining 
\begin{equation}\label{HWFicmp}
  \overline{\mathcal{I}}(\boldsymbol{U},\boldsymbol{K},\boldsymbol{P},p)
  =\mathcal{I}(\boldsymbol{U},\boldsymbol{K},\boldsymbol{P})\Big|_{J(\boldsymbol{K})=1} 
  + \int_{\mathcal{B}}p\,C\big(J(\boldsymbol{K})\big)dV,
\end{equation}
where $C:\mathbb{R}^{+}\rightarrow \mathbb{R}$ is a smooth function such that $C(J)=0$ if and only if $J=1$ and $p \in L^{2}(\mathcal{B})$ is a pressure-like scalar field, to which we may refer simply as pressure. For 3D computations, in order to improve the stability of the mixed finite element methods, a stabilizing term is added to (\ref{HWFicmp}) as
\begin{equation}\label{HWFicmpStab}
  \mathcal{J}(\boldsymbol{U},\boldsymbol{K},\boldsymbol{P},p)
  = \overline{\mathcal{I}}(\boldsymbol{U},\boldsymbol{K},\boldsymbol{P},p) 
  +\frac{\alpha}{2}\llangle \boldsymbol{K}-\mathbf{grad}\,\boldsymbol{U},\boldsymbol{K}
  -\mathbf{grad}\,\boldsymbol{U}\rrangle,
\end{equation}
where $\alpha\geq0$ is a penalty constant for enforcing $\boldsymbol{K}=\mathbf{grad}\,\boldsymbol{U}$.
Extremizing (\ref{HWFicmpStab}), as discussed in \mbox{\cite[\S 2.2]{FaYa2018}}, results in the following weak formulation of the boundary-value problem for incompressible nonlinear elastostatics:

\bigskip
\begin{minipage}{0.95\textwidth}{\textit{Given a body force $\boldsymbol{B}$ of $L^{2}$-class, a boundary displacement $\overline{\boldsymbol{U}}$ on $\Gamma_{d}$ of $H^{1/2}$-class, a boundary traction $\overline{\boldsymbol{T}}$ on $\Gamma_{t}$ of $L^{2}$-class, and a stability constant $\alpha\geq0$, find $(\boldsymbol{U},\boldsymbol{K},\boldsymbol{P},p)\in H^{1}(T\mathcal{B},\Gamma_{d},\overline{\boldsymbol{U}})\times H^{\mathbf{c}}(\mathcal{B}) \times H^{\mathbf{d}}(\mathcal{B})\times L^{2}(\mathcal{B})$ such that}}
\begin{equation}\label{ElasMixF_1}
\begin{alignedat}{3}
\llangle\boldsymbol{P},\mathbf{grad}\,\boldsymbol{\Upsilon}\rrangle + \alpha s_1\!\left(\boldsymbol{U},\boldsymbol{K},\boldsymbol{\Upsilon}\right)&=f(\boldsymbol{\Upsilon}), &\quad &\forall \boldsymbol{\Upsilon}\in H^{1}(T\mathcal{B},\Gamma_{d}),\\
\llangle \widetilde{\boldsymbol{P}}(\boldsymbol{K}),\boldsymbol{\kappa}\rrangle - \llangle\boldsymbol{P},\boldsymbol{\kappa}\rrangle + \llangle p\boldsymbol{Q}(\boldsymbol{K}),\boldsymbol{\kappa}\rrangle + \alpha s_2\!\left(\boldsymbol{U},\boldsymbol{K},\boldsymbol{\kappa}\right)&= 0,& &\forall \boldsymbol{\kappa}\in H^{\mathbf{c}}(\mathcal{B}),\\
\llangle \mathbf{grad}\,\boldsymbol{U}, \boldsymbol{\pi} \rrangle - \llangle \boldsymbol{K},\boldsymbol{\pi}\rrangle &= 0, & &\forall \boldsymbol{\pi}\in H^{\mathbf{d}}(\mathcal{B}),\\
\llangle C(J), q\rrangle &= 0, & &\forall q\in L^{2}(\mathcal{B}),
\end{alignedat}
\end{equation}
\textit{where}
\begin{equation}\label{traction}
 f(\boldsymbol{\Upsilon}) = \llangle \rho_{0}\boldsymbol{B},\boldsymbol{\Upsilon} \rrangle + \int_{\Gamma_{t}}\langle \overline{\boldsymbol{T}},\boldsymbol{\Upsilon}\rangle \,dA,
\end{equation}
\textit{and}
\begin{equation}\label{stab}
\begin{alignedat}{3}
s_1\!\left(\boldsymbol{U},\boldsymbol{K},\boldsymbol{\Upsilon}\right) &= \llangle\mathbf{grad}\,\boldsymbol{U},\mathbf{grad}\,\boldsymbol{\Upsilon}\rrangle-\llangle\boldsymbol{K},\mathbf{grad}\,\boldsymbol{\Upsilon}\rrangle,\\
s_2\!\left(\boldsymbol{U},\boldsymbol{K},\boldsymbol{\kappa}\right) &= \llangle\boldsymbol{K},\boldsymbol{\kappa}\rrangle-\llangle\mathbf{grad}\,\boldsymbol{U},\boldsymbol{\kappa}\rrangle.
 \end{alignedat}
\end{equation}
\end{minipage}
\bigskip
 \\
 In \eqref{ElasMixF_1}$_2$, $\widetilde{\boldsymbol{P}}(\boldsymbol{K})={\partial\widetilde{W}}/{\partial\boldsymbol{K}}$ with $\widetilde{W}=\widehat{W}(\mathbf{X},I_1,I_2,I_3)\big|_{I_3=1}$ is the constitutive part of the stress, and $\boldsymbol{Q}(\boldsymbol{K})={\partial C}/{\partial\boldsymbol{K}}=C^\prime(J)(\boldsymbol{F}^{-1})^{\mathsf{T}}$ is the contribution of the incompressibility constraint $J=1$. Note that setting \mbox{$\alpha=0$} results in the standard weak formulation of incompressible nonlinear elastostatics \cite[\S 2.2]{FaYa2018}. The solutions of the above weak formulation are the critical points of the functional (\ref{HWFicmpStab}). Using Green's formula $\llangle \mathbf{div}\boldsymbol{P}, \boldsymbol{\Upsilon}\rrangle= -\llangle \boldsymbol{P}, \mathbf{grad}\boldsymbol{\Upsilon}\rrangle + \int_{\partial\mathcal{B}}\langle \boldsymbol{P}\boldsymbol{N}, \boldsymbol{\Upsilon}\rangle \,dA, ~\forall \boldsymbol{\Upsilon}\in H^{1}(T\mathcal{B},\Gamma_{d})$ and assuming that $\int_{\partial\mathcal{B}}\langle \boldsymbol{P}\boldsymbol{N}, \boldsymbol{\Upsilon}\rangle \,dA = \int_{\Gamma_{t}}\langle \overline{\boldsymbol{T}},\boldsymbol{\Upsilon}\rangle \,dA$ holds, $\forall \boldsymbol{\Upsilon}\in H^{1}(T\mathcal{B},\Gamma_{d})$, one can show that (\ref{ElasMixF_1}) results in the following set of governing equations for incompressible nonlinear elastostatics:
\begin{subequations}\label{3DNonLinStrong}
\begin{IEEEeqnarray}{rlrll}
& \mathbf{div}\boldsymbol{P} + \rho_{0} \boldsymbol{B} = \boldsymbol{0}, & &\quad &
\\*
& \boldsymbol{P} =\widetilde{\boldsymbol{P}}(\boldsymbol{K})+p\boldsymbol{Q}(\boldsymbol{K}), \label{3DNonLinStrong2}
\vspace*{-7pt}
\\*
& &\smash{\left.\IEEEstrut[32pt]\right\}} & &  \text{on } \mathcal{B}, \nonumber
\vspace*{-8pt}
\\*
& \boldsymbol{K}=\mathbf{grad}~\boldsymbol{U}, & & &
\\*
& J=1, & & &
\vspace*{3pt}
\\*
& \boldsymbol{U}=\overline{\boldsymbol{U}}, & & &\text{on } \Gamma_{d},
\\*
& \boldsymbol{P}\boldsymbol{N}=\overline{\boldsymbol{T}}, & & &\text{on } \Gamma_{t}.
\end{IEEEeqnarray}
\end{subequations}
Conversely, one can obtain (\ref{ElasMixF_1}) from (\ref{3DNonLinStrong}), see \cite[\S 2.2]{AnFSYa2017}. Note that (\ref{3DNonLinStrong2}) is the constitutive relation of an incompressible solid, which in terms of the Cauchy stress reads $\boldsymbol{\sigma}=\widetilde{\boldsymbol{P}}(\boldsymbol{K})\boldsymbol{F}^{\mathsf{T}}+\bar{p}\boldsymbol{I}$, where $\bar{p} = p\,C^\prime(J)$. Note that adding the stabilizing terms (\ref{stab}) to the weak formulation (\ref{ElasMixF_1}) does not change the set of governing equations (\ref{3DNonLinStrong}). In other words, these terms will vanish for the exact solutions of (\ref{ElasMixF_1}). Hence, with proper discretization, the extra terms (\ref{stab}) may improve the stability of the resulting mixed finite element methods without compromising their consistency (see \citep{arnold2015stability} for consistency and stability). We discuss this further in \S\ref{solvability}.

By setting $p=q=0$ in (\ref{ElasMixF_1}) and replacing $\widetilde{\boldsymbol{P}}(\boldsymbol{K})$ with $\widehat{\boldsymbol{P}}(\boldsymbol{K})={\partial\widehat{W}}/{\partial\boldsymbol{K}}$, one can readily arrive at the following weak formulation of the boundary-value problem of compressible nonlinear elastostatics:

\bigskip
\begin{minipage}{0.95\textwidth}{\textit{Given $\boldsymbol{B}$, $\overline{\boldsymbol{U}}$, $\overline{\boldsymbol{T}}$, and $\alpha\geq0$, find $(\boldsymbol{U},\boldsymbol{K},\boldsymbol{P})\in H^{1}(T\mathcal{B},\Gamma_{d},\overline{\boldsymbol{U}})\times H^{\mathbf{c}}(\mathcal{B}) \times H^{\mathbf{d}}(\mathcal{B})$ such that}}
\begin{equation}\label{ElasMixF_2}
\begin{alignedat}{3}
\llangle\boldsymbol{P},\mathbf{grad}\,\boldsymbol{\Upsilon}\rrangle + \alpha s_1\!\left(\boldsymbol{U},\boldsymbol{K},\boldsymbol{\Upsilon}\right)&=f(\boldsymbol{\Upsilon}), &\quad &\forall \boldsymbol{\Upsilon}\in H^{1}(T\mathcal{B},\Gamma_{d}),\\
\llangle \widehat{\boldsymbol{P}}(\boldsymbol{K}),\boldsymbol{\kappa}\rrangle - \llangle\boldsymbol{P},\boldsymbol{\kappa}\rrangle + \alpha s_2\!\left(\boldsymbol{U},\boldsymbol{K},\boldsymbol{\kappa}\right)&= 0,& &\forall \boldsymbol{\kappa}\in H^{\mathbf{c}}(\mathcal{B}),\\
\llangle \mathbf{grad}\,\boldsymbol{U}, \boldsymbol{\pi} \rrangle - \llangle \boldsymbol{K},\boldsymbol{\pi}\rrangle &= 0, & &\forall \boldsymbol{\pi}\in H^{\mathbf{d}}(\mathcal{B}).
\end{alignedat}
\end{equation}
\end{minipage}
\bigskip
\\
Similarly, one can show that (\ref{ElasMixF_2}) results in the following set of governing equations for compressible nonlinear elastostatics:
\begin{subequations}\label{3DNonLinStrong_cmp}
\begin{IEEEeqnarray}{rlrll}
& \mathbf{div}\,\boldsymbol{P}+\rho_{0} \boldsymbol{B}=\boldsymbol{0}, & &\quad & \label{3DNonLinStrong_cmp1}
\\*
& \boldsymbol{P}=\widehat{\boldsymbol{P}}(\boldsymbol{K}),  &\smash{\left.
\IEEEstrut[8\jot]
\right\}} & & \text{on } \mathcal{B}, \label{3DNonLinStrong_cmp2}
\\*
& \boldsymbol{K}=\mathbf{grad}\,\boldsymbol{U}, & & & \label{3DNonLinStrong_cmp3}
\\*
& \boldsymbol{U}=\overline{\boldsymbol{U}}, & & &\text{on } \Gamma_{d}, \label{3DNonLinStrong_cmp4}
\\*
& \boldsymbol{P}\boldsymbol{N}=\overline{\boldsymbol{T}}, & & &\text{on } \Gamma_{t}. \label{3DNonLinStrong_cmp5}
\end{IEEEeqnarray}
\end{subequations}

\section{Finite Element Approximations}

\newcommand{\vek}[1] {\lceil#1\rceil}
\newcommand{\bigVek}[2] {#1\lceil#2 #1\rceil} 
\newcommand{\jac} {\boldsymbol{\mathsf{J}}_\EuScript{T}}
\newcommand{\detJ}{\mathrm{det}\jac}
\newcommand{\bo}{\boldsymbol{\mathsf{b}}^1}
\newcommand{\go}{\boldsymbol{\mathsf{g}}^1}
\newcommand{\bc}{\boldsymbol{\mathsf{b}}^\mathbf{c}}
\newcommand{\bd}{\boldsymbol{\mathsf{b}}^\mathbf{d}}
\newcommand{\fes}[3]{\mathcal{P}^{\mathbf{#1}#2}_{#3}(\mathcal{B}_{h})}
\newcommand{\Rm}[2] {\boldsymbol{\mathsf{#1}}^{#2}}                     
\newcommand{\lm}[2] {\boldsymbol{\mathsf{#1}}^{#2}_{\EuScript{T}} }     
\newcommand{\lmi}[3] {\mathsf{#1}^{#2}_{\EuScript{T},{#3}} }            
\newcommand{\gm}[2] {\boldsymbol{\mathsf{#1}}^{#2}_{h}}   
\newcommand{\bz}{\boldsymbol{0}}
\newcommand{\base}[6]{
  \begin{bmatrix}
    #1 &#4 &#2 &#5 &\cdots &#3 &#6 \\
    #6 &#1 &#4 &#2 &\cdots &#5 &#3
 \end{bmatrix}
}
\newcommand{\bbase}[3]{\begin{bmatrix}  #1 &#2 &\cdots &#3 \end{bmatrix}}
\newcommand{\gradN}{\mathbf{grad}\mathsf{N}}
\newcommand{\divN}{\mathbf{div}\,\bd}
\newcommand{\st}{o_\EuScript{T}}

\subsection{Finite Elements}\label{FEs}

Suppose $\mathcal{P}_{r}(\mathbb{R}^{3})$ is the space of $\mathbb{R}$-valued polynomials in three variables $\{\mathrm{X}^{1},\mathrm{X}^{2},\mathrm{X}^{3}\}$ of degree at most $r\geq0$ and suppose $\mathcal{H}_{r}(\mathbb{R}^{3})\subset \mathcal{P}_{r}(\mathbb{R}^{3})$ is the space of homogeneous polynomials of degree $r$, that is, all the terms of the members of $\mathcal{H}_{r}(\mathbb{R}^{3})$ are of degree $r$. For $r<0$, these spaces are assumed to be empty. By $\mathcal{P}_{r}(T\mathbb{R}^{3})$ and $\mathcal{P}_{r}(\text{\large$\otimes$}^{2}T\mathbb{R}^{3})$ we denote the spaces of polynomial vector and $\binom{2}{0}$-tensor fields in $\mathbb{R}^{3}$ with Cartesian components in $\mathcal{P}_{r}(\mathbb{R}^{3})$. The spaces $\mathcal{H}_{r}(T\mathbb{R}^{3})$ and $\mathcal{H}_{r}(\text{\large$\otimes$}^{2}T\mathbb{R}^{3})$ are defined similarly. 
Next define the following subspaces of $\mathcal{P}_{r}(T\mathbb{R}^{3})$:
\begin{align*}
\mathcal{P}^{-}_{r}(T\mathbb{R}^{3})&:=\mathcal{P}_{r-1}(T\mathbb{R}^{3}) \oplus \mathrm{L}_{1}\!\left(  \mathcal{H}_{r-1}(T\mathbb{R}^{3})  \right), \\
\mathcal{P}^{\ominus}_{r}(T\mathbb{R}^{3}) &:= \mathcal{P}_{r-1}(T\mathbb{R}^{3}) \oplus \mathrm{L}_{2}\!\left(  \mathcal{H}_{r-1}(\mathbb{R}^{3})  \right),
\end{align*}
where $(\mathrm{L}_{1}(\boldsymbol{Y}))^{I}=\varepsilon_{IJL}\mathrm{X}^{L}Y^{J}$ for any vector field $\boldsymbol{Y}$, and $(\mathrm{L}_{2}(f))^{I}=\mathrm{X}^{I}f$ for any scalar field $f$. Similarly, one defines the following subspaces of $\mathcal{P}_{r}(\text{\large$\otimes$}^{2}T\mathbb{R}^{3})$:
\begin{align*}
\mathcal{P}^{-}_{r}(\text{\large$\otimes$}^{2}T\mathbb{R}^{3})&:=\mathcal{P}_{r-1}(\text{\large$\otimes$}^{2}T\mathbb{R}^{3}) \oplus \mathbf{L}_{1}\!\left(  \mathcal{H}_{r-1}(\text{\large$\otimes$}^{2}T\mathbb{R}^{3})  \right), \\
\mathcal{P}^{\ominus}_{r}(\text{\large$\otimes$}^{2}T\mathbb{R}^{3}) &:= \mathcal{P}_{r-1}(\text{\large$\otimes$}^{2}T\mathbb{R}^{3}) \oplus \mathbf{L}_{2}\!\left(  \mathcal{H}_{r-1}(T\mathbb{R}^{3})  \right), 
\end{align*}
where $(\mathbf{L}_{1}(\boldsymbol{T}))^{IJ} =\varepsilon_{JLK}\mathrm{X}^{K}T^{IL}$ for any $\binom{2}{0}$-tensor field $\boldsymbol{T}$, and $ (\mathbf{L}_{2}(\boldsymbol{Y}))^{IJ}=\mathrm{X}^{J}Y^{I}$ for any vector field $\boldsymbol{Y}$. One can show that 
\begin{equation*}
\begin{aligned}
 \dim \mathcal{P}_{r}(\text{\large$\otimes$}^{2}T\mathbb{R}^{3})  = 3\dim \mathcal{P}_{r}(T\mathbb{R}^{3})=9\dim\mathcal{P}_{r}(\mathbb{R}^{3})&= \frac{3}{2}(r+1)(r+2)(r+3), \\
\dim \mathcal{P}^{-}_{r}(\text{\large$\otimes$}^{2}T\mathbb{R}^{3}) = 3\dim\mathcal{P}^{-}_{r}(T\mathbb{R}^{3}) &=\frac{3}{2}r(r+2)(r+3), \\
\dim \mathcal{P}^{\ominus}_{r}(\text{\large$\otimes$}^{2}T\mathbb{R}^{3}) = 3\dim\mathcal{P}^{\ominus}_{r}(T\mathbb{R}^{3})&= \frac{3}{2}r(r+1)(r+3).\\
\end{aligned} 
\end{equation*}

Let $\widehat{\EuScript{T}}$ be a reference tetrahedral element with coordinates $\boldsymbol{\xi}=(\xi^1,\xi^2,\xi^3)$ shown in Figure \ref{RefElSch}. We denote the edges of $\widehat{\EuScript{T}}$ by $\widehat{\EuScript{E}}_{i}, i=1,2,\dots,6$ and their corresponding lengths by $\hat{\ell}_i, i=1,2,\dots,6$, and the faces of $\widehat{\EuScript{T}}$ by $\widehat{\EuScript{F}}_{i}, i=1,2,3,4$ and their corresponding areas by $\hat{A}_i, i=1,2,3,4$. For an edge joining two vertices $i$ and $j$, one defines a unique orientation as $i \rightarrow j$, where $i < j$.
We also define a unit tangent vector $\hat{\boldsymbol{\mathsf{t}}}_i$ on each edge such that it agrees with the edge orientation. Moreover, on each face containing three edges $\widehat{\EuScript{E}}_{i}$, $\widehat{\EuScript{E}}_{j}$, and $\widehat{\EuScript{E}}_{k}$, we define a unit normal vector $\hat{\boldsymbol{\mathsf{n}}}_l=\hat{\boldsymbol{\mathsf{t}}}_i\times\hat{\boldsymbol{\mathsf{t}}}_j$, where $i < j< k$.
\begin{figure}[h]
\begin{center}
\includegraphics[scale = 0.9 ,angle = 0]{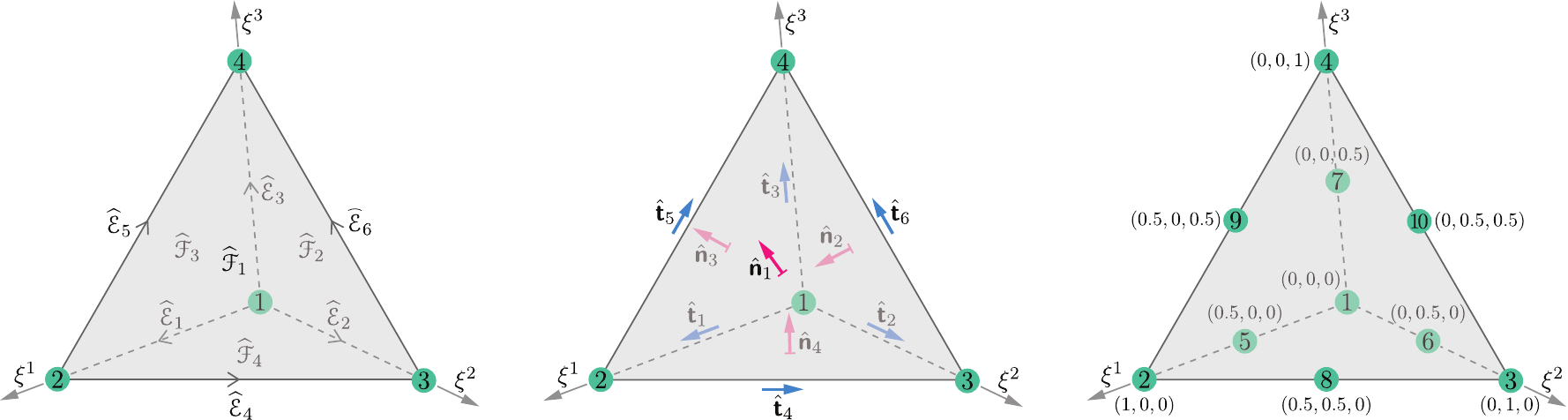}
\end{center}
\vspace*{-0.1in}
\caption{\footnotesize The four-node reference element and the edge and face numbers (left), the reference unit tangent and normal vectors (middle), and the ten-node reference element (right).}
\label{RefElSch}
\end{figure}

Following \citep{Ciarlet1978,ErnGuermond2004}, we define a finite element as a triplet $(\EuScript{T},\mathcal{P}(\EuScript{T}),\text{\textit{\textSigma}})$, where $\EuScript{T}$ is a tetrahedron in $\mathbb{R}^{3}$, $\mathcal{P}(\EuScript{T})$ is a space of polynomials on $\EuScript{T}$, and $\text{\textit{\textSigma}}$ is a set of $\mathbb{R}$-valued linear functionals acting on the members of $\mathcal{P}(\EuScript{T})$. The members of $\text{\textit{\textSigma}}$ are called the local degrees of freedom (DOF) and the local shape functions form a basis for $\mathcal{P}(\EuScript{T})$ (see \citep[\S3.1]{FaYa2018}). We consider the following \emph{reference} finite elements for the four field variables:
\begin{equation} \label{RFEs}
\begin{aligned}
&\left(\widehat{\EuScript{T}},\mathcal{P}_{r}(T\widehat{\EuScript{T}}),\text{\textit{\textSigma}}^{\widehat{\EuScript{T}},1}\right), &&\text{ for displacement } \boldsymbol{U},\\
&\left(\widehat{\EuScript{T}},\mathcal{P}^{-}_{r}(\text{\large$\otimes$}^{2}T\widehat{\EuScript{T}}), \text{\textit{\textSigma}}^{\widehat{\EuScript{T}},\mathbf{c}-}\right), \left(\widehat{\EuScript{T}},\mathcal{P}_{r}(\text{\large$\otimes$}^{2}T\widehat{\EuScript{T}}), \text{\textit{\textSigma}}^{\widehat{\EuScript{T}},\mathbf{c}}\right), &&\text{ for displacement gradient } \boldsymbol{K},   \\
&\left(\widehat{\EuScript{T}},\mathcal{P}^{\ominus}_{r}(\text{\large$\otimes$}^{2}T\widehat{\EuScript{T}}), \text{\textit{\textSigma}}^{\widehat{\EuScript{T}},\mathbf{d}-}\right), \left(\widehat{\EuScript{T}},\mathcal{P}_{r}(\text{\large$\otimes$}^{2}T\widehat{\EuScript{T}}), \text{\textit{\textSigma}}^{\widehat{\EuScript{T}},\mathbf{d}}\right), &&\text{ for stress } \boldsymbol{P},\\
&\left(\widehat{\EuScript{T}},\mathcal{P}_{r}(\widehat{\EuScript{T}}),\text{\textit{\textSigma}}^{\widehat{\EuScript{T}},\ell}\right), &&\text{ for the pressure-like field } p.
\end{aligned}{}
\end{equation}
Note that $\mathcal{P}_{r}(\widehat{\EuScript{T}}) = \mathcal{P}_{r}(\mathbb{R}^{3})\big|_{\widehat{\EuScript{T}}}$ and $\mathcal{P}_{r}(T\widehat{\EuScript{T}})$, $\mathcal{P}_{r}(\text{\large$\otimes$}^{2}T\widehat{\EuScript{T}})$, $\mathcal{P}^{-}_{r}(\text{\large$\otimes$}^{2}T\widehat{\EuScript{T}})$, and $\mathcal{P}^{\ominus}_{r}(\text{\large$\otimes$}^{2}T\widehat{\EuScript{T}})$ are defined similarly. The finite element for $\boldsymbol{U}$ is based on the standard Lagrange finite elements. For a vector field $\boldsymbol{V}:\widehat{\EuScript{T}}\rightarrow\mathbb{R}^3$, the set of local degrees of freedom is $\text{\textit{\textSigma}}^{\widehat{\EuScript{T}},1}=\{V^1(\boldsymbol{\xi}_1),V^2(\boldsymbol{\xi}_1),V^3(\boldsymbol{\xi}_1),\dots,V^1(\boldsymbol{\xi}_{m}),V^2(\boldsymbol{\xi}_{m}),V^3(\boldsymbol{\xi}_{m})\}$, where $\boldsymbol{\xi}_i$ contains the coordinates of the $i$-th node of the $m$-node $\widehat{\EuScript{T}}$, where $m=4$ ($m=10$) for $r=1$ ($r=2$). For $r=1,2$, a basis of the polynomial space $\mathcal{P}_{r}(T\widehat{\EuScript{T}})$ includes
\begin{equation}\label{sf1}
\boldsymbol{h}^{\widehat{\EuScript{T}}}_{3i-2} = \begin{bmatrix} l^{r}_i \\ 0 \\ 0 \end{bmatrix},\quad
\boldsymbol{h}^{\widehat{\EuScript{T}}}_{3i-1} = \begin{bmatrix} 0 \\ l^{r}_i \\ 0\end{bmatrix}, \quad
\boldsymbol{h}^{\widehat{\EuScript{T}}}_{3i} = \begin{bmatrix} 0 \\ 0 \\ l^{r}_i \end{bmatrix}, \quad 
i = 1,2,...,m.
\end{equation}
The Lagrange polynomials $l^{r}_i$ for the four-node reference tetrahedron $\widehat{\EuScript{T}}$ are
\begin{equation}\label{lagFun3}
 l^{1}_1 = 1-\xi^1-\xi^2-\xi^3,\quad
 l^{1}_2 = \xi^1,\quad
 l^{1}_3 = \xi^2,\quad
l^{1}_4 = \xi^3. 
\end{equation}
For the ten-node $\widehat{\EuScript{T}}$, the Lagrange polynomials are $l^{2}_i = l^{1}_i(2l^{1}_i-1)$ for the nodes at the vertices $i=1,2,3,4$ and $l^{2}_{k} = 4l^{1}_{i}l^{1}_{j}$ for the middle node of each edge joining vertices $i$ and $j$ as shown in Figure \ref{RefElSch}.  We will use $\mathcal{P}_{r}(T\widehat{\EuScript{T}}), r=1,2$ spanned by $\boldsymbol{h}^{\widehat{\EuScript{T}}}_l, l=1,2,...,3m$ to construct the approximation space of $\boldsymbol{U}$. To interpolate $\boldsymbol{K}\in H^{\mathbf{c}}$, we define two finite elements given in \eqref{RFEs}$_2$ based on the N\'{e}d\'{e}lec $1^{\mathrm{st}}$-kind \emph{edge} elements in $\mathbb{R}^3$ (NE1) \citep{Ned1980} and the N\'{e}d\'{e}lec $2^{\mathrm{nd}}$-kind \emph{edge} elements in $\mathbb{R}^3$ (NE2) \citep{Ned1986}, respectively. Let $\overrightarrow{\boldsymbol{T}}_{\!\!I}:= \begin{bmatrix} T^{I1} & T^{I2} & T^{I3}\end{bmatrix}^{\mathsf{T}}$ be a vector containing the elements of the $I$-th row of a $\binom{2}{0}$-tensor $\boldsymbol{T}$. The set of the local degrees of freedom $\text{\textit{\textSigma}}^{\widehat{\EuScript{T}},\mathbf{c}-}$ ($\text{\textit{\textSigma}}^{\widehat{\EuScript{T}},\mathbf{c}}$) in \eqref{RFEs}$_2$ is defined as $\left\{ \phi^{\widehat{\EuScript{T}},\widehat{\EuScript{E}}_{k}}_{I,J}, \phi^{\widehat{\EuScript{T}},\widehat{\EuScript{F}}_{l}}_{I,J}, \phi^{\widehat{\EuScript{T}},\widehat{\EuScript{T}}}_{I,J} \right\}$, where
\begin{equation}\label{DOF_hc}
\begin{aligned}
\phi^{\widehat{\EuScript{T}},\widehat{\EuScript{E}}_{k}}_{I,J}(\boldsymbol{T}) &=\int_{\widehat{\EuScript{E}}_k} f_J\langle \overrightarrow{\boldsymbol{T}}_{\!\!I} ,\hat{\boldsymbol{\mathsf{t}}}_k \rangle \,d\hat{s}, && \forall f_J\text{ that form a basis for } \mathcal{P}_{r-1}(\mathbb{R}^{3})\big|_{\widehat{\EuScript{E}}_k} \left(\mathcal{P}_{r}(\mathbb{R}^{3})\big|_{\widehat{\EuScript{E}}_k}\right),\\
\phi^{\widehat{\EuScript{T}},\widehat{\EuScript{F}}_{l}}_{I,J}(\boldsymbol{T}) & =\int_{\widehat{\EuScript{F}}_l} \langle \overrightarrow{\boldsymbol{T}}_{\!\!I}\times\boldsymbol{Y}_{\!\!J},\hat{\boldsymbol{\mathsf{n}}}_l \rangle \,d\hat{A}, && \forall \boldsymbol{Y}_{\!\!J} \text{ that form a basis for } \mathcal{P}_{r-2}(T\mathbb{R}^{3})\big|_{\widehat{\EuScript{F}}_l} \left(\mathcal{P}^{-}_{r-1}(T\mathbb{R}^{3})\big|_{\widehat{\EuScript{F}}_l}\right),   \\
\phi^{\widehat{\EuScript{T}},\widehat{\EuScript{T}}}_{I,J}(\boldsymbol{T}) &=\int_{\widehat{\EuScript{T}}}\langle \overrightarrow{\boldsymbol{T}}_{\!\!I} , \boldsymbol{Z}_{\!J} \rangle \,d\hat{V}, && \forall\boldsymbol{Z}_{\!J} \text{ that form a basis for } \mathcal{P}_{r-3}(T\mathbb{R}^{3})\big|_{\widehat{\EuScript{T}}} \left(\mathcal{P}^{\ominus}_{r-2}(T\mathbb{R}^{3})\big|_{\widehat{\EuScript{T}}}\right).
\end{aligned}{}
\end{equation}
We next discuss the corresponding local shape functions of the two finite elements given in \eqref{RFEs}$_2$. Let $\boldsymbol{v}_{J}^{\widehat{\EuScript{T}},\widehat{\EuScript{E}}_{k}}$, $\boldsymbol{v}_{J}^{\widehat{\EuScript{T}},\widehat{\EuScript{F}}_{l}}$, and $\boldsymbol{v}_{J}^{\widehat{\EuScript{T}},\widehat{\EuScript{T}}}$ denote the shape functions of those N\'{e}d\'{e}lec edge elements that are associated to the $k-$th edge of $\widehat{\EuScript{T}}$, the $l-$th face of $\widehat{\EuScript{T}}$, and the entire $\widehat{\EuScript{T}}$, respectively. We consider these vector-valued polynomials in $\mathbb{R}^3$ as row vectors and define the following tensorial shape functions:
\begin{equation}\label{sfc}
  \boldsymbol{r}_{1,J}^{\widehat{\EuScript{T}},\widehat{\EuScript{E}}_{k}} = \begin{bmatrix} \boldsymbol{v}_{J}^{\widehat{\EuScript{T}},\widehat{\EuScript{E}}_{k}} \\ \boldsymbol{0} \\ \boldsymbol{0} \end{bmatrix}_{3\times3},\quad 
  \boldsymbol{r}_{2,J}^{\widehat{\EuScript{T}},\widehat{\EuScript{E}}_{k}} = \begin{bmatrix} \boldsymbol{0} \\ \boldsymbol{v}_{J}^{\widehat{\EuScript{T}},\widehat{\EuScript{E}}_{k}} \\ \boldsymbol{0} \end{bmatrix}_{3\times3},\quad
  \boldsymbol{r}_{3,J}^{\widehat{\EuScript{T}},\widehat{\EuScript{E}}_{k}} = \begin{bmatrix} \boldsymbol{0} \\ \boldsymbol{0} \\\boldsymbol{v}_{J}^{\widehat{\EuScript{T}},\widehat{\EuScript{E}}_{k}}  \end{bmatrix}_{3\times3}.
\end{equation}
Similarly, we define $\boldsymbol{r}_{I,J}^{\widehat{\EuScript{T}},\widehat{\EuScript{F}}_{l}}$ and $\boldsymbol{r}_{I,J}^{\widehat{\EuScript{T}},\widehat{\EuScript{T}}}$ for $I=1,2,3$ using $\boldsymbol{v}_{J}^{\widehat{\EuScript{T}},\widehat{\EuScript{F}}_{l}}$ and $\boldsymbol{v}_{J}^{\widehat{\EuScript{T}},\widehat{\EuScript{T}}}$, respectively. The polynomial spaces $\mathcal{P}^{-}_{r}(\text{\large$\otimes$}^{2}T\widehat{\EuScript{T}})$ and $\mathcal{P}_{r}(\text{\large$\otimes$}^{2}T\widehat{\EuScript{T}})$ in \eqref{RFEs}$_2$ are spanned by a basis $\left\{ \boldsymbol{r}_{I,J}^{\widehat{\EuScript{T}},\widehat{\EuScript{E}}_{k}}, \boldsymbol{r}_{I,J}^{\widehat{\EuScript{T}},\widehat{\EuScript{F}}_{l}}, \boldsymbol{r}_{I,J}^{\widehat{\EuScript{T}},\widehat{\EuScript{T}}} \right\}$ that are respectively based on the shape functions of NE1 and NE2. Moreover, the following relations hold:
\begin{equation}\label{duality1}
\begin{aligned}
\phi^{\widehat{\EuScript{T}},\widehat{\EuScript{E}}_p}_{M,N}\Big(\boldsymbol{r}^{\widehat{\EuScript{T}},\widehat{\EuScript{E}}_k}_{I,J}\Big) &= \begin{cases}
1, & \text{ if } p=k \text{ and } I=M \text{ and } J=N,\\
0, & \text{otherwise},\end{cases} & \phi^{\widehat{\EuScript{T}},\widehat{\EuScript{E}}_p}_{M,N}\Big(\boldsymbol{r}^{\widehat{\EuScript{T}},\widehat{\EuScript{F}}_l}_{I,J}\Big)=\phi^{\widehat{\EuScript{T}},\widehat{\EuScript{E}}_p}_{M,N}\Big(\boldsymbol{r}^{\widehat{\EuScript{T}},\widehat{\EuScript{T}}}_{I,J}\Big)=0, \\
\phi^{\widehat{\EuScript{T}},\widehat{\EuScript{F}}_q}_{M,N}\Big(\boldsymbol{r}^{\widehat{\EuScript{T}},\widehat{\EuScript{F}}_l}_{I,J}\Big) &= \begin{cases}
1, & \text{ if } q=l \text{ and } I=M \text{ and } J=N,\\
0, & \text{ otherwise},\end{cases}  & \phi^{\widehat{\EuScript{T}},\widehat{\EuScript{F}}_q}_{M,N}\Big(\boldsymbol{r}^{\widehat{\EuScript{T}},\widehat{\EuScript{E}}_k}_{I,J}\Big)=\phi^{\widehat{\EuScript{T}},\widehat{\EuScript{F}}_q}_{M,N}\Big(\boldsymbol{r}^{\widehat{\EuScript{T}},\widehat{\EuScript{T}}}_{I,J}\Big)=0, \\
\phi^{\widehat{\EuScript{T}},\widehat{\EuScript{T}}}_{M,N}\Big(\boldsymbol{r}^{\widehat{\EuScript{T}},\widehat{\EuScript{T}}}_{I,J}\Big) &= \begin{cases}
1, & \text{if } I=M \text{ and } J=N,\\
0, & \text{ otherwise}, \end{cases}  &\phi^{\widehat{\EuScript{T}},\widehat{\EuScript{T}}}_{M,N}\Big(\boldsymbol{r}^{\widehat{\EuScript{T}},\widehat{\EuScript{E}}_k}_{I,J}\Big)=\phi^{\widehat{\EuScript{T}},\widehat{\EuScript{T}}}_{M,N}\Big(\boldsymbol{r}^{\widehat{\EuScript{T}},\widehat{\EuScript{F}}_l}_{I,J}\Big)=0.
\end{aligned}
\end{equation}
Later in this section, we will provide explicit expressions for some of the above shape functions that are used in our numerical examples. 
To interpolate $\boldsymbol{P}\in H^{\mathbf{d}}$, we define the two finite elements given in \eqref{RFEs}$_3$ based on respectively the N\'{e}d\'{e}lec $1^{\mathrm{st}}$-kind \emph{face} elements in $\mathbb{R}^3$ (NF1) \citep{Ned1980}, and the N\'{e}d\'{e}lec $2^{\mathrm{nd}}$-kind \emph{face} elements in $\mathbb{R}^3$ (NF2) \citep{Ned1986}. We denote the set of the local degrees of freedom $\text{\textit{\textSigma}}^{\widehat{\EuScript{T}},\mathbf{d}-}$ ($\text{\textit{\textSigma}}^{\widehat{\EuScript{T}},\mathbf{d}}$) by $\left\{ \psi^{\widehat{\EuScript{T}},\widehat{\EuScript{F}}_{l}}_{I,J}, \psi^{\widehat{\EuScript{T}},\widehat{\EuScript{T}}}_{I,J}\right\}$, where
\begin{equation}\label{DOF_hd}
\begin{aligned}
\psi^{\widehat{\EuScript{T}},\widehat{\EuScript{F}}_{l}}_{I,J}(\boldsymbol{T}) &=\int_{\widehat{\EuScript{F}}_l} f_J\langle \overrightarrow{\boldsymbol{T}}_{\!\!I},\hat{\boldsymbol{\mathsf{n}}}_l \rangle \,d\hat{A}, && \forall f_J\text{ that form a basis for } \mathcal{P}_{r-1}(\mathbb{R}^{3})\big|_{\widehat{\EuScript{F}}_l} \left(\mathcal{P}_{r}(\mathbb{R}^{3})\big|_{\widehat{\EuScript{F}}_l}\right),\\
\psi^{\widehat{\EuScript{T}},\widehat{\EuScript{T}}}_{I,J}(\boldsymbol{T}) &=\int_{\widehat{\EuScript{T}}}\langle \overrightarrow{\boldsymbol{T}}_{\!\!I} , \boldsymbol{Z}_{\!J} \rangle \,d\hat{V}, &&  \forall \boldsymbol{Z}_{\!\!J} \text{ that form a basis for }  \mathcal{P}_{r-2}(T\mathbb{R}^{3})\big|_{\widehat{\EuScript{T}}} \left(\mathcal{P}^{-}_{r-1}(T\mathbb{R}^{3})\big|_{\widehat{\EuScript{T}}}\right).
\end{aligned}{}
\end{equation}
We denote the set of the local shape functions of \eqref{RFEs}$_3$ by $\left\{\boldsymbol{s}_{I,J}^{\widehat{\EuScript{T}},\widehat{\EuScript{F}}_{l}}, \boldsymbol{s}_{I,J}^{\widehat{\EuScript{T}},\widehat{\EuScript{T}}} \right\}$. Both $\boldsymbol{s}_{I,J}^{\widehat{\EuScript{T}},\widehat{\EuScript{F}}_{l}}$ and $\boldsymbol{s}_{I,J}^{\widehat{\EuScript{T}},\widehat{\EuScript{T}}}$ are defined similar to (\ref{sfc}) but using the vector-valued shape functions of N\'{e}d\'{e}lec face elements, which we denote by $\boldsymbol{u}_{J}^{\widehat{\EuScript{T}},\widehat{\EuScript{F}}_{l}}$ and $\boldsymbol{u}_{J}^{\widehat{\EuScript{T}},\widehat{\EuScript{T}}}$. Also, one has
\begin{equation}\label{duality1}
\begin{aligned}
\psi^{\widehat{\EuScript{T}},\widehat{\EuScript{F}}_q}_{M,N}\Big(\boldsymbol{s}^{\widehat{\EuScript{T}},\widehat{\EuScript{F}}_l}_{I,J}\Big) &= \begin{cases}
1, & \text{ if } q=l \text{ and } I=M \text{ and } J=N,\\
0, & \text{ otherwise},\end{cases}  &\psi^{\widehat{\EuScript{T}},\widehat{\EuScript{F}}_q}_{M,N}\Big(\boldsymbol{s}^{\widehat{\EuScript{T}},\widehat{\EuScript{T}}}_{I,J}\Big)=0, \\
\psi^{\widehat{\EuScript{T}},\widehat{\EuScript{T}}}_{M,N}\Big(\boldsymbol{s}^{\widehat{\EuScript{T}},\widehat{\EuScript{T}}}_{I,J}\Big) &= \begin{cases}
1, & \text{if } I=M \text{ and } J=N,\\
0, & \text{ otherwise}, \end{cases}  &\psi^{\widehat{\EuScript{T}},\widehat{\EuScript{T}}}_{M,N}\Big(\boldsymbol{s}^{\widehat{\EuScript{T}},\widehat{\EuScript{F}}_l}_{I,J}\Big)=0.
\end{aligned}
\end{equation}
For the reference finite element of pressure \eqref{RFEs}$_4$, we have $\text{\textit{\textSigma}}^{\widehat{\EuScript{T}},\ell}=\{\omega^{\widehat{\EuScript{T}}}_{i}\}$, where $\omega^{\widehat{\EuScript{T}}}_{i}(f)=\frac{1}{\hat{A}}\int_{\widehat{\EuScript{T}}} p_{i} f\,d\hat{V}$ for all the polynomials $p_i$ that form a basis for $\mathcal{P}_{r}(\mathbb{R}^{3})\big|_{\widehat{\EuScript{T}}}$. Also, the set of local shape functions $\big\{t^{\widehat{\EuScript{T}}}_{i}\big\}$, which spans $\mathcal{P}_{r}(\widehat{\EuScript{T}})$, is $\{1\}$ for $r=0$, $\{1,\xi^1,\xi^2,\xi^3\}$ for $r=1$, and $\{1,\xi^1,\xi^2,(\xi^1)^2,(\xi^2)^2,(\xi^3)^2,\xi^1\xi^2,\xi^1\xi^3,\xi^2\xi^3\}$ for $r=2$. Choosing $p_i$ properly, one can show that $\omega^{\widehat{\EuScript{T}}}_{i}(t_{j}^{\widehat{\EuScript{T}}})=\delta_{ij}$.
\begin{figure}[t]
\begin{center}
\includegraphics[width = 0.9\textwidth]{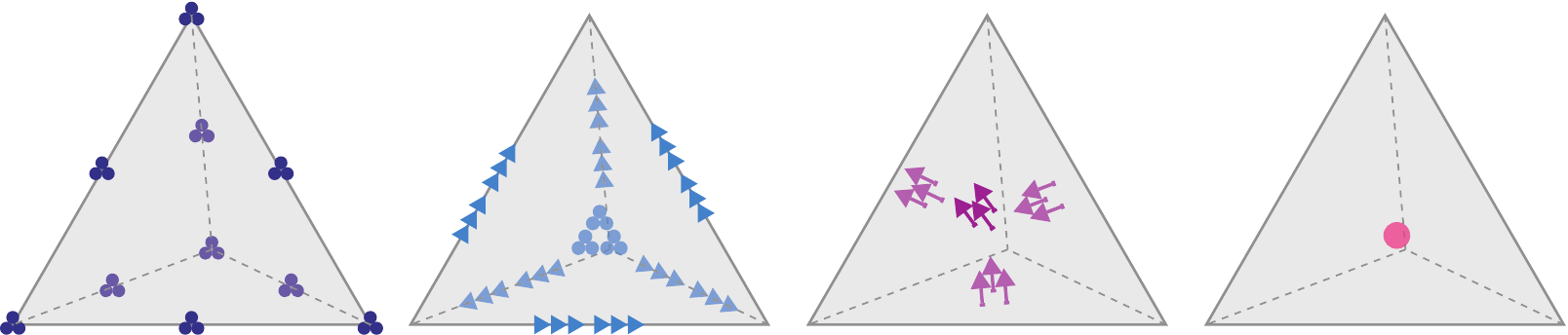}
\end{center}
 \caption{\footnotesize 
The schematic diagrams for the finite elements (\ref{RFEs_m}). The elements form left to right are for $\boldsymbol{U}$, $\boldsymbol{K}$, $\boldsymbol{P}$, and $p$. The total number of DOF is $88$.}
\label{CSFEMs3D}
\end{figure}

To extend 2D CSFEMs \citep{FaYa2018} to 3D, we first followed the same approach we had proposed in \citep{FaYa2018} and considered $r=1,2$ for the finite elements of $\boldsymbol{U}$, $\boldsymbol{K}$, and $\boldsymbol{P}$ and $r=0,1,2$ for the finite elements of $p$ in (\ref{RFEs}). This provides $96$ combinations of elements for discretizing the boundary-value problem (\ref{ElasMixF_1}). 
Using the matrix formulation of the linearization of (\ref{ElasMixF_1}) for $\alpha=0$ using the approach discussed in \citep[\S3.5]{FaYa2018}, we concluded that all the $96$ choices lead to strictly singular or unstable methods in 3D. We will discuss this further in \S\ref{solvability}.  
To overcome this singularity issue, we modify a suitable combination of elements among the aforementioned unstable $96$ choices and propose the following convergent finite elements for $\boldsymbol{U}$, $\boldsymbol{K}$,$\boldsymbol{P}$, and $p$:
\begin{equation} \label{RFEs_m}
\begin{aligned}
\left(\widehat{\EuScript{T}},\mathcal{P}_{2}(T\widehat{\EuScript{T}}),\text{\textit{\textSigma}}^{\widehat{\EuScript{T}},1}\right), \quad 
\left(\widehat{\EuScript{T}},\overline{\mathcal{P}}_{3}(\text{\large$\otimes$}^{2}T\widehat{\EuScript{T}}), \overline{\text{\textit{\textSigma}}}^{\widehat{\EuScript{T}},\mathbf{c}}\right), \quad 
\left(\widehat{\EuScript{T}},\mathcal{P}^{\ominus}_{1}(\text{\large$\otimes$}^{2}T\widehat{\EuScript{T}}), \text{\textit{\textSigma}}^{\widehat{\EuScript{T}},\mathbf{d}-}\right), \quad 
\left(\widehat{\EuScript{T}},\mathcal{P}_{0}(\widehat{\EuScript{T}}),\text{\textit{\textSigma}}^{\widehat{\EuScript{T}},\ell}\right),
\end{aligned}{}
\end{equation}
where $\overline{\mathcal{P}}_{3}(\text{\large$\otimes$}^{2}T\widehat{\EuScript{T}}) := \mathcal{P}_{1}(\text{\large$\otimes$}^{2}T\widehat{\EuScript{T}})\oplus\text{span}\left\{\boldsymbol{r}^{\widehat{\EuScript{T}},\widehat{\EuScript{T}}}_{I,J}\right\}_{I,J=1,2,3}$ for $\boldsymbol{r}^{\widehat{\EuScript{T}},\widehat{\EuScript{T}}}_{I,J} \in \mathcal{P}^{-}_{3}(\text{\large$\otimes$}^{2}T\widehat{\EuScript{T}})$, and $\overline{\text{\textit{\textSigma}}}^{\widehat{\EuScript{T}},\mathbf{c}}$ is the union of $\text{\textit{\textSigma}}^{\widehat{\EuScript{T}},\mathbf{c}}$ for $r=1$ and $\left\{\phi^{\widehat{\EuScript{T}},\widehat{\EuScript{T}}}_{I,J}\right\}_{I,J=1,2,3}\subset\text{\textit{\textSigma}}^{\widehat{\EuScript{T}},\mathbf{c}-}$ for $r=3$. The schematic diagram of (\ref{RFEs_m}) is illustrated in Figure \ref{CSFEMs3D}.
Moreover, the shape functions for the finite element of $\boldsymbol{U}$ in (\ref{RFEs_m}) are given by (\ref{sf1}) for $r=2$ and $m =10$, and the shape function for the finite element of $p$ in (\ref{RFEs_m}) is simply $t^{\widehat{\EuScript{T}}}=1$. We use the results of \citet{AFW2009} to provide the explicit expression for the shape functions for the finite elements of $\boldsymbol{K}$ and $\boldsymbol{P}$. 
The finite element of $\boldsymbol{K}$ in (\ref{RFEs_m}) has $6$ shape functions associated to each edge $\widehat{\EuScript{E}}_{k}$ of $\widehat{\EuScript{T}}$ and $9$ shape functions associated to $\widehat{\EuScript{T}}$. Let us ignore the superscript of $l^{1}_i, i=1,2,3,4$ in (\ref{lagFun3}) and consider $\nabla l_i= \begin{bmatrix} \partial l_i/\partial \xi^1  & \partial l_i/\partial \xi^2 & \partial l_i/\partial \xi^3 \end{bmatrix}$ as a row vector. Then, for an edge $\widehat{\EuScript{E}}_{k}$ joining two vertices $i$ and $j$ as shown in Figure \ref{RefElSch}, the $6$ shape functions $\boldsymbol{r}_{I,2}^{\widehat{\EuScript{T}},\widehat{\EuScript{E}}_{k}},\boldsymbol{r}_{I,3}^{\widehat{\EuScript{T}},\widehat{\EuScript{E}}_{k}}, I=1,2,3$ are obtained using (\ref{sfc}) and the following vector-valued shape functions for NE2 of order $1$ \citep{AFW2009}:
\begin{equation}\nonumber
  \boldsymbol{v}_{1}^{\widehat{\EuScript{T}},\widehat{\EuScript{E}}_{k}} = l_i\nabla l_j, \quad \boldsymbol{v}_{2}^{\widehat{\EuScript{T}},\widehat{\EuScript{E}}_{k}} = l_j\nabla l_i. 
\end{equation}
The $9$ remaining shape functions $\boldsymbol{r}^{\widehat{\EuScript{T}},\widehat{\EuScript{T}}}_{I,1},\boldsymbol{r}^{\widehat{\EuScript{T}},\widehat{\EuScript{T}}}_{I,2},\boldsymbol{r}^{\widehat{\EuScript{T}},\widehat{\EuScript{T}}}_{I,3},I=1,2,3$ are obtained similar to (\ref{sfc}) and using the following vector-valued shape functions for NE1 of order $3$ \citep{AFW2009}:
\begin{equation}\nonumber
  \boldsymbol{v}_{1}^{\widehat{\EuScript{T}},\widehat{\EuScript{T}}} = l_3 l_4 \boldsymbol{w}_{12}, \quad  \boldsymbol{v}_{2}^{\widehat{\EuScript{T}},\widehat{\EuScript{T}}} = l_2 l_4 \boldsymbol{w}_{13}, \quad \boldsymbol{v}_{3}^{\widehat{\EuScript{T}},\widehat{\EuScript{T}}} = l_2 l_3 \boldsymbol{w}_{14},
\end{equation}
 where $\boldsymbol{w}_{ij} = l_i\nabla l_j-l_j\nabla l_i$. The finite element of $\boldsymbol{P}$ in (\ref{RFEs_m}) has $3$ shape functions $\boldsymbol{s}_{I}^{\widehat{\EuScript{T}},\widehat{\EuScript{F}}_{l}}, I=1,2,3$ associated to each face $\widehat{\EuScript{F}}_{l}$ of $\widehat{\EuScript{T}}$ that contains the three vertices $i$, $j$, and $k$ according to Figure \ref{RefElSch}. These shape functions are obtained similar to (\ref{sfc}) and using the following vector-valued shape function for NF1 of order $1$\citep{AFW2009}:
\begin{equation}\nonumber
\boldsymbol{u}^{\widehat{\EuScript{T}},\widehat{\EuScript{F}}_{l}} = l_i\nabla l_j \times \nabla l_k - l_j\nabla l_i \times \nabla l_k+ l_k\nabla l_i \times \nabla l_j.
\end{equation}

\newcommand{\TT}{ \boldsymbol{\mathsf{T}} } 
Next we explain how to calculate the shape functions in an arbitrary element in a mesh using the shape functions of the reference finite elements. 
 Let $\mathcal{B}_{h}$ be a triangulation of the reference configuration $\mathcal{B}$ consisting of arbitrary tetrahedra $\EuScript{T}$ such that the intersection of any two distinct tetrahedra is either empty or a common face/edge/vertex of each. The discretization parameter $h$ is defined as $h:={\max}\,\mathrm{diam}\,\EuScript{T},\, \forall {\EuScript{T}\in\mathcal{B}_{h}}$.
We define a local ordering for vertices of each $\EuScript{T}\in\mathcal{B}_h$ by assigning the numbers $1,2,3,4$ to them. 
We then denote the Cartesian coordinates of the $i$-th vertex of $\EuScript{T}$ by a column vector $\boldsymbol{\mathsf{X}}^\EuScript{T}_i = [ \mathsf{X}^{1,\EuScript{T}}_{i} ~ \mathsf{X}^{2,\EuScript{T}}_{i} ~  \mathsf{X}^{3,\EuScript{T}}_{i} ]^\mathsf{T}$ and define the following affine transformation:
\begin{equation}\label{affineMap}
\TT_\EuScript{T}: \widehat{\EuScript{T}} \longrightarrow \EuScript{T},\quad
 \TT_\EuScript{T}(\boldsymbol{\xi}) :=
 \jac \boldsymbol{\xi} +\boldsymbol{\mathsf{X}}^\EuScript{T}_1,
\end{equation}
where $\jac  = [ \boldsymbol{\mathsf{X}}^\EuScript{T}_2-\boldsymbol{\mathsf{X}}^\EuScript{T}_1  ~ \boldsymbol{\mathsf{X}}^\EuScript{T}_3-\boldsymbol{\mathsf{X}}^\EuScript{T}_1 ~\boldsymbol{\mathsf{X}}^\EuScript{T}_4-\boldsymbol{\mathsf{X}}^\EuScript{T}_1 ]_{3\times3}$.
Note that $\TT_\EuScript{T}$ is bijective and $\jac$ is invertible.
For an element $\EuScript{T}\in\mathcal{B}_{h}$, we denote its edges by $\EuScript{E}^\EuScript{T}_i=\TT_\EuScript{T}(\widehat{\EuScript{E}}_i)$, $i=1,2,\dots,6$, and its faces by $\EuScript{F}^\EuScript{T}_i=\TT_\EuScript{T}(\widehat{\EuScript{F}}_i)$, $i=1,2,3,4$. We assume that $\EuScript{E}^\EuScript{T}_i$ inherits the orientation of its reference counterpart $\widehat{\EuScript{E}}_i$. Moreover, the tangent vector $\boldsymbol{\mathsf{t}}_i$ defined on $\EuScript{E}^\EuScript{T}_i$ inherits the orientation of $\EuScript{E}^\EuScript{T}_i$, and the normal vector on $\EuScript{F}^\EuScript{T}_l$ containing three edges $\EuScript{E}^\EuScript{T}_i$, $\EuScript{E}^\EuScript{T}_j$, and $\EuScript{E}^\EuScript{T}_k$ such that $i < j< k$ is defined as $\boldsymbol{\mathsf{n}}_l=\boldsymbol{\mathsf{t}}_i\times\boldsymbol{\mathsf{t}}_j$. One can show that $\boldsymbol{\mathsf{t}}_i = \jac \hat{\boldsymbol{\mathsf{t}}}_i$ and $\boldsymbol{\mathsf{n}}_i =\mathrm{det}\,\jac \jac^{\mathsf{-T}} \hat{\boldsymbol{\mathsf{n}}}_i$. For efficient assembly of the finite elements of $\boldsymbol{K}\in H^{\mathbf{c}}$ and $\boldsymbol{P}\in H^{\mathbf{d}}$, we use the numbering scheme discussed in \citep{RaKiLo2012}. Using this scheme, one first assumes that every vertex in a mesh $\mathcal{B}_{h}$ has a distinct global number and then the local ordering of four vertices of every tetrahedron in that mesh $\EuScript{T}\in\mathcal{B}_{h}$ agree with the \emph{ascending} order of the global numbers of its four vertices. Considering the edge orientations of the reference element shown in Figure \ref{RefElSch}, the orientation of every edge in the mesh is from a vertex with a smaller global number to a vertex with a larger global number. The advantage of this scheme is that the orientation of a common edge between elements in a mesh is uniquely defined and is identical to that of the edge in any of those elements. It follows that some elements in a mesh sharing a common edge have an identical tangent vector on that edge, and any two elements with a common face have an identical normal vector on that face. For an illustration of this, see \citep[Figue 5.2]{RaKiLo2012}. Note that using this scheme, the normal vectors of some of the exterior faces of the mesh are \emph{not} pointed outward, and $\st = \mathrm{sign}\left( \mathrm{det}\,\boldsymbol{\mathsf{J}}_{\EuScript{T}} \right)$ can be either $1$ or $-1$.

Consider the following mappings:
\begin{equation}\label{maps}
\begin{aligned}
&\TT^{1}_\EuScript{T}: C^0(T\widehat{\EuScript{T}}) \longrightarrow C^0(T\EuScript{T}),\quad
\TT^{1}_\EuScript{T}(\widehat{\boldsymbol{V}}) := \widehat{\boldsymbol{V}}\circ \TT^{-1}_\EuScript{T}, \\
&\TT^{\mathbf{c}}_\EuScript{T}: H^{\mathbf{c}}(T\widehat{\EuScript{T}}) \longrightarrow H^{\mathbf{c}}(T\EuScript{T}),\quad
\TT^{\mathbf{c}}_\EuScript{T}(\widehat{\boldsymbol{V}}) := \boldsymbol{\mathsf{J}}^{\mathsf{-T}}_{\EuScript{T}}\widehat{\boldsymbol{V}} \circ \TT^{-1}_\EuScript{T}, \\
&\TT^{\mathbf{d}}_\EuScript{T}: H^{\mathbf{d}}(T\widehat{\EuScript{T}}) \longrightarrow H^{\mathbf{d}}(T\EuScript{T}),\quad
\TT^{\mathbf{d}}_\EuScript{T}(\widehat{\boldsymbol{V}}) := \frac{1}{\mathrm{det}\,\boldsymbol{\mathsf{J}}_{\EuScript{T}}}\boldsymbol{\mathsf{J}}_{\EuScript{T}}\widehat{\boldsymbol{V}}\circ \TT^{-1}_\EuScript{T},\\
&\TT^{\ell}_\EuScript{T}: L^2(\widehat{\EuScript{T}}) \longrightarrow L^2(\EuScript{T}),\quad
\TT^{\ell}_\EuScript{T}(\hat{f}) := \hat{f}\circ \TT^{-1}_\EuScript{T},
\end{aligned}{}
\end{equation}
where $\TT^{\mathbf{c}}_\EuScript{T}$ and $\TT^{\mathbf{d}}_\EuScript{T}$ are known as the Piola transforms. For a $\binom{2}{0}$-tensor $\boldsymbol{T}$, one calculates the Piola transforms separately for each row:
\begin{equation}\nonumber
  \TT^{\mathbf{c}}_\EuScript{T}(\boldsymbol{T}) := \begin{bmatrix}  \TT^{\mathbf{c}}_\EuScript{T}(\overrightarrow{\boldsymbol{T}}_{\!\!1})^\mathsf{T} \\[4pt] \TT^{\mathbf{c}}_\EuScript{T}(\overrightarrow{\boldsymbol{T}}_{\!\!2})^\mathsf{T} \\[4pt] \TT^{\mathbf{c}}_\EuScript{T}(\overrightarrow{\boldsymbol{T}}_{\!\!3})^\mathsf{T} \end{bmatrix},
\end{equation}
and $\TT^{\mathbf{d}}_\EuScript{T}(\boldsymbol{T})$ is calculated similarly. 
Using \citep[Proposition 8]{FaYa2018}, (\ref{maps}), and the local shape functions in the reference element $\widehat{\EuScript{T}}$, one can obtain the local shape functions in any element $\EuScript{T}\in\mathcal{B}_h$ enabling one to locally interpolate the four field variables \mbox{$(\boldsymbol{U},\boldsymbol{K},\boldsymbol{P},p)$} over that element. In particular, the local shape functions for $\boldsymbol{U}$ are obtained as \mbox{$\boldsymbol{h}^{\EuScript{T}}_{k}=\TT^{1}_\EuScript{T}\left(\boldsymbol{h}^{\widehat{\EuScript{T}}}_{k}\right)$}; the local shape functions for $\boldsymbol{K}$ are $\boldsymbol{r}_{I,J}^{\EuScript{T},\EuScript{E}_{k}} = \TT^{\mathbf{c}}_{\EuScript{T}}\left( \boldsymbol{r}_{I,J}^{\widehat{\EuScript{T}},\widehat{\EuScript{E}}_{k}} \right)$, $\boldsymbol{r}_{I,J}^{\EuScript{T},\EuScript{F}_{l}} = \TT^{\mathbf{c}}_\EuScript{T}\left( \boldsymbol{r}_{I,J}^{\widehat{\EuScript{T}},\widehat{\EuScript{F}}_{l}}\right)$, and $\boldsymbol{r}_{I,J}^{\EuScript{T},\EuScript{T}} = \TT^{\mathbf{c}}_\EuScript{T}\left(\boldsymbol{r}_{I,J}^{\widehat{\EuScript{T}},\widehat{\EuScript{T}}}\right)$; the local shape functions for $\boldsymbol{P}$ are $\boldsymbol{s}_{I,J}^{\EuScript{T},\EuScript{F}_{l}} = \TT^{\mathbf{d}}_\EuScript{T}\left( \boldsymbol{s}_{I,J}^{\widehat{\EuScript{T}},\widehat{\EuScript{F}}_{l}} \right)$ and $\boldsymbol{s}_{I,J}^{\EuScript{T},{\EuScript{T}}} = \TT^{\mathbf{d}}_\EuScript{T}\left(\boldsymbol{s}_{I,J}^{\widehat{\EuScript{T}},\widehat{\EuScript{T}}} \right)$; and $t_i^{\EuScript{T}}=\TT^{\ell}_\EuScript{T}\left(t_i^{\widehat{\EuScript{T}}}\right)$ gives the local shape functions for $p$. Using \citep[Proposition 8]{FaYa2018} and (\ref{maps}), one can also obtain the local degrees of freedom for the finite elements of any element $\EuScript{T}\in\mathcal{B}_h$.  For example, considering (\ref{DOF_hd}), $\psi^{{\EuScript{T}},{\EuScript{F}}_{l}}_{I,J}(\boldsymbol{T})=\left(\psi^{\widehat{\EuScript{T}},\widehat{\EuScript{F}}_{l}}_{I,J} \circ {\TT^{\mathbf{d}}_\EuScript{T}}^{-1}\right)(\boldsymbol{T})$ and $\psi^{{\EuScript{T}},\EuScript{T}}_{I,J}(\boldsymbol{T})=\left(\psi^{\widehat{\EuScript{T}},\widehat{\EuScript{T}}}_{I,J}\circ {\TT^{\mathbf{d}}_\EuScript{T}}^{-1}\right)(\boldsymbol{T})$ are the degrees of freedom for the finite element of $\EuScript{T}$ that we use for $\boldsymbol{P}$. The other degrees of freedom can be written similarly using their reference counterparts. In this work, the traction boundary conditions are imposed weakly through (\ref{traction}); one does not need to impose them directly by calculating the related degrees of freedom on the boundary of the mesh. Thus, in practice, all degrees of freedom, even those on the boundary of the mesh, are obtained by solving the final discrete system; calculating their explicit expressions are not required.

\subsection{Finite Element Spaces}\label{FES} 
\newcommand{\sif}{ \EuScript{F}_h^i }                         
\newcommand{\sie}{ \EuScript{E}_h^i }  

Next, some conforming finite element spaces are introduced in order to discretize (\ref{ElasMixF_1}) and (\ref{ElasMixF_2}). Let $\sif$ be the set of all \emph{interior} faces of a 3D mesh $\mathcal{B}_h$. Given a face $\EuScript{F}\in\sif$, there are two elements $\EuScript{T},\EuScript{T}'\in\mathcal{B}_{h}$ such that $\EuScript{F} = \EuScript{T}\cap\EuScript{T}^\prime$. Suppose $\boldsymbol{V}$ is a vector-valued function and $\boldsymbol{T}$ is a tensor-valued function both defined on $\mathcal{B}_h$ with limits on both sides of $\forall\EuScript{F}\in\sif$. We define the following notions of jump across a face $\EuScript{F}\in\sif$:
\begin{equation}\label{jumps}
  \llbracket\boldsymbol{V}\rrbracket_{\EuScript{F}} 
  := \boldsymbol{V}_{\EuScript{T}^\prime}-\boldsymbol{V}_{\EuScript{T}}, \quad
  \llbracket \mathsf{t}\boldsymbol{T}\rrbracket_{\EuScript{F}} 
  := \left(\boldsymbol{T}_{\EuScript{T}^\prime}-\boldsymbol{T}_{\EuScript{T}} \right) \boldsymbol{\mathsf{t}},  \quad
  \llbracket \mathsf{n}\boldsymbol{T}\rrbracket_{\EuScript{F}} 
  := \left(\boldsymbol{T}_{\EuScript{T}^\prime}-\boldsymbol{T}_{\EuScript{T}} \right) \boldsymbol{\mathsf{n}},
\end{equation}
where $\boldsymbol{V}_{\EuScript{T}} := \boldsymbol{V}|_{\EuScript{T}}$, $\boldsymbol{T}_{\EuScript{T}} := \boldsymbol{T}|_{\EuScript{T}}$, and $\boldsymbol{V}_{\EuScript{T^\prime}}$ and $\boldsymbol{T}_{\EuScript{T}^\prime}$ are defined similarly. $\boldsymbol{\mathsf{t}}$ ($\boldsymbol{\mathsf{n}}$) is a unit vector tangent (normal) to $\EuScript{F}$. We write $\llbracket \mathsf{t}\boldsymbol{T}\rrbracket_{\EuScript{F}}=\boldsymbol{0}$ ($\llbracket \mathsf{n}\boldsymbol{T}\rrbracket_{\EuScript{F}}=\boldsymbol{0}$), if the jump is zero for any unit vector $\boldsymbol{\mathsf{t}}$ ($\boldsymbol{\mathsf{n}}$) on $\EuScript{F}$. Note that all the above jumps are vector-valued functions in 3D. Consider the following finite element spaces:
\begin{equation}\nonumber
\begin{aligned}
&V_{h,r}^{1} := \left\{\boldsymbol{V}_h \in L^2(T\mathcal{B}_h) : \forall\EuScript{T}\in\mathcal{B}_h,~ \boldsymbol{V}_h|_{\EuScript{T}}\in\mathcal{P}_{r}(T\EuScript{T}),~ \forall\EuScript{F}\in \sif,~\llbracket\boldsymbol{V}_h\rrbracket_{\EuScript{F}}=\boldsymbol{0}\right\}, \\ 
&V_{h,r}^{\mathbf{c}{-}} := \left\{\boldsymbol{T}_h \in L^2(\otimes^2 T\mathcal{B}_h): \forall\EuScript{T}\in\mathcal{B}_h,~ \boldsymbol{T}_h|_{\EuScript{T}} \in \mathcal{P}^{-}_{r}(\text{\large$\otimes$}^{2}T\EuScript{T}),~ \forall\EuScript{F}\in \sif,~\llbracket\mathsf{t}\boldsymbol{T}_h\rrbracket_{\EuScript{F}}=\boldsymbol{0}\right\},\\
&V_{h,r}^{\mathbf{c}} := \left\{\boldsymbol{T}_h \in L^2(\otimes^2 T\mathcal{B}_h): \forall\EuScript{T}\in\mathcal{B}_h,~ \boldsymbol{T}_h|_{\EuScript{T}} \in \mathcal{P}_{r}(\text{\large$\otimes$}^{2}T\EuScript{T}),~ \forall\EuScript{F}\in \sif,~\llbracket\mathsf{t}\boldsymbol{T}_h\rrbracket_{\EuScript{F}}=\boldsymbol{0}\right\},\\
&V_{h,r}^{\mathbf{d}{-}} := \left\{\boldsymbol{T}_h \in L^2(\otimes^2 T\mathcal{B}_h) : \forall\EuScript{T}\in\mathcal{B}_h,~ \boldsymbol{T}_h|_{\EuScript{T}} \in \mathcal{P}^{\ominus}_{r}(\text{\large$\otimes$}^{2}T\EuScript{T}),~ \forall\EuScript{F}\in \sif,~\llbracket\mathsf{n}\boldsymbol{T}_h\rrbracket_{\EuScript{F}}=\boldsymbol{0}\right\},\\
&V_{h,r}^{\mathbf{d}} := \left\{\boldsymbol{T}_h \in L^2(\otimes^2 T\mathcal{B}_h) : \forall\EuScript{T}\in\mathcal{B}_h,~ \boldsymbol{T}_h|_{\EuScript{T}} \in \mathcal{P}_{r}(\text{\large$\otimes$}^{2}T\EuScript{T}),~ \forall\EuScript{F}\in \sif,~\llbracket\mathsf{n}\boldsymbol{T}_h\rrbracket_{\EuScript{F}}=\boldsymbol{0}\right\},\\
&V_{h,r}^{\ell} := \left\{f_h \in L^2(\mathcal{B}_h) : \forall\EuScript{T}\in\mathcal{B}_h,~ f_{h}|_{\EuScript{T}}\in\mathcal{P}_{r}(\EuScript{T})\right\}.
\end{aligned}
\end{equation}
Note that the above sapces are conforming, i.e.,  $V_{h,r}^{1} \subset H^{1}(T\mathcal{B}_h)$, \mbox{$V_{h,r}^{\mathbf{c}{-}} \subset V_{h,r}^{\mathbf{c}}\subset H^{\mathbf{c}}(\mathcal{B}_h)$}, \mbox{$V_{h,r}^{\mathbf{d}{-}} \subset V_{h,r}^{\mathbf{d}}\subset H^{\mathbf{d}}(\mathcal{B}_h)$}, and $V_{h,r}^{\ell}\subset L^2(\mathcal{B}_h)$. Recalling the definition of $\overline{\mathcal{P}}_{3}(\text{\large$\otimes$}^{2}T\EuScript{T})$ in (\ref{RFEs_m}), we define 
\begin{equation}\label{Vc3}
	\overline{V}_{h,3}^{\mathbf{c}} := V_{h,1}^{\mathbf{c}} \oplus\left\{\boldsymbol{T}_h \in L^2(\otimes^2 T\mathcal{B}_h): \forall\EuScript{T}\in\mathcal{B}_h,~ \boldsymbol{T}_h|_{\EuScript{T}}\in\text{span}\left\{\boldsymbol{r}^{\EuScript{T},\EuScript{T}}_{I,J}\right\}_{I,J=1,2,3}\subset\mathcal{P}^{-}_{3}(\text{\large$\otimes$}^{2}T\EuScript{T}) \right\}.
\end{equation}
Note that $\big(\boldsymbol{r}^{\EuScript{T},\EuScript{T}}_{I,J}\boldsymbol{\mathsf{t}}\big)\big|_{\EuScript{F}}=\boldsymbol{0}$ for $I,J=1,2,3$, and for every $\EuScript{F}$ in the mesh and $\overline{V}_{h,3}^{\mathbf{c}}\subset H^{\mathbf{c}}(\mathcal{B}_h)$.  

\subsection{ The Compatible-Strain Mixed Finite Element Methods}\label{Sec_CSFEM} 

We write the following mixed finite element method for the boundary-value problem of \emph{incompressible} nonlinear elastostatics (\ref{ElasMixF_1}) based on the reference elements (\ref{RFEs_m}) and the corresponding approximation spaces \mbox{($V_{h,2}^{1}$, $\overline{V}_{h,3}^{\mathbf{c}}$, $V_{h,1}^{\mathbf{d}-}$, $V_{h,0}^{\ell}$)} defined in the previous section:

\bigskip
\begin{minipage}{0.95\textwidth}{\textit{Given a body force $\boldsymbol{B}$ of $L^{2}$-class, a boundary displacement $\overline{\boldsymbol{U}}$ on $\Gamma_{d}$ of $H^{1/2}$-class, a boundary traction $\overline{\boldsymbol{T}}$ on $\Gamma_{t}$ of $L^{2}$-class, and a stability constant $\alpha\geq0$,  find $(\boldsymbol{U}_h,\boldsymbol{K}_h,\boldsymbol{P}_h,p_h)\in V_{h,2}^{1}(\Gamma_{d},\overline{\boldsymbol{U}})\times \overline{V}_{h,3}^{\mathbf{c}} \times V_{h,1}^{\mathbf{d}-}\times V_{h,0}^{\ell}$ such that}}
\begin{equation}\label{FEh}
\begin{alignedat}{3}
\llangle\boldsymbol{P}_h,\mathbf{grad}\,\boldsymbol{\Upsilon}_h\rrangle+\alpha s_{h_1}\!\left(\boldsymbol{U}_h,\boldsymbol{K}_h,\boldsymbol{\Upsilon}_h\right) &= f_h(\boldsymbol{\Upsilon}_h), &\quad &\forall \boldsymbol{\Upsilon}_h\in V_{h,2}^{1}(\Gamma_{d}),\\
\llangle \widetilde{\boldsymbol{P}}(\boldsymbol{K}_h),\boldsymbol{\kappa}_h\rrangle - \llangle\boldsymbol{P}_h,\boldsymbol{\kappa}_h\rrangle + \llangle p_h\boldsymbol{Q}(\boldsymbol{K}_h),\boldsymbol{\kappa}_h\rrangle+\alpha s_{h_2}\!\left(\boldsymbol{U}_h,\boldsymbol{K}_h,\boldsymbol{\kappa}_h\right)&= 0,& &\forall \boldsymbol{\kappa}_h\in \overline{V}_{h,3}^{\mathbf{c}},\\
\llangle \mathbf{grad}\,\boldsymbol{U}_h, \boldsymbol{\pi}_h \rrangle - \llangle \boldsymbol{K}_h,\boldsymbol{\pi}_h\rrangle &= 0, & &\forall \boldsymbol{\pi}_h\in V_{h,1}^{\mathbf{d}-},\\
\llangle C(J_h), q_h\rrangle &= 0, & &\forall q_h\in V_{h,0}^{\ell},
\end{alignedat}
\end{equation}
\textit{where} 
\begin{equation}\nonumber
 f_h(\boldsymbol{\Upsilon}_h) = \llangle \rho_{0}\boldsymbol{B},\boldsymbol{\Upsilon}_h \rrangle + \int_{\Gamma_{t}}\langle \overline{\boldsymbol{T}},\boldsymbol{\Upsilon}_h\rangle\,dA,
\end{equation}
\textit{and}
\begin{equation}\nonumber
\begin{alignedat}{3}
s_{h_1}\!\left(\boldsymbol{U}_h,\boldsymbol{K}_h,\boldsymbol{\Upsilon}_h\right) &= \llangle\mathbf{grad}\,\boldsymbol{U}_h,\mathbf{grad}\,\boldsymbol{\Upsilon}_h\rrangle-\llangle\boldsymbol{K}_h,\mathbf{grad}\,\boldsymbol{\Upsilon}_h\rrangle,\\
s_{h_2}\!\left(\boldsymbol{U}_h,\boldsymbol{K}_h,\boldsymbol{\kappa}_h\right) &= \llangle\boldsymbol{K}_h,\boldsymbol{\kappa}_h\rrangle-\llangle\mathbf{grad}\,\boldsymbol{U}_h,\boldsymbol{\kappa}_h\rrangle.
\end{alignedat}
\end{equation}
\end{minipage}
\bigskip

\noindent Similarly, one can define the following mixed finite element method for the boundary-value problem of \mbox{\emph{compressible}} nonlinear elastostatics (\ref{ElasMixF_2}):

\bigskip
\begin{minipage}{0.95\textwidth}{\textit{Given ($\boldsymbol{B},\overline{\boldsymbol{U}},{\boldsymbol{T}}$) and $\alpha\geq0$, find $(\boldsymbol{U}_h,\boldsymbol{K}_h,\boldsymbol{P}_h)\in V_{h,2}^{1}(\Gamma_{d},\overline{\boldsymbol{U}})\times \overline{V}_{h,3}^{\mathbf{c}} \times V_{h,1}^{\mathbf{d}-}$ such that}}
\begin{equation}\label{FEh_2}
\begin{alignedat}{3}
\llangle\boldsymbol{P}_h,\mathbf{grad}\,\boldsymbol{\Upsilon}_h\rrangle + \alpha s_{h_1}\!\left(\boldsymbol{U}_h,\boldsymbol{K}_h,\boldsymbol{\Upsilon}_h\right)&=f_h(\boldsymbol{\Upsilon}_h), &\quad &\forall \boldsymbol{\Upsilon}_h\in V_{h,2}^{1}(\Gamma_{d}),\\
\llangle \widehat{\boldsymbol{P}}(\boldsymbol{K}_h),\boldsymbol{\kappa}_h\rrangle - \llangle\boldsymbol{P}_h,\boldsymbol{\kappa}_h\rrangle + \alpha s_{h_2}\!\left(\boldsymbol{U}_h,\boldsymbol{K}_h,\boldsymbol{\kappa}_h\right)&= 0,& &\forall \boldsymbol{\kappa}_h\in \overline{V}_{h,3}^{\mathbf{c}},\\
\llangle \mathbf{grad}\,\boldsymbol{U}_h, \boldsymbol{\pi}_h \rrangle - \llangle \boldsymbol{K}_h,\boldsymbol{\pi}_h\rrangle &= 0, & &\forall \boldsymbol{\pi}_h\in V_{h,1}^{\mathbf{d}-}.
\end{alignedat}
\end{equation}
\end{minipage}
\bigskip
\\

\noindent The above mixed finite element methods are extensions of the \emph{compatible-strain mixed finite element methods (\mbox{CSFEMs)}} introduced in \citep{FaYa2018} and \citep{AnFSYa2017} to three dimensional problems.  

\begin{rem}[Compatibility of Strain and Continuity of Traction]\label{feature}\leavevmode
\begin{enumerate}[$(i)$]
\item Recalling (\ref{jumps}) and considering a displacement gradient field $\boldsymbol{K}_h$ on a 3D mesh $\mathcal{B}_{h}$, the Hadamard jump condition is defined as the zero jump $\llbracket\mathsf{t}\boldsymbol{K}_h\rrbracket_{\EuScript{F}}=\boldsymbol{0}$ for any tangent vector on ${\EuScript{F}}$ and the three edges enclosing ${\EuScript{F}}$. A necessary condition for the existence of $\boldsymbol{U}_h\in  H^{1}(T\mathcal{B}_h)$ such that $\boldsymbol{K}_h = \mathbf{grad}\,\boldsymbol{U}_h$ is that the Hadamard jump condition holds $\forall\EuScript{F}\in\EuScript{F}_h^i$ \citep{AngoshtariYavari2016}. By construction, the mixed finite element methods (\ref{FEh}) and (\ref{FEh_2}) satisfy this necessary condition as $\boldsymbol{K}_h\in\overline{V}_{h,3}^{\mathbf{c}}\subset H^{\mathbf{c}}(\mathcal{B}_h)$.

\item Let $\boldsymbol{P}_h$ be a stress field on a 3D mesh $\mathcal{B}_{h}$. The localization of the balance of linear momentum requires that $\llbracket\mathsf{n}\boldsymbol{P}_h\rrbracket_{\EuScript{F}}=\boldsymbol{0}$, $\forall\EuScript{F}\in\EuScript{F}_h^i$, that is, the traction vector associated with $\boldsymbol{P}_h$ is continuous across all the internal faces of $\mathcal{B}_{h}$. By construction, (\ref{FEh}) and (\ref{FEh_2}) satisfy this requirement as $\boldsymbol{P}_h\in V_{h,1}^{\mathbf{d}-}\subset H^{\mathbf{d}}(\mathcal{B}_h)$. 
\end{enumerate}
\end{rem}

\newcommand{\ivek}[1] {[#1]}
\newcommand{\bigiVek}[2]{#1[ #2 #1]} 
\subsection{Matrix Formulation of CSFEMs}

The procedure of obtaining the matrix formulation of (\ref{FEh}) or (\ref{FEh_2}) is similar to that of 2D CSFEMs, which we discussed in detail in \cite[\S 3.4]{FaYa2018}. In this section, we only write the final formulations needed for the implementation and studying the stability of the 3D CSFEMs. One can write  (\ref{FEh}) in the following matrix form
\begin{equation}\label{FEd}
  \mathbb{K}_h\mathbb{Q}_h+\mathbb{N}_h(\mathbb{Q}_h) = \mathbb{F}_h,
\end{equation}
where
\begin{equation}\nonumber
\mathbb{K}_h =
\begin{bmatrix}
\begin{array}{llll}
  \alpha\gm{M}{11}            &\alpha\gm{M}{1\mathbf{c}}          &\gm{K}{1\mathbf{d}}           &\bz\\
  \alpha\gm{M}{1\mathbf{c}}   &\alpha\gm{M}{\mathbf{c}\mathbf{c}} &\gm{K}{\mathbf{c}\mathbf{d}}  &\bz \\
  ~~\gm{K}{\mathbf{d}1}       &~~\gm{K}{\mathbf{d}\mathbf{c}}     &\bz                           &\bz \\
  ~~\bz                       &~~\bz          					  &\bz                           &\bz  
\end{array}
\end{bmatrix}
\quad
\mathbb{Q}_h =
\begin{bmatrix} \gm{q}{1} \\  \gm{q}{\mathbf{c}} \\ \gm{q}{\mathbf{d}} \\ \gm{q}{\ell}  \end{bmatrix},
\quad
\mathbb{N}_h(\mathbb{Q}_h) =
\begin{bmatrix} \bz \\  \gm{N}{\mathbf{c}}(\gm{q}{\mathbf{c}},\gm{q}{\ell}) \\ \bz \\ \gm{N}{\ell}(\gm{q}{\mathbf{c}}) \end{bmatrix},
\quad
\mathbb{F}_h =
\begin{bmatrix}
\gm{F}{1}+\boldsymbol{\mathsf{F}}^{1}_{\Gamma_{t}}
\\ \bz \\ \bz \\  \bz
\end{bmatrix}.
\end{equation}
The column vectors $\gm{q}{1}$, $\gm{q}{\mathbf{c}}$, $\gm{q}{\mathbf{d}}$, $\gm{q}{\ell}$ contain all the unknown global degrees of freedom for $\boldsymbol{U}$, $\boldsymbol{K}$, $\boldsymbol{P}$, and $p$, respectively. Let $n$ be the total number of nodes in $\mathcal{B}_h$ except those lying on the displacement boundary $\Gamma_d$, and let $n_\EuScript{E}$, $n_\EuScript{F}$, and $n_\EuScript{T}$ be the total numbers of edges, faces, and elements in $\mathcal{B}_h$, respectively.
The number of degrees of freedom in $\gm{q}{1}$, $\gm{q}{\mathbf{c}}$, $\gm{q}{\mathbf{d}}$, and $\gm{q}{\ell}$ is $3n$, $6n_\EuScript{E}+9n_\EuScript{T}$, $3n_\EuScript{F}$, and $n_\EuScript{T}$, respectively, see Figure \ref{CSFEMs3D}. The total number of degrees of freedom is $N = 3n+6n_\EuScript{E}+3n_\EuScript{F}+10n_\EuScript{T}$. The size of $\mathbb{K}_h$ is $N \times N$, and the size of $\mathbb{Q}_h$, $\mathbb{N}_h$, and $\mathbb{F}_h$ is $N \times 1$. Let us define $\boldsymbol{V}_{\EuScript{T}}:=\boldsymbol{V}_h|_{\EuScript{T}}$ and $\boldsymbol{T}_{\EuScript{T}}:=\boldsymbol{T}_h|_{\EuScript{T}}$ for any discrete vector field $\boldsymbol{V}_h$ and any discrete tensor field $\boldsymbol{T}_h$. 
The global \emph{sparse} matrices $\gm{M}{11}$, $\gm{M}{1\mathbf{c}}$, $\gm{M}{\mathbf{c}\mathbf{c}}$, $\gm{K}{1\mathbf{d}}$, and $\gm{K}{\mathbf{c}\mathbf{d}}$ in $\mathbb{K}_h$ are the result of assembling a set of $n_\EuScript{T}$ local matrices that are obtained from calculating respectively $\llangle\mathbf{grad}\,\boldsymbol{U}_{\EuScript{T}},\mathbf{grad}\,\boldsymbol{\Upsilon}_{\EuScript{T}}\rrangle$, 
$-\llangle\boldsymbol{K}_{\EuScript{T}},\mathbf{grad}\,\boldsymbol{\Upsilon}_{\EuScript{T}}\rrangle$,
$\llangle\boldsymbol{K}_{\EuScript{T}},\boldsymbol{\kappa}_{\EuScript{T}}\rrangle$,
$\llangle\boldsymbol{P}_{\EuScript{T}},\mathbf{grad}\,\boldsymbol{\Upsilon}_{\EuScript{T}}\rrangle$, and 
$-\llangle\boldsymbol{P}_{\EuScript{T}},\boldsymbol{\kappa}_{\EuScript{T}}\rrangle$, $\forall\EuScript{T}\in\mathcal{B}_h$.
Moreover, $\mathbb{K}_h$ is a symmetric matrix and $\gm{M}{\mathbf{c}1} = (\gm{M}{1\mathbf{c}} )^{\mathsf{T}}$, $\gm{K}{\mathbf{d}1} = (\gm{K}{1\mathbf{d}})^{\mathsf{T}}$, and $\gm{K}{\mathbf{d}\mathbf{c}} = (\gm{K}{\mathbf{c}\mathbf{d}})^{\mathsf{T}}$. 
For given $\gm{q}{\mathbf{c}}$ and $\gm{q}{\ell}$, one obtains the global vectors $\gm{N}{\mathbf{c}}(\gm{q}{\mathbf{c}},\gm{q}{\ell})$ and $\gm{N}{\ell}(\gm{q}{\mathbf{c}})$ in $\mathbb{N}_h$ by assembling a set of $n_\EuScript{T}$ local vectors that are obtained form calculating the nonlinear terms $\llangle \widetilde{\boldsymbol{P}}(\boldsymbol{K}_{\EuScript{T}})+p_{\EuScript{T}}\boldsymbol{Q}(\boldsymbol{K}_{\EuScript{T}}),\boldsymbol{\kappa}_{\EuScript{T}}\rrangle$ and $\llangle C(J_{\EuScript{T}}), q_{\EuScript{T}}\rrangle$, $\forall\EuScript{T}\in\mathcal{B}_h$, respectively. 
Similarly, for a given body force $\boldsymbol{B}$, one obtains $\gm{F}{1}$ in $\mathbb{F}_h$ by calculating $\llangle \rho_{0}\boldsymbol{B},\boldsymbol{\Upsilon}_{\EuScript{T}} \rrangle$, $\forall\EuScript{T}\in\mathcal{B}_h$. 
Finally, for a given traction $\overline{\boldsymbol{T}}$ on $\Gamma_t$, one obtains $\boldsymbol{\mathsf{F}}^{1}_{\Gamma_{t}}$ in $\mathbb{F}_h$ through assembling all the local vectors obtained from $\int_{\EuScript{F}^\EuScript{T}_t}\langle \overline{\boldsymbol{T}},\boldsymbol{\Upsilon}|_{\EuScript{F}^\EuScript{T}_t}\rangle \,dA$ for every face $\EuScript{F}^\EuScript{T}_t$ lying on $\Gamma_t$.  
See \cite[(3.21)-(3.33)]{FaYa2018} for details of calculating the local matrices and vectors in each element. 
To obtain the matrix form of (\ref{FEh_2}) for compressible solids, we modify (\ref{FEd}) by setting $p_h=0$ ($\gm{q}{\ell} = \boldsymbol{0}$) and removing the fourth row and the fourth column of $\mathbb{K}_h$ and the forth entries of $\mathbb{Q}_h$, $\mathbb{N}_h$, and $\mathbb{F}_h$. We also use $\widehat{\boldsymbol{P}}(\boldsymbol{K})$ instead of $\widetilde{\boldsymbol{P}}(\boldsymbol{K})$ in our calculations.

Using Newton's method, one can approximate the solution of the nonlinear equation (\ref{FEd}) iteratively using $\mathbb{Q}^{(i+1)}_h = \mathbb{Q}^{(i)}_h-\mathbb{K}^{-1}_{t_h}\left(\mathbb{Q}^{(i)}_h\right) \mathbb{R}_h\left(\mathbb{Q}^{(i)}_h\right)$, where $\mathbb{R}_h(\mathbb{Q}_h) = \mathbb{K}_h\mathbb{Q}_h +\mathbb{N}_h(\mathbb{Q}_h) - \mathbb{F}_h$ is the residual vector and $\mathbb{K}_{{t}_h}$ is the tangent stiffness matrix (Jacobian matrix) given by
\begin{equation}\label{tan_icmp}
\mathbb{K}_{{t}_h}(\gm{q}{\mathbf{c}},\gm{q}{\ell}) =
\begin{bmatrix}
\begin{array}{llll}
  \alpha\gm{M}{11}            &\alpha\gm{M}{1\mathbf{c}}                                        &\gm{K}{1\mathbf{d}}           &\bz\\
  \alpha\gm{M}{1\mathbf{c}}   &\alpha\gm{M}{\mathbf{c}\mathbf{c}} \! + \! \gm{H}{\mathbf{c}\mathbf{c}}(\gm{q}{\mathbf{c}},\gm{q}{\ell})  &\gm{K}{\mathbf{c}\mathbf{d}}  &\gm{H}{\mathbf{c}\ell}(\gm{q}{\mathbf{c}}) \\
  ~~\gm{K}{\mathbf{d}1}       &~~\gm{K}{\mathbf{d}\mathbf{c}}                                   &\bz                           &\bz \\
  ~~\bz                       &~~\gm{H}{\ell\mathbf{c}}(\gm{q}{\mathbf{c}})                                         &\bz                           &\bz 
\end{array}
\end{bmatrix}.
\end{equation}
The matrix $\gm{H}{\mathbf{c}\mathbf{c}}(\gm{q}{\mathbf{c}},\gm{q}{\ell})$ $\big(\gm{H}{\mathbf{c}\ell}(\gm{q}{\mathbf{c}})\big)$ contains the derivative of components of $\gm{N}{\mathbf{c}}(\gm{q}{\mathbf{c}},\gm{q}{\ell})$ in $\mathbb{N}_h$ with respect to components of $\gm{q}{\mathbf{c}}$ ($\gm{q}{\ell}$). Also, $\gm{H}{\ell\mathbf{c}}(\gm{q}{\mathbf{c}})$ contains the derivative of components of $\gm{N}{\ell}(\gm{q}{\mathbf{c}})$  in $\mathbb{N}_h$  with respect to components of $\gm{q}{\mathbf{c}}$. 
Linearizing $\llangle \widetilde{\boldsymbol{P}}(\boldsymbol{K})+p\boldsymbol{Q}(\boldsymbol{K}),\boldsymbol{\kappa}\rrangle$ in \eqref{ElasMixF_1}$_2$ at a given displacement gradient $\boldsymbol{K}^0 \in H^{\mathbf{c}}(\mathcal{B})$ and a given pressure $p^0 \in L^{2}(\mathcal{B})$ gives \mbox{$\llangle \boldsymbol{\widetilde{A}}(\boldsymbol{K}^0,p^0) \!:\! \boldsymbol{K},\boldsymbol{\kappa}\rrangle+\llangle p\boldsymbol{Q}(\boldsymbol{K}^0),\boldsymbol{\kappa}\rrangle$}, where  $ \boldsymbol{\widetilde{A}}$ is the elasticity tensor and $( \boldsymbol{\widetilde{A}} \!:\! \boldsymbol{K})^{IJ} = \widetilde{A}^{IJMN} K^{MN}$. Also, linearizing $\llangle C(J(\boldsymbol{K})), q\rrangle$  in \eqref{ElasMixF_1}$_4$ at $\boldsymbol{K}^0$ results in $\llangle \boldsymbol{Q}(\boldsymbol{K}^0)\!:\!\boldsymbol{K}, q\rrangle=\llangle \boldsymbol{K}, q\boldsymbol{Q}(\boldsymbol{K}^0)\rrangle$, where $\boldsymbol{Q}\!:\!\boldsymbol{K} ={Q}^{IJ}{K}^{IJ}$. 
After discretization, for given $\boldsymbol{K}^0_h$ and $p^0_h$ (or $\gm{q}{\mathbf{c}}$ and $\gm{q}{\ell}$), one can calculate the local matrices for $\llangle \boldsymbol{\widetilde{A}}_{\EuScript{T}}(\boldsymbol{K}^0_{\EuScript{T}},p^0_{\EuScript{T}}) \!:\! \boldsymbol{K}_{\EuScript{T}},\boldsymbol{\kappa}_{\EuScript{T}}\rrangle$, $\llangle p_{\EuScript{T}}\boldsymbol{Q}(\boldsymbol{K}^0_{\EuScript{T}}),\boldsymbol{\kappa}_{\EuScript{T}}\rrangle$, and $\llangle \boldsymbol{K}_{\EuScript{T}}, q_{\EuScript{T}}\boldsymbol{Q}(\boldsymbol{K}^0_{\EuScript{T}})\rrangle$ $\forall\EuScript{T}\in\mathcal{B}_h$ to assemble $\gm{H}{\mathbf{c}\mathbf{c}}$, $\gm{H}{\mathbf{c}\ell}$, and $\gm{H}{\ell\mathbf{c}}$, respectively. For more details, see \cite[(3.36)]{FaYa2018}, and note that $\gm{H}{\mathbf{c}\ell} = (\gm{H}{\ell\mathbf{c}})^{\mathsf{T}}$. 
For compressible solids, (\ref{tan_icmp}) simplifies to 
\begin{equation}\label{tan_comp}
\mathbb{K}_{{t}_h}(\gm{q}{\mathbf{c}}) =
\begin{bmatrix}
\begin{array}{lll}
  \alpha\gm{M}{11}            &\alpha\gm{M}{1\mathbf{c}}                                        							&\gm{K}{1\mathbf{d}}\\          
  \alpha\gm{M}{1\mathbf{c}}   &\alpha\gm{M}{\mathbf{c}\mathbf{c}} \! + \! \widehat{\boldsymbol{\mathsf{H}}}_{h}^{\mathbf{c}\mathbf{c}}(\gm{q}{\mathbf{c}})  &\gm{K}{\mathbf{c}\mathbf{d}}\\ 
  ~~\gm{K}{\mathbf{d}1}       &~~\gm{K}{\mathbf{d}\mathbf{c}}                                   							&\bz                          
\end{array}
\end{bmatrix},
\end{equation}
where $\widehat{\boldsymbol{\mathsf{H}}}_{h}^{\mathbf{c}\mathbf{c}}(\gm{q}{\mathbf{c}})$ is obtained by linearizing $\llangle \widehat{\boldsymbol{P}}(\boldsymbol{K}_h),\boldsymbol{\kappa}_h\rrangle$ in \eqref{FEh_2}$_{2}$.

We use neo-Hookean materials for our numerical examples. However, note that our formulation can use any elastic constitutive equation. For compressible solids, we use the energy function
\begin{equation}\label{Wcomp}
	\widehat{W}(I_1,I_3) = \frac{\mu}{2}\left(I_1-3\right)-\frac{\mu}{2}\ln{I_3}+\frac{\kappa}{8}(\ln{I_3})^2,
\end{equation}
where $\mu$ and $\kappa$ are the shear and bulk moduli at the ground state, respectively. Recalling that $\boldsymbol{F} = \boldsymbol{I}+\boldsymbol{K}$, the constitutive relation reads
\begin{equation}\nonumber
 \widehat{\boldsymbol{P}}(\boldsymbol{K}) = \mu\left(\boldsymbol{F}-\boldsymbol{F}^{-\mathsf{T}}\right)+\kappa\ln{J}\boldsymbol{F}^{-\mathsf{T}}.
\end{equation}
To calculate $\mathbb{K}_{t_h}$ defined in (\ref{tan_comp}), one needs to obtain the elasticity tensor $\boldsymbol{\widehat{A}}(\boldsymbol{K})$ by taking the derivative of components of $\widehat{\boldsymbol{P}}(\boldsymbol{K})$ with respect to components of $\boldsymbol{K}$. 
In the implementation (see \cite[(3.36)]{FaYa2018}), it is more convenient to represent the elasticity tensor as a matrix $\widehat{\boldsymbol{\mathsf{A}}}$, whose size is $9\times9$ in 3D. Let $\vek{\boldsymbol{T}}$ be a vector representation of a tensor $\boldsymbol{T}$, and let $\left[\boldsymbol{V}\right]_{\!\times}$ be a skew-symmetric matrix representing a vector $\boldsymbol{V}$ that are defined as
\begin{equation}\nonumber
\vek{\boldsymbol{T}} := {\begin{bmatrix} T^{11} &T^{12} &T^{13} &T^{21} &T^{22} &T^{23} &T^{31} &T^{32} &T^{33} \end{bmatrix}}^{\mathsf{T}},
  ~~\text{and}~~~
	\left[\boldsymbol{V}\right]_{\!\times} := 
  \begin{bmatrix}
   	~0    	&-V^3     &~~V^2 \\
   	~~V^3  	&~0        &-V^1 \\
   -V^2  	&~~V^1      &~0
  \end{bmatrix}.
\end{equation}
Considering (\ref{Wcomp}), one obtains
\begin{equation}\nonumber
\widehat{\boldsymbol{\mathsf{A}}}(\boldsymbol{K}) = 
\mu\boldsymbol{\mathsf{I}}+
(\mu-\kappa\ln{J}+\kappa)\bigVek{\big}{\boldsymbol{F}^{-\mathsf{T}} }\bigVek{\big}{\boldsymbol{F}^{-\mathsf{T}} }^\mathsf{T}
-\frac{\mu-\kappa\ln{J}}{J}\boldsymbol{\mathsf{S}}(\boldsymbol{F}),
\end{equation}
where $\boldsymbol{\mathsf{I}}$ is the $9\times9$ identity matrix and
\begin{equation}\nonumber
\boldsymbol{\mathsf{S}}(\boldsymbol{F}) := 
\begin{bmatrix}
\begin{array}{ccc}
   \bz &-\!\left[\overrightarrow{\boldsymbol{F}}_{\!\!3}\right]_{\!\times}  &~~\left[\overrightarrow{\boldsymbol{F}}_{\!\!2}\right]_{\!\times}\\
   ~~\left[\overrightarrow{\boldsymbol{F}}_{\!\!3}\right]_{\!\times}  &\bz  &-\!\left[\overrightarrow{\boldsymbol{F}}_{\!\!1}\right]_{\!\times} \\
   -\!\left[\overrightarrow{\boldsymbol{F}}_{\!\!2}\right]_{\!\times} &~~\left[\overrightarrow{\boldsymbol{F}}_{\!\!1}\right]_{\!\times}      &\bz
\end{array}
\end{bmatrix}_{9\times9}. 
\end{equation}
For incompressible solids with $I_3=1$, the following neo-Hookean energy function is used.
\begin{equation}\label{Wicmp}
	\widetilde{W}(I_1) = \frac{\mu}{2}(I_1-3).
\end{equation}
The constitutive part of stress is $\widetilde{P}(\boldsymbol{K})=\mu(\boldsymbol{I}+\boldsymbol{K})$.
To impose the incompressibility constraint $J=1$, we use the constraint functions $C_1(J) = J-1$ or $C_2(J) = \ln J$; we choose the function that results in a better numerical performance of the method in a given example. To obtain $\mathbb{N}_h$ in (\ref{FEd}) and $\mathbb{K}_{t_h}$ in (\ref{tan_icmp}), one needs the following matrices:
\begin{align*}
\boldsymbol{Q}_1(\boldsymbol{K}) &= J\boldsymbol{F}^{-\mathsf{T}}, 
\quad
\widetilde{\boldsymbol{\mathsf{A}}}_{1}(\boldsymbol{K},p) = \mu\boldsymbol{\mathsf{I}} +p\,\boldsymbol{\mathsf{S}}(\boldsymbol{F}), \\
\boldsymbol{Q}_2(\boldsymbol{K}) &= \boldsymbol{F}^{-\mathsf{T}}, 
\quad\quad\!\!
\widetilde{\boldsymbol{\mathsf{A}}}_{2}(\boldsymbol{K},p) = \mu\boldsymbol{\mathsf{I}}-p\bigVek{\big}{\boldsymbol{F}^{-\mathsf{T}} }\bigVek{\big}{\boldsymbol{F}^{-\mathsf{T}} }^\mathsf{T}+\frac{p}{J}\boldsymbol{\mathsf{S}}(\boldsymbol{F}).
\end{align*}

\subsection{Solvability and Stability} \label{solvability}
\begin{thm}\label{solve} Let $N^1$, $N^c$, $N^d$, and $N^\ell$ be the numbers of degrees of freedom in $\gm{q}{1}$, $\gm{q}{\mathbf{c}}$, $\gm{q}{\mathbf{d}}$, and $\gm{q}{\ell}$, respectively. For $\alpha>0$, the tangent stiffness matrix $\mathbb{K}_{t_h}$ of incompressible solids (\ref{tan_icmp}) is non-singular if and only if the following conditions hold.
\begin{enumerate}[$(i)$]
\item $\operatorname{ker}(\gm{H}{\mathbf{c}\ell}) = \{\boldsymbol{0}_{N^\ell\times1}\}$,
 \item $\operatorname{ker}(\gm{K}{1\mathbf{d}})\cap\operatorname{ker}(\boldsymbol{\mathsf{B}}^{\ell\mathbf{c}}_{0}\gm{K}{\mathbf{c}\mathbf{d}})= \{\boldsymbol{0}_{N^d\times1}\}$,
 \quad $\left( \operatorname{ker}(\gm{K}{\mathbf{c}\mathbf{d}}) \subseteq \operatorname{ker}(\boldsymbol{\mathsf{B}}^{\ell\mathbf{c}}_{0}\gm{K}{\mathbf{c}\mathbf{d}})\right)$,
\item $\operatorname{ker}\left(\gm{H}{\mathbf{c}\mathbf{c}}+\alpha\gm{M}{\mathbf{c}\mathbf{c}}-\alpha\gm{M}{\mathbf{c}1}\left(\gm{M}{11}\right)^{-1}\gm{M}{1\mathbf{c}}\right)\cap\operatorname{ker}\left(\gm{K}{\mathbf{d}\mathbf{c}}-\gm{K}{\mathbf{d}1}\left(\gm{M}{11}\right)^{-1}\gm{M}{1\mathbf{c}}\right)\cap\operatorname{ker}\left(\gm{H}{\ell\mathbf{c}}\right)=\{\boldsymbol{0}_{N^c\times1}\}$,
\end{enumerate}
where $\boldsymbol{\mathsf{B}}^{\ell\mathbf{c}}_{0}$ is a matrix whose rows form a basis for $\operatorname{ker}(\gm{H}{\ell\mathbf{c}})$. For $\alpha=0$, $\mathbb{K}_{t_h}$ is non-singular if and only if $(i)$ and $(ii)$ and the following conditions hold.
\begin{enumerate}[$(i)^\prime$] \setcounter{enumi}{2}
\item $\operatorname{ker}(\gm{K}{\mathbf{d}1}) = \{\boldsymbol{0}_{N^1\times1}\}$,
\item $\operatorname{ker}(\gm{H}{\mathbf{c}\mathbf{c}})\cap\operatorname{ker}\left(\boldsymbol{\mathsf{B}}^{1\mathbf{d}}_{0}\gm{K}{\mathbf{d}\mathbf{c}}\right)\cap\operatorname{ker}\left(\gm{H}{\ell\mathbf{c}}\right)=\{\boldsymbol{0}_{N^c\times1}\}$,
\end{enumerate}
where $\boldsymbol{\mathsf{B}}^{1\mathbf{d}}_{0}$ is a matrix whose rows form a basis for $\operatorname{ker}(\gm{K}{1\mathbf{d}})$.
\end{thm}

\begin{proof} Rearrange the rows and the columns of $\mathbb{K}_{t_h}$ to obtain
 \begin{equation}\nonumber
  \mathbb{K}_{t_h} =
  \begin{bmatrix}
   \mathbb{A}_h  &\mathbb{B}^\mathsf{T}_h    \\
   \mathbb{B}_h         &\bz                  
  \end{bmatrix},
  ~~~~
  \mathbb{A}_h =
  \begin{bmatrix}
   \gm{H}{\mathbf{c}\mathbf{c}}  &\bz    \\
   \bz         &\bz                  
  \end{bmatrix}+
   \alpha 
   \begin{bmatrix}
   \gm{M}{\mathbf{c}\mathbf{c}}  &\gm{M}{\mathbf{c}1}    \\
   \gm{M}{1\mathbf{c}}        	 &\gm{M}{11}              
  \end{bmatrix},
  ~~~~
  \mathbb{B}_h = 
  \begin{bmatrix}
  \gm{K}{\mathbf{d}\mathbf{c}}  &\gm{K}{\mathbf{d}1} \\
   \gm{H}{\ell\mathbf{c}}  &\bz                        
  \end{bmatrix},
   ~~~~
  \mathbb{B}^\mathsf{T}_h = 
  \begin{bmatrix}
   \gm{K}{\mathbf{c}\mathbf{d}} &\gm{H}{\mathbf{c}\ell}\\
   \gm{K}{1\mathbf{d}}          &\bz                           
  \end{bmatrix}.
\end{equation}
Then, according to \citep[Theorem 3.2.1]{boffi2013mixed}, the matrix $\mathbb{K}_{t_h}$ is non-singular if and only if the following holds:
\begin{enumerate}[$(1)$]
\item The restriction of $\mathbb{A}_h$ to $\operatorname{ker}(\mathbb{B}_h)$ is surjective (or equivalently injective),
\item $\mathbb{B}_h$ is surjective (or equivalently $\mathbb{B}^\mathsf{T}_h$ is injective or $\operatorname{ker}(\mathbb{B}^\mathsf{T}_h) = \{\boldsymbol{0}\}$).
\end{enumerate}

Consider the following sets:
\begin{equation}\nonumber
\begin{aligned}
	S_1 &:= \left\{ \begin{bmatrix}  \boldsymbol{0}_{N^c\times1} \\ \boldsymbol{\mathsf{Y}}_{N^1\times1}\end{bmatrix}: 
	\boldsymbol{0}\neq\boldsymbol{\mathsf{Y}}\in\operatorname{ker}(\gm{K}{\mathbf{d}1}) \right\},\\
	S_2 &:= \left\{ \begin{bmatrix}  \boldsymbol{\mathsf{X}}_{N^c\times1}\\ \boldsymbol{\mathsf{Y}}_{N^1\times1}\end{bmatrix}:
	\boldsymbol{0}\neq\boldsymbol{\mathsf{X}}\in\operatorname{ker}(\gm{H}{\ell\mathbf{c}}) \text{ and }  \gm{K}{\mathbf{d}\mathbf{c}}\boldsymbol{\mathsf{X}}+\gm{K}{\mathbf{d}1}\boldsymbol{\mathsf{Y}}=\boldsymbol{0} \text{ for some }\boldsymbol{\mathsf{Y}}\in\mathbb{R}^{N^1} \right\},\\
	S^\prime_1 &:= \left\{ \begin{bmatrix}  \boldsymbol{0}_{N^d\times1} \\ \boldsymbol{\mathsf{Y}}_{N^\ell\times1}\end{bmatrix}: 
	\boldsymbol{0}\neq\boldsymbol{\mathsf{Y}}\in\operatorname{ker}(\gm{H}{\mathbf{c}\ell}) \right\},\\
	S^\prime_2 &:= \left\{ \begin{bmatrix}  \boldsymbol{\mathsf{X}}_{N^d\times1}\\ \boldsymbol{\mathsf{Y}}_{N^\ell\times1}\end{bmatrix}:
	 \boldsymbol{0}\neq\boldsymbol{\mathsf{X}}\in\operatorname{ker}(\gm{K}{1\mathbf{d}}) \text{ and }  \gm{K}{\mathbf{c}\mathbf{d}}\boldsymbol{\mathsf{X}}+\gm{H}{\mathbf{c}\ell}\boldsymbol{\mathsf{Y}}=\boldsymbol{0} \text{ for some }\boldsymbol{\mathsf{Y}}\in\mathbb{R}^{N^\ell} \right\}.
\end{aligned}
\end{equation}
One can show that
\begin{equation}\nonumber
	\operatorname{ker}\left(\mathbb{B}_h\right) = \left\{\boldsymbol{0}_{(N^c+N^1)\times1}\right\} \sqcup S_1 
	\sqcup S_2, ~~\text{and}~~~
	\operatorname{ker}\left(\mathbb{B}^\mathsf{T}_h\right) 
	= \left\{\boldsymbol{0}_{(N^d+N^\ell)\times1}\right\} \sqcup S^\prime_1 \sqcup S^\prime_2.
\end{equation}
Therefore, the requirement $(2)$ ($\operatorname{ker}(\mathbb{B}^\mathsf{T}_h) = \{\boldsymbol{0}\}$) is equivalent to $S^\prime_1=S^\prime_2=\emptyset$. It is straightforward to show that $S^\prime_1=\emptyset$ is equivalent to $(i)$. We write $S^\prime_2=\emptyset$ as $\operatorname{ker}(\gm{K}{1\mathbf{d}})\cap s^{\prime}_2=\{\boldsymbol{0}\}$, where $s^{\prime}_2 = \left\{ \boldsymbol{\mathsf{X}}: \gm{K}{\mathbf{c}\mathbf{d}}\boldsymbol{\mathsf{X}}= -\gm{H}{\mathbf{c}\ell}\boldsymbol{\mathsf{Y}}\text{ for some }\boldsymbol{\mathsf{Y}}\in\mathbb{R}^{N^\ell} \right\}$. Using $\operatorname{Im}(\gm{H}{\mathbf{c}\ell})=\operatorname{ker}(\gm{H}{\ell\mathbf{c}})^{\perp}$, where the superscript $\perp$ indicates the orthogonal complement, one can readily show that $s^{\prime}_2 = \left\{ \boldsymbol{\mathsf{X}}: \boldsymbol{\mathsf{b}}_{0}^{\mathsf{T}}\gm{K}{\mathbf{c}\mathbf{d}}\boldsymbol{\mathsf{X}}= 0, \forall \boldsymbol{\mathsf{b}}_{0}\in \operatorname{ker}(\gm{H}{\ell\mathbf{c}}) \right\}$. Let $\boldsymbol{\mathsf{B}}^{\ell\mathbf{c}}_{0}$ be a matrix whose rows form a basis for $\operatorname{ker}(\gm{H}{\ell\mathbf{c}})$, then $s^{\prime}_2=\operatorname{ker}(\boldsymbol{\mathsf{B}}^{\ell\mathbf{c}}_{0}\gm{K}{\mathbf{c}\mathbf{d}})$. Therefore, $S^\prime_2=\emptyset$ is equivalent to $(ii)$. 
Next, we assume that $\alpha>0$ and show that $(1)$ is equivalent to $(iii)$. For $\alpha>0$, one can write
\begin{equation}\nonumber
\operatorname{ker}(\mathbb{A}_h)=
\left\{\begin{bmatrix}  \boldsymbol{\mathsf{X}}_{N^c\times1}\\ \boldsymbol{\mathsf{Y}}_{N^1\times1}\end{bmatrix}: \boldsymbol{\mathsf{X}}\in\operatorname{ker}\left(\gm{H}{\mathbf{c}\mathbf{c}}+\alpha\gm{M}{\mathbf{c}\mathbf{c}}-\alpha\gm{M}{\mathbf{c}1}\left(\gm{M}{11}\right)^{-1}\gm{M}{1\mathbf{c}}\right) \text{ and }  \boldsymbol{\mathsf{Y}}=-\left(\gm{M}{11}\right)^{-1}\gm{M}{1\mathbf{c}}\boldsymbol{\mathsf{X}} \right\},
\end{equation}
where use was made of the fact that $\gm{M}{11}$ is a Gram matrix and positive-definite by construction (and hence injective). Note that $(1)$ is  equivalent to $\operatorname{ker}(\mathbb{A}_h)\cap\operatorname{ker}(\mathbb{B}_h)=(\operatorname{ker}(\mathbb{A}_h)\cap\{\boldsymbol{0}\})\sqcup(\operatorname{ker}(\mathbb{A}_h) \cap S_1)\sqcup(\operatorname{ker}(\mathbb{A}_h) \cap S_2)=\{\boldsymbol{0}\}$, which is equivalent to $\operatorname{ker}(\mathbb{A}_h) \cap S_1 = \operatorname{ker}(\mathbb{A}_h) \cap S_2=\emptyset$. We know that $\mathbb{A}_h\boldsymbol{\mathsf{Q}}\neq\boldsymbol{0}$, $\forall\boldsymbol{\mathsf{Q}}\in S_1$, due to injectivity of $\gm{M}{11}$, so $\operatorname{ker}(\mathbb{A}_h) \cap S_1 =\emptyset$ is trivial. The remaining condition $\operatorname{ker}(\mathbb{A}_h) \cap S_2 =\emptyset$ simplifies to $(iii)$. For $\alpha=0$, one can write
\begin{equation}\nonumber
\operatorname{ker}(\mathbb{A}_h)=
\left\{\begin{bmatrix}  \boldsymbol{\mathsf{X}}_{N^c\times1}\\ \boldsymbol{\mathsf{Y}}_{N^1\times1}\end{bmatrix}: \boldsymbol{\mathsf{X}}\in\operatorname{ker}\left(\gm{H}{\mathbf{c}\mathbf{c}}\right) \text{ and }  \boldsymbol{\mathsf{Y}}\in\mathbb{R}^{N^1} \right\}.
\end{equation}
Now, $\operatorname{ker}(\mathbb{A}_h) \cap S_1=\emptyset$ simplifies to $(iii)^{\prime}$,  and $\operatorname{ker}(\mathbb{A}_h) \cap S_2=\emptyset$ simplifies to $(iv)^{\prime}$.
\end{proof}

\begin{cor}\label{solve1} For $\alpha>0$, the tangent stiffness matrix $\mathbb{K}_{t_h}$ of compressible solids (\ref{tan_comp}) is non-singular if and only if the following conditions hold:
\begin{enumerate}[$(i)$]
 \item $\operatorname{ker}(\gm{K}{1\mathbf{d}})\cap\operatorname{ker}(\gm{K}{\mathbf{c}\mathbf{d}})= \{\boldsymbol{0}_{N^d\times1}\}$,
\item $\operatorname{ker}\left(\gm{H}{\mathbf{c}\mathbf{c}}+\alpha\gm{M}{\mathbf{c}\mathbf{c}}-\alpha\gm{M}{\mathbf{c}1}\left(\gm{M}{11}\right)^{-1}\gm{M}{1\mathbf{c}}\right)\cap\operatorname{ker}\left(\gm{K}{\mathbf{d}\mathbf{c}}-\gm{K}{\mathbf{d}1}\left(\gm{M}{11}\right)^{-1}\gm{M}{1\mathbf{c}}\right)=\{\boldsymbol{0}_{N^c\times1}\}$.
\end{enumerate}
For $\alpha=0$, $\mathbb{K}_{t_h}$ is non-singular if and only if $(i)$ and the following conditions hold:
\begin{enumerate}[$(i)^\prime$] \setcounter{enumi}{1}
\item $\operatorname{ker}(\gm{K}{\mathbf{d}1}) = \{\boldsymbol{0}_{N^1\times1}\}$,
\item $\operatorname{ker}(\gm{H}{\mathbf{c}\mathbf{c}})\cap\operatorname{ker}\left(\boldsymbol{\mathsf{B}}^{1\mathbf{d}}_{0}\gm{K}{\mathbf{d}\mathbf{c}}\right)=\{\boldsymbol{0}_{N^c\times1}\}$.
\end{enumerate}
\end{cor}

\begin{cor}\label{Solve2} If the tangent stiffness matrix $\mathbb{K}_{t_h}$ is non-singular, then
\begin{enumerate}[$(1)$]
\item $N^d\leq N^c+N^1$ for $\alpha\geq0$ and for both compressible and incompressible solids,
\item $N^\ell \leq N^c$ only for incompressible solids,
\item $N^1 \leq N^d$ only for $\alpha=0$.
\end{enumerate}
\end{cor}
\begin{proof} Noting that $\operatorname{ker}(\gm{K}{\mathbf{c}\mathbf{d}}) \subseteq \operatorname{ker}(\boldsymbol{\mathsf{B}}^{\ell\mathbf{c}}_{0}\gm{K}{\mathbf{c}\mathbf{d}})$, both Theorem \ref{solve} $(ii)$ and Corollary \ref{solve1} $(i)$, imply $(1)$. Theorem \ref{solve} (i) implies $(2)$. Both Theorem \ref{solve} $(iii)^\prime$ and Corollary \ref{solve1} $(ii)^\prime$, imply $(3)$.
\end{proof}

\noindent In view of Theorem \ref{solve}, one can see how adding (\ref{stab}) to the weak formulation (\ref{ElasMixF_1}) may improve the stability and the performance of the resulting finite element methods. 
Without the stabilization terms ($\alpha=0$), the violation of \mbox{$\operatorname{ker}(\gm{K}{\mathbf{d}1}) = \{\boldsymbol{0}_{N^1\times1}\}$} or more strongly having $ N^1>N^d$ leads to a singular tangent stiffness matrix $\mathbb{K}_{t_h}$. This restricts the choices of finite elements for the displacement and stress in both 2D and 3D. 
In particular, considering $(\boldsymbol{U}_h,\boldsymbol{P}_h)$ in $V_{h,m}^{1}\times V_{h,n}^{\mathbf{d}-}$ or $V_{h,m}^{1}\times V_{h,n}^{\mathbf{d}}$ such that $m>n$ results in a singular $\mathbb{K}_{t_h}$ in both 2D and 3D independent of the size of the mesh. Adding the stabilization terms (\ref{stab}) ($\alpha>0$) overcomes this limitation and enables one to improve the convergence of the displacement field by discretizing it using second-order shape functions without the need for modifying the finite elements of other fields. For instance, for $\alpha=0$, the finite elements (\ref{RFEs_m}) result in a singular system, but they converge to correct solutions for large values of $\alpha$. To avoid the singularity of (\ref{RFEs_m}) for $\alpha=0$, we have no choice but to approximate the displacement field using first-order polynomials and to compromise the rate of convergence of the method. Also, note that approximating the displacement $\boldsymbol{U}$ in a second-order polynomial space leads to a more accurate discretization of $\boldsymbol{K}=\mathbf{grad}\,\boldsymbol{U}$ as the intersection of image of $\mathbf{grad}$ and the approximation space of $\boldsymbol{K}$ becomes larger at the discrete level. 

We next discuss how modifying the finite element of the displacement gradient $\boldsymbol{K}$ in (\ref{RFEs_m}) and its resulting finite element space (\ref{Vc3}) lead to solvability of the mixed finite element methods (\ref{FEh}) and (\ref{FEh_2}). 
Let $(\boldsymbol{U}_h,\boldsymbol{K}_h,\boldsymbol{P}_h) \in V_{h,m}^{1} \times V_{h,n}^{\mathbf{c}-}(V_{h,n}^{\mathbf{c}}) \times V_{h,k}^{\mathbf{d}-}(V_{h,k}^{\mathbf{d}})$ for $m,n,k = 1,2$, which results in $32$ different combinations (note that pressure is not relevant here).  
In 3D, all the $32$ combinations except $V_{h,m}^{1} \times V_{h,2}^{\mathbf{c}} \times V_{h,1}^{\mathbf{d}-}, m=1,2$ result in a singular $\mathbb{K}_{t_h}$. These combinations either give $N^d > N^c+N^1$ for any mesh or their smallest singular value of $\begin{bmatrix} \gm{K}{\mathbf{d}\mathbf{c}} &\gm{K}{\mathbf{d}1} \end{bmatrix}^\mathsf{T}$ goes to zero as one refines the mesh (see \cite[Remark 16]{FaYa2018}). Any of these two cases is a violation of Theorem \ref{solve} $(ii)$ (or Corollary \ref{solve1} $(i)$).
The two remaining choices $V_{h,m}^{1} \times V_{h,2}^{\mathbf{c}}\times V_{h,1}^{\mathbf{d}-}, m=1,2$ are not practical as they have poor performances considering their expensive computational cost. $V_{h,2}^{\mathbf{c}}$ of displacement gradient has $90$ degrees of freedom per element, which significantly increases the computational cost, but paired with the lowest-order space of stress $V_{h,1}^{\mathbf{d}-}$, it cannot improve the overall convergence of the method.
To resolve this issue, we proposed $\overline{V}_{h,3}^{\mathbf{c}}$ in (\ref{Vc3}) and considered $(\boldsymbol{U}_h,\boldsymbol{K}_h,\boldsymbol{P}_h) \in V_{h,2}^{1} \times \overline{V}_{h,3}^{\mathbf{c}} \times V_{h,1}^{\mathbf{d}-}$. Note that, in each element, $\overline{V}_{h,3}^{\mathbf{c}}$ has only $9$ degrees of freedom more than the first-order space $V_{h,1}^{\mathbf{c}}$ with $36$ degrees of freedom (see \mbox{ Figure \ref{CSFEMs3D} }). Hence, it does not increase the computational cost of the method significantly. 
Moreover, we observe that the smallest singular value of $[ \gm{K}{\mathbf{d}\mathbf{c}} ~\gm{K}{\mathbf{d}1} ]^\mathsf{T}$ for $V_{h,2}^{1} \times \overline{V}_{h,3}^{\mathbf{c}} \times V_{h,1}^{\mathbf{d}-}$ remains positive as we refine different arbitrary meshes. 

So far, we have discussed that, for $\alpha>0$, (\ref{FEh}) and (\ref{FEh_2}) do not result in a singular $\mathbb{K}_{t_h}$ even if $\operatorname{ker}(\gm{K}{\mathbf{d}1}) \neq \{\boldsymbol{0}_{N^1\times1}\}$, and they result in $N^d\leq N^c+N^1$ and $\operatorname{ker}(\gm{K}{1\mathbf{d}}) \cap \operatorname{ker}(\gm{K}{\mathbf{c}\mathbf{d}}) = \{\boldsymbol{0}_{N^d\times1}\}$, which are required for satisfying Theorem \ref{solve} $(ii)$ or Corollary \ref{solve1} $(i)$. These have been made possible through studying the linear operators $\gm{K}{\mathbf{c}\mathbf{d}}$ and $\gm{K}{1\mathbf{d}}$ in $\mathbb{K}_{t_h}$, which are independent of the physics of the problem. 
The \emph{stability} of (\ref{FEh}) requires that all the conditions of Theorem \ref{solve} hold as one refines the mesh. However, given the nonlinear nature of the problems of interest here, this is difficult to check. In particular, the nonlinear operators $\gm{H}{\mathbf{c}\mathbf{c}}(\gm{q}{\mathbf{c}},\gm{q}{\ell})$ and $\gm{H}{\mathbf{c}\ell}(\gm{q}{\mathbf{c}})$ in $\mathbb{K}_{t_h}$ depend on the material properties of the body and its state of deformation. Therefore, one cannot draw a general conclusion for stability or convergence of the mixed methods only by studying the formulations and without considering the physics of the problem. Based on the various numerical examples presented in the next section, we have concluded that (\ref{FEh}) and (\ref{FEh_2}) have an overall good performance in capturing the large deformations of incompressible and compressible solids in 3D.

\section{Numerical Examples}

In this section, we consider several examples to assess the performance of the mixed finite elements (\ref{FEh}) and (\ref{FEh_2}) in modeling compressible and incompressible solids in 3D. We use the Frobenius norm $\|\boldsymbol{T}\|:= (\sum\nolimits_{I,J}T^{IJ}T^{IJ})^{\frac{1}{2}}$ for $\boldsymbol{K}_h$ and $\boldsymbol{P}_h$ in the deformed configurations. We use the $L^{2}$-norm for $\boldsymbol{U}_h$, $\boldsymbol{K}_h$, $\boldsymbol{P}_h$, and $p_h$ over the entire mesh in convergence analyses. We use $\alpha=1\times 10^6$ in all the examples (the solutions actually converge for smaller values of $\alpha$ in each example; assuming larger values does not change the solutions).

\newcommand{\norm}[1]{\left\lVert#1\right\rVert}

\begin{figure}[H]
\begin{center}
\includegraphics[width = .8\textwidth]{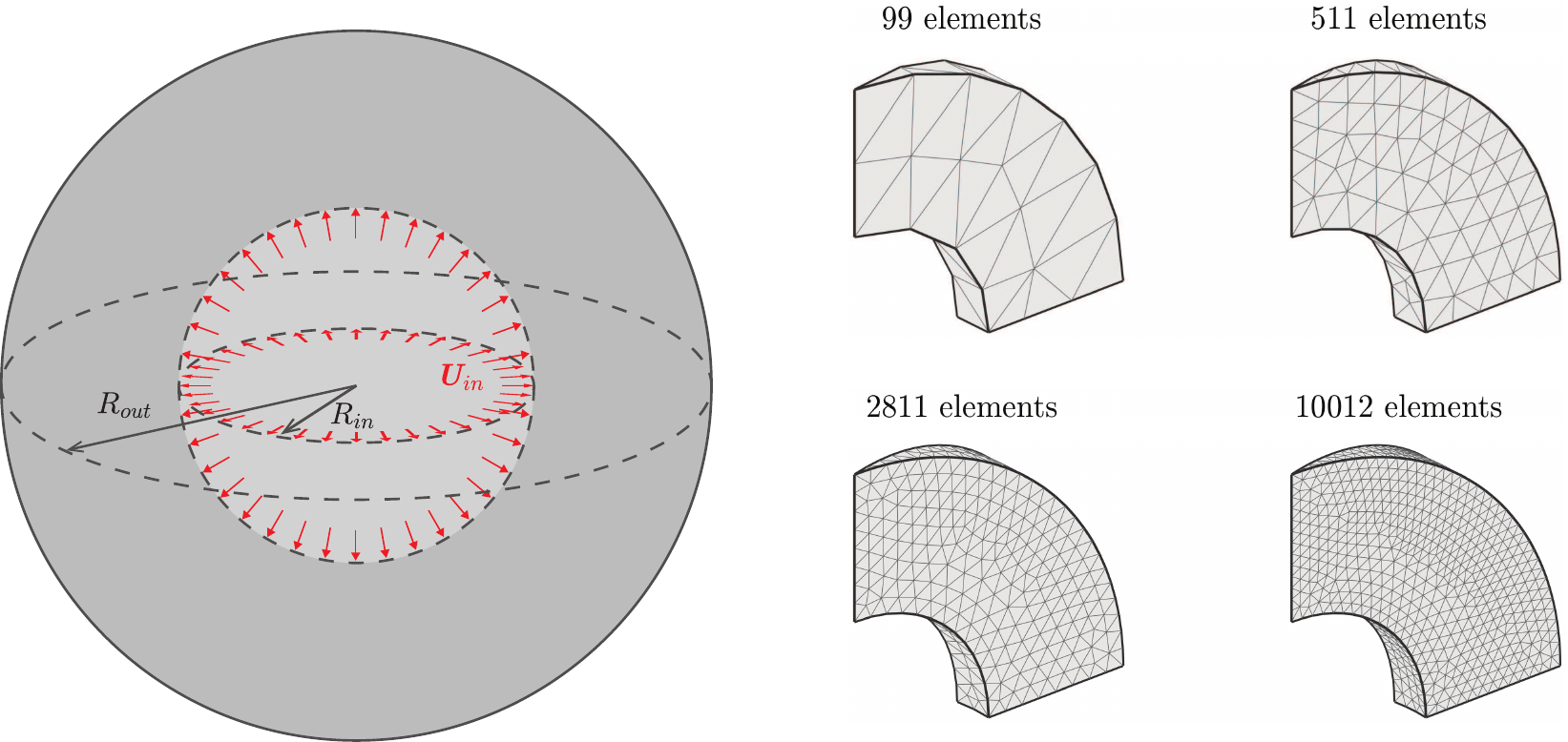}
\end{center}
\caption{\footnotesize Inflation of a hollow spherical ball: Geometry and four unstructured meshes. The outer boundary of the sphere is traction free.}
\label{EX1Mesh}
\end{figure}

\paragraph{Example 1: Inflation of a Hollow Spherical Ball.} Let us consider an incompressible hollow spherical ball shown in Figure \ref{EX1Mesh}. We assume that the inner boundary of the ball is subjected to the displacement boundary condition $\boldsymbol{U}_{\text{in}} = (\lambda-1)\mathbf{X}$, the outer boundary is traction free, and there are no body forces.  This is an example of a universal deformation \citep{Ericksen1954} and the exact solution reads
\begin{equation}
	\boldsymbol{U}_{e}(\mathbf{X}) = \left[\frac{r(R)}{R}-1\right]\mathbf{X},~~~~~
  p_{e}(\mathbf{X}) = -\mu \frac{R_{\text{out}}^4}{r^4(R_{\text{out}})} 
  +\frac{\mu}{2}\left[g(R)-g(R_{\text{out}})\right],
\end{equation}
where $R=\|\mathbf{X}\|$, $r(R) = \left(R^3+(\lambda^3-1)R_{\text{in}}^3\right)^{\frac{1}{3}}$, and $g(R) = R\left(3r^3(R)+(\lambda^3-1)R^3_{\text{in}}\right)/r^4(R)$. It follows that $\boldsymbol{K}_e=\mathbf{grad}\,\boldsymbol{U}_{e}$, and $\boldsymbol{P}_e = \widetilde{\boldsymbol{P}}(\boldsymbol{K}_e)+p_e\boldsymbol{Q}(\boldsymbol{K}_e)$. Having the exact solution, we assess the accuracy and convergence of CSFEM given in (\ref{FEh}). For our computations, we consider the \mbox{neo-Hookean} energy function (\ref{Wicmp}) with $\mu =1~\mathrm{N}/\mathrm{mm}^{2}$, the constraint function $C(J) = J-1$, $R_{\text{in}} = 0.5\,\mathrm{mm}$, $R_{\text{out}} = 1\,\mathrm{mm}$, and $\lambda=3$. Using symmetry, we model only $1/24$ of a hemisphere as shown in Figure \ref{EX1Mesh}.
To study the convergence order of (\ref{FEh}), we plot the relative errors of the field variables versus the maximum diameter $h$ of some unstructured meshes in a log-log graph in Figure \ref{EX1Order}. The convergence order of the displacement $\boldsymbol{U}_h$ is close to $2$, and those of the displacement gradient $\boldsymbol{K}_h$, the stress $\boldsymbol{P}_h$, and the pressure-like variable $p_h$ are almost $1$.
Figure \ref{EX1ConfStress} shows the reference and the deformed configurations of the four unstructured meshes given in Figure \ref{EX1Mesh} obtained using CSFEM in (\ref{FEh}) for $\lambda=3$. Colors show the values of $\|\boldsymbol{K}_h\|$ in the first row and the values of $p_h$ in the second row with lighter colors associated with the larger values.
\begin{figure}[htb]
\begin{center}
\includegraphics[width = \textwidth]{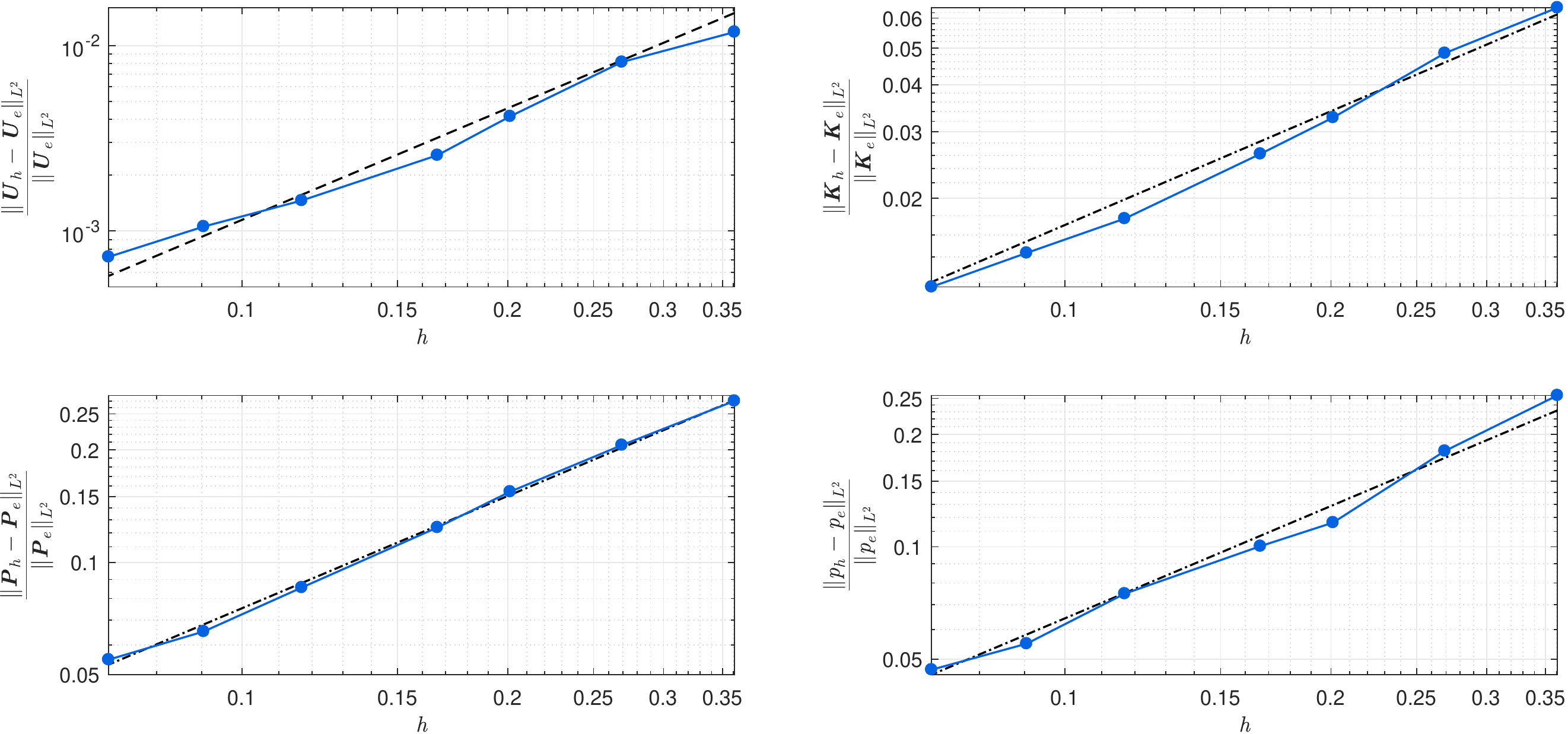}
\end{center}
\caption{\footnotesize Relative $L^{2}$-norms of errors for approximating displacement, displacement gradient, stress, and pressure versus the maximum diameter $h$ using (\ref{FEh}). The dash-dot and the dashed lines have the slopes of $1$ and $2$, respectively. }
\label{EX1Order}
\end{figure}
\begin{figure}[H]
\begin{center}
\vspace*{0.2in}
\includegraphics[width = \textwidth]{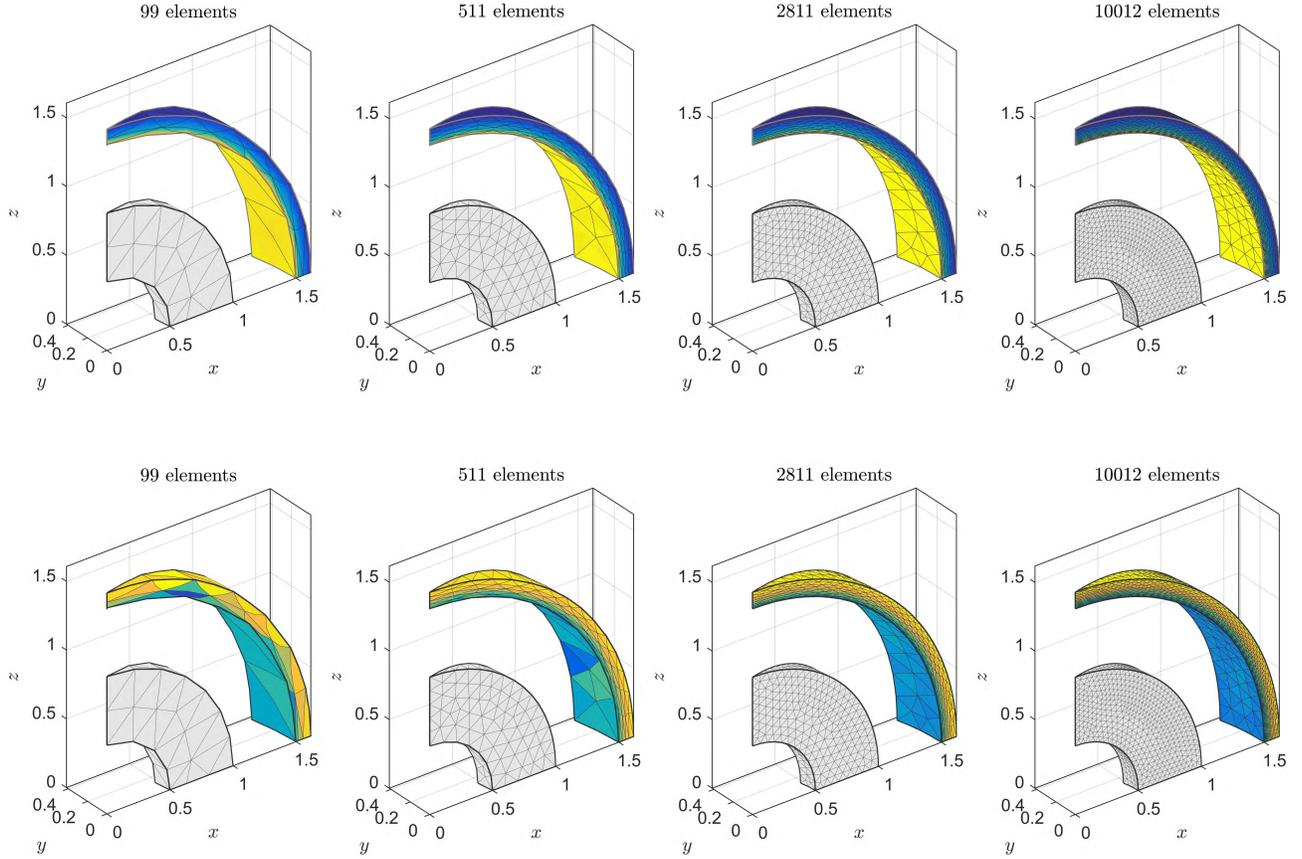} 
\end{center}
\caption{\footnotesize The reference and the deformed configurations of the sphere for  $\lambda=3$ using (\ref{FEh}). Colors indicate values of $\|\boldsymbol{K}_h\|$ in the first row and pressure $p_h$ in the second row, where lighter colors correspond to larger values.}
\label{EX1ConfStress}
\end{figure}

\paragraph{Example 2: $\boldsymbol{3}$D Cook's Membrane.} 
In this example, the 3D Cook's membrane problem depicted in Figure \ref{EX2Mesh} is analyzed in order to study the performance of CSFEMs in bending analysis. We consider two cases of tractions imposed on the right side of the membrane (on $16\,\mathrm{mm}\times10\,\mathrm{mm}$ face): $\overline{\boldsymbol{T}}_1 = \left(0,f,0\right)$ and $\overline{\boldsymbol{T}}_2 = \left(0,2f,f\right)$. We use the energy function (\ref{Wicmp}) with $\mu = 1\,\mathrm{N}/\mathrm{mm}^{2}$ and $C(J) = \ln J$ to impose the incompressibility constraint.
Figure \ref{EX2Valid} shows the convergence of the vertical displacement of point $A$ indicated in Figure \ref{EX2Mesh} for different values of traction $\overline{\boldsymbol{T}}_1 = \left(0,f,0\right)$ using the mixed method (\ref{FEh}). Since the membrane deforms in two dimensions, the results of the 3D analysis using (\ref{FEh}) are compared to those obtained by a 2D analysis using $\mathsf{H}2\mathsf{c}2\mathsf{d}\bar{2}\mathsf{L}1$ in \citep{FaYa2018}. The comparison shows a good agreement between the two analyses. Considering $\overline{\boldsymbol{T}}_2 = \left(0,2f,f\right)$, the membrane deforms in three dimensions, for which a similar convergence graph for point $A$ is presented in Figure \ref{EX2Valid2}. 
The convergence of the independent field variables $(\boldsymbol{U}_h,\boldsymbol{K}_h, \boldsymbol{P}_h, p_h)$ obtained using the mixed method (\ref{FEh}) is illustrated in Figure \ref{EX2Conv} for different values of $\overline{\boldsymbol{T}}_2 = \left(0,2f,f\right)$. One observes that $\boldsymbol{U}_h$ and $\boldsymbol{K}_h$ have a faster convergence in comparison with $\boldsymbol{P}_h$ or $p_h$.
The deformed configurations of the four meshes in Figure \ref{EX2Mesh} using the mixed method (\ref{FEh}) are given in Figure \ref{EX2ConfStress1} and Figure \ref{EX2ConfStress2} for $\overline{\boldsymbol{T}}_1 = \left(0,0.3,0\right)\,\mathrm{N}/\mathrm{mm}^{2}$ and $\overline{\boldsymbol{T}}_2 = \left(0,0.2,0.1\right)\,\mathrm{N}/\mathrm{mm}^{2}$, respectively. In both figures, colors indicate the values of $\|\boldsymbol{p}_h\|$ in the first row and the values of $p_h$ in the second row with lighter colors corresponding to larger values.
It is well-known that the standard displacement-pressure mixed methods for incompressible materials approximate displacement accurately but  suffer form numerical artifacts in approximating pressure (they are unable to provide an approximation of stress either). 
By contrast, Figures \ref{EX2ConfStress1} and \ref{EX2ConfStress2} clearly show that the mixed method (\ref{FEh}) does not suffer from any numerical artifacts in approximating the stresses and the pressure in a large deformation of an incompressible solid even for relatively coarse meshes. 
\begin{figure}[H]
\vspace*{0.3in}
\begin{center}
\includegraphics[width = 0.8\textwidth]{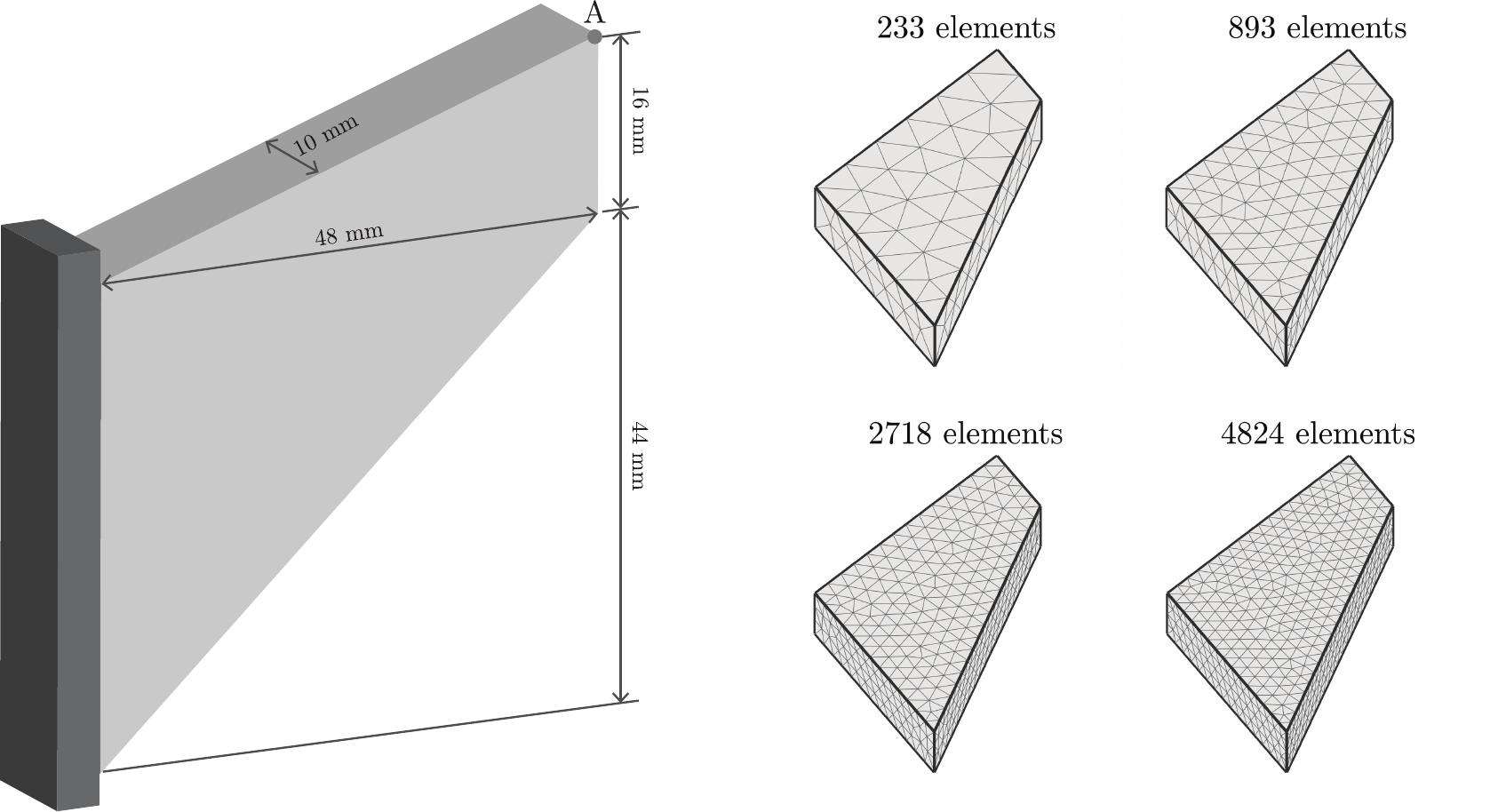}
\end{center}
\caption{\footnotesize 3D Cook's membrane: Geometry and four unstructured meshes. }
\label{EX2Mesh}
\end{figure}

\begin{figure}[H] 
\begin{center}
\includegraphics[width = 0.6\textwidth]{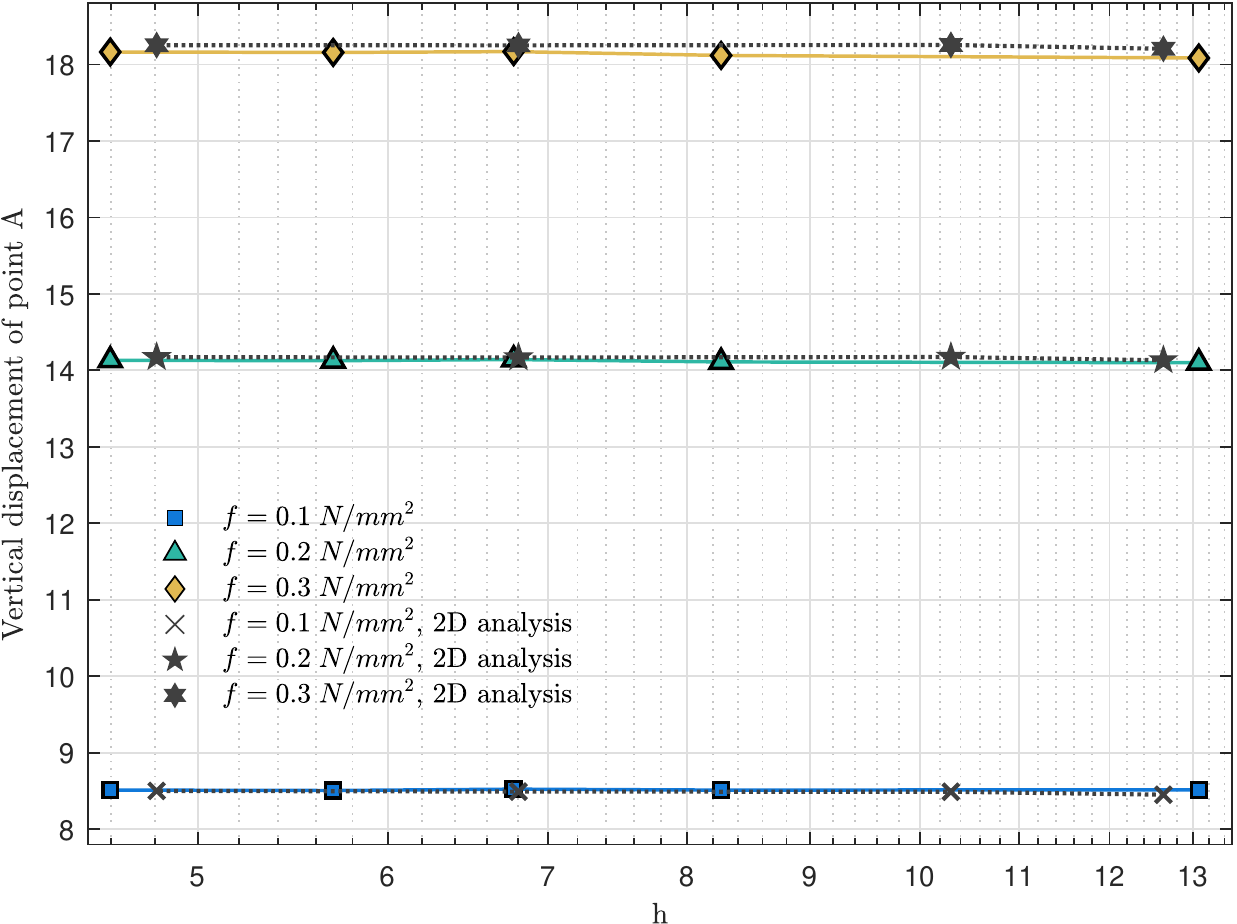}
\end{center}
\caption{\footnotesize 3D Cook's membrane: Vertical displacement of point $A$ in Figure \ref{EX2Mesh} for different values of traction $\overline{\boldsymbol{T}}_1= \left(0,f,0\right)$ versus the maximum edge length $h$ in the mesh using (\ref{FEh}). The dotted line indicates the results of $\mathsf{H}2\mathsf{c}2\mathsf{d}\bar{2}\mathsf{L}1$ given in \citep{FaYa2018}. } 
\label{EX2Valid}
\end{figure}

\begin{figure}[H] 
\begin{center}
\vspace*{0.2in}
\includegraphics[width = 0.5\textwidth]{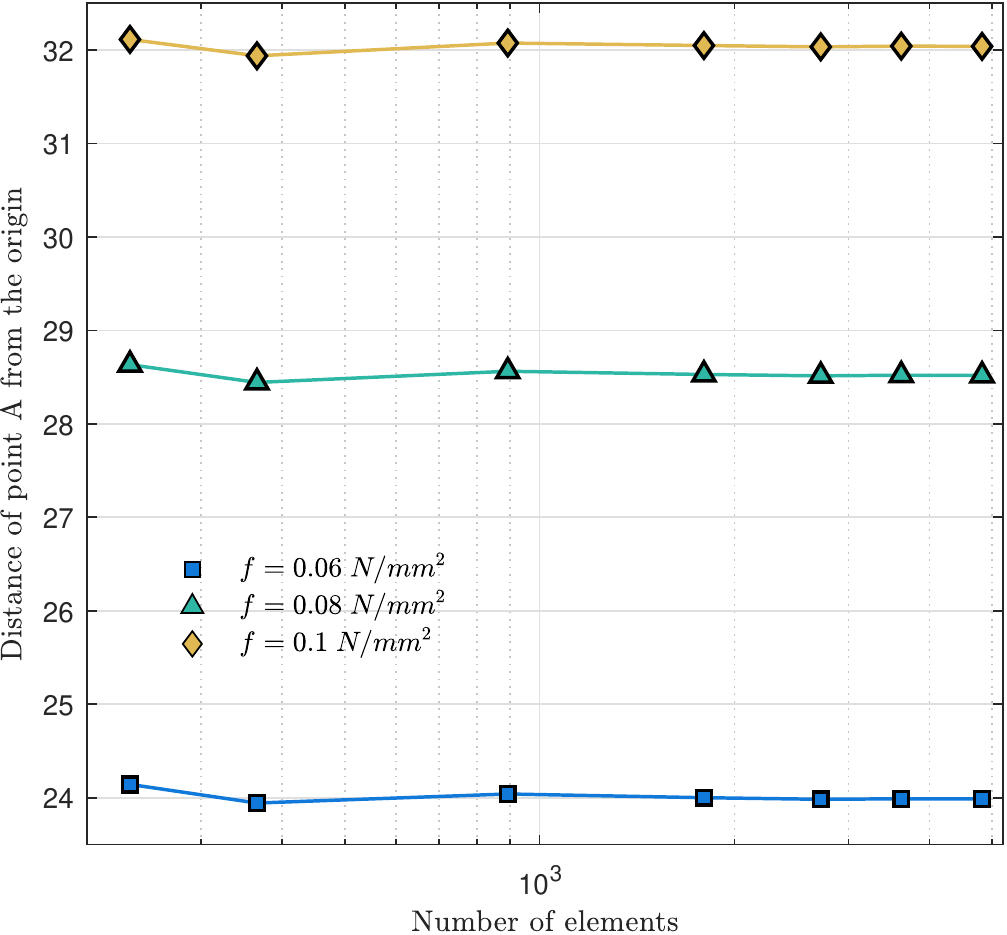}
\end{center}
\caption{\footnotesize 3D Cook's membrane: Distance of point $A$ from the origin in Figure \ref{EX2Mesh} for different values of traction $\overline{\boldsymbol{T}}_2 = \left(0,2f,f\right)$ versus the number of elements in the mesh using (\ref{FEh}). } 
\label{EX2Valid2}
\end{figure}

\begin{figure}[H]
\begin{center}
 \includegraphics[width = \textwidth]{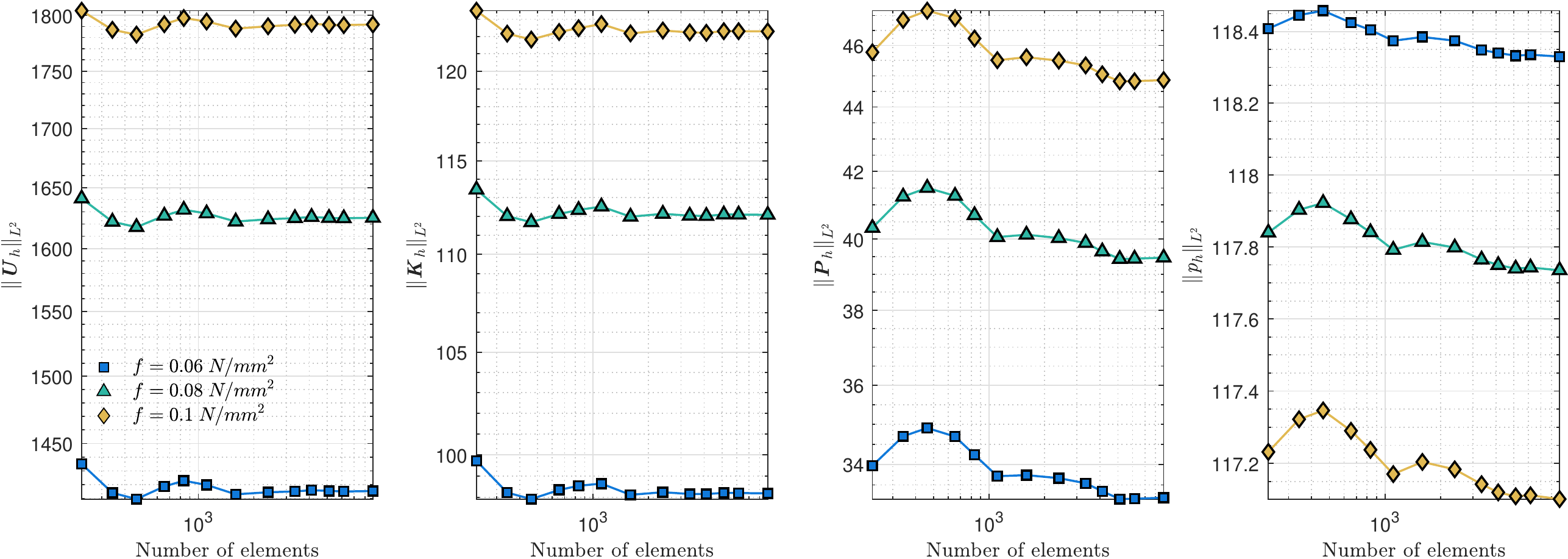}
\end{center}
\caption{\footnotesize 3D Cook's membrane: $L^{2}$-norms of displacement, displacement gradient, stress, and pressure versus the number of elements in the mesh for different values of traction $\overline{\boldsymbol{T}}_2 = \left(0,2f,f\right)$ using (\ref{FEh}).} 
\label{EX2Conv}
\end{figure}

\begin{figure}[H]
\begin{center}
\vspace*{0.2in}
\includegraphics[width = \textwidth]{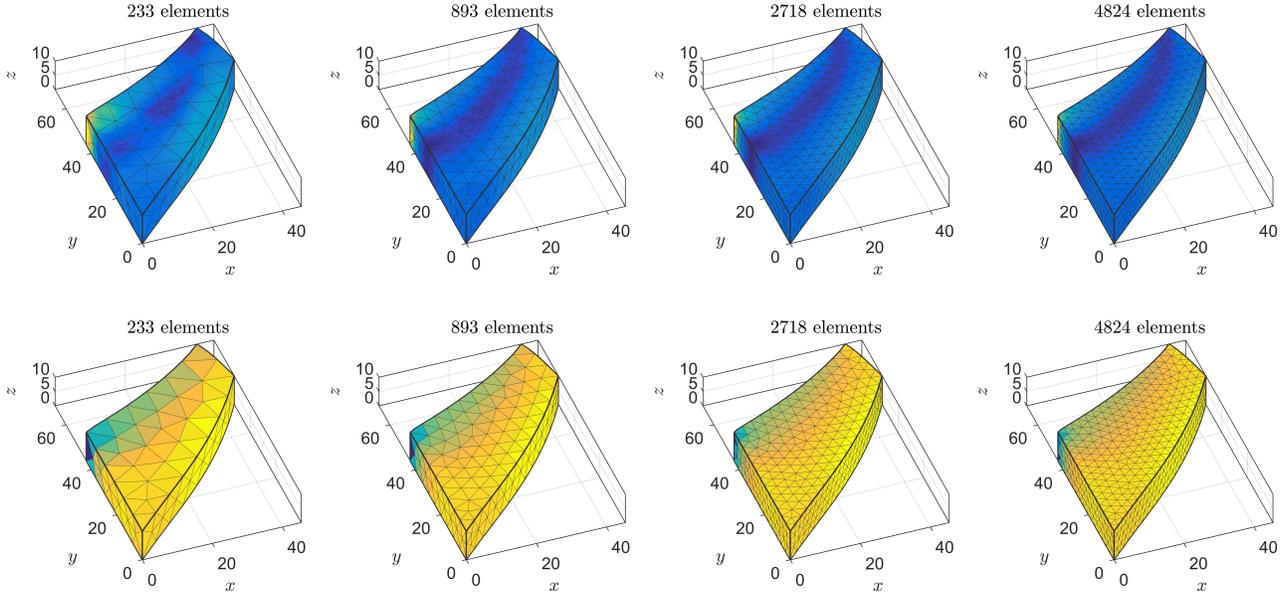} 
\end{center}
\caption{\footnotesize The deformed configurations of 3D Cook's membrane for traction $\overline{\boldsymbol{T}}_2 = (0,0.3,0)\,\mathrm{N}/\mathrm{mm}^{2}$ using (\ref{FEh}). Colors indicate values of $\|\boldsymbol{P}_h\|$ in the first row and pressure $p_h$ in the second row, where lighter colors correspond to larger values.}
\label{EX2ConfStress1}
\end{figure}

\begin{figure}[H]
\begin{center}
\includegraphics[width = \textwidth]{Fig11.pdf} 
\end{center}
\caption{\footnotesize The deformed configurations of 3D Cook's membrane for traction $\overline{\boldsymbol{T}}_2 = (0,0.2,0.1)\,\mathrm{N}/\mathrm{mm}^{2}$ using (\ref{FEh}). Colors indicate values of $\|\boldsymbol{P}_h\|$ in the first row and pressure $p_h$ in the second row, where lighter colors correspond to larger values.}
\label{EX2ConfStress2}
\end{figure}

\paragraph{Example 3. Compression of a Near-Incompressible Block.}
Let us consider a block under compression as shown in Figure \ref{EX3Mesh}. The length and the width of the block are $2\,\mathrm{mm}$ and its hight is $1\,\mathrm{mm}$. The loading square surface on the upper face of the block has an edge of $1\,\mathrm{mm}$ and is subjected to a traction $\overline{\boldsymbol{T}}=(0,0,f)$. The vertical (horizontal) displacement at the bottom (top) of the block is zero. As shown in Figure \ref{EX3Mesh}, using symmetry the meshes are generated for only a quarter of the block.
\begin{figure}[H]
\vspace*{0.2in}
\begin{center}
\includegraphics[width = 1\textwidth]{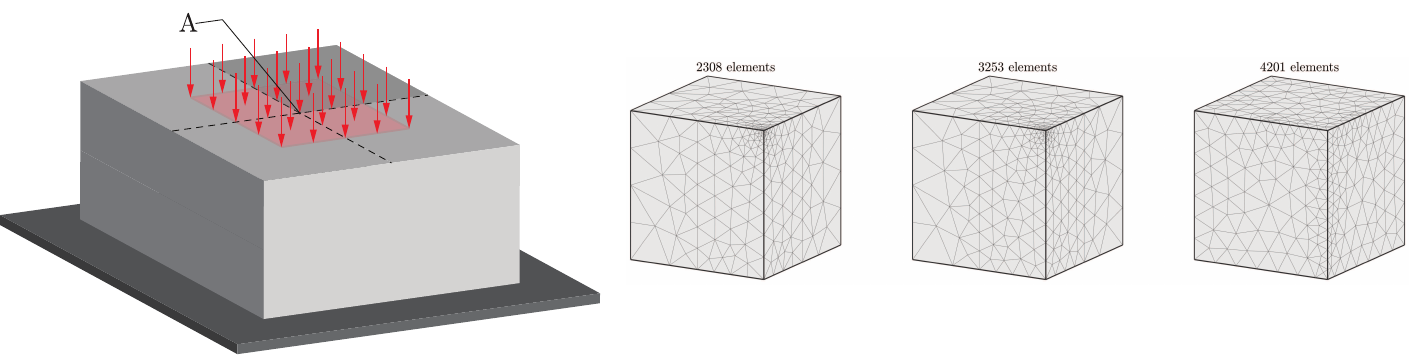}
\end{center}
\caption{\footnotesize Compression of a near-incompressible block: Geometry and three unstructured meshes.The length and the width of the block are $2\,\mathrm{mm}$ and its hight is $1\,\mathrm{mm}$. The loading square surface on the top has an edge of $1\,\mathrm{mm}$. }
\label{EX3Mesh}
\end{figure}

In this example we test the performance of the mixed method (\ref{FEh_2}) in the near-incompressible regime.
Note that many of the existing finite element methods are unable to solve this problem or suffer from numerical artifacts. \citet{reese2000new} developed a reduced-integration stabilized brick element and used it to solve this problem.
To compare our numerical results to those of \citep{reese2000new}, we consider the energy function (\ref{Wcomp}) with $\lambda = 400889.806\,\mathrm{N}/\mathrm{mm}^{2}$ and $\mu = 80.194\,\mathrm{N}/\mathrm{mm}^{2}$. 
Figure \ref{EX3Valid} illustrates the convergence of the vertical displacement of point $A$ (see Figure \ref{EX3Mesh}) for different values of $\overline{\boldsymbol{T}}=(0,0,f)$. The results obtained using (\ref{FEh_2}) agree with those reported by \citet{reese2000new}. Figure \ref{EX3ConfStrain} depicts the deformed configuration of the block for $\overline{\boldsymbol{T}}=(0,0,320)~\mathrm{N}/\mathrm{mm}^{2}$. Colors show the values of $\|\boldsymbol{K}_h\|$, where lighter colors are assigned to larger values. 
\begin{figure}[H] 
\vspace*{0.3in}
\begin{center}
\includegraphics[width = 0.5\textwidth]{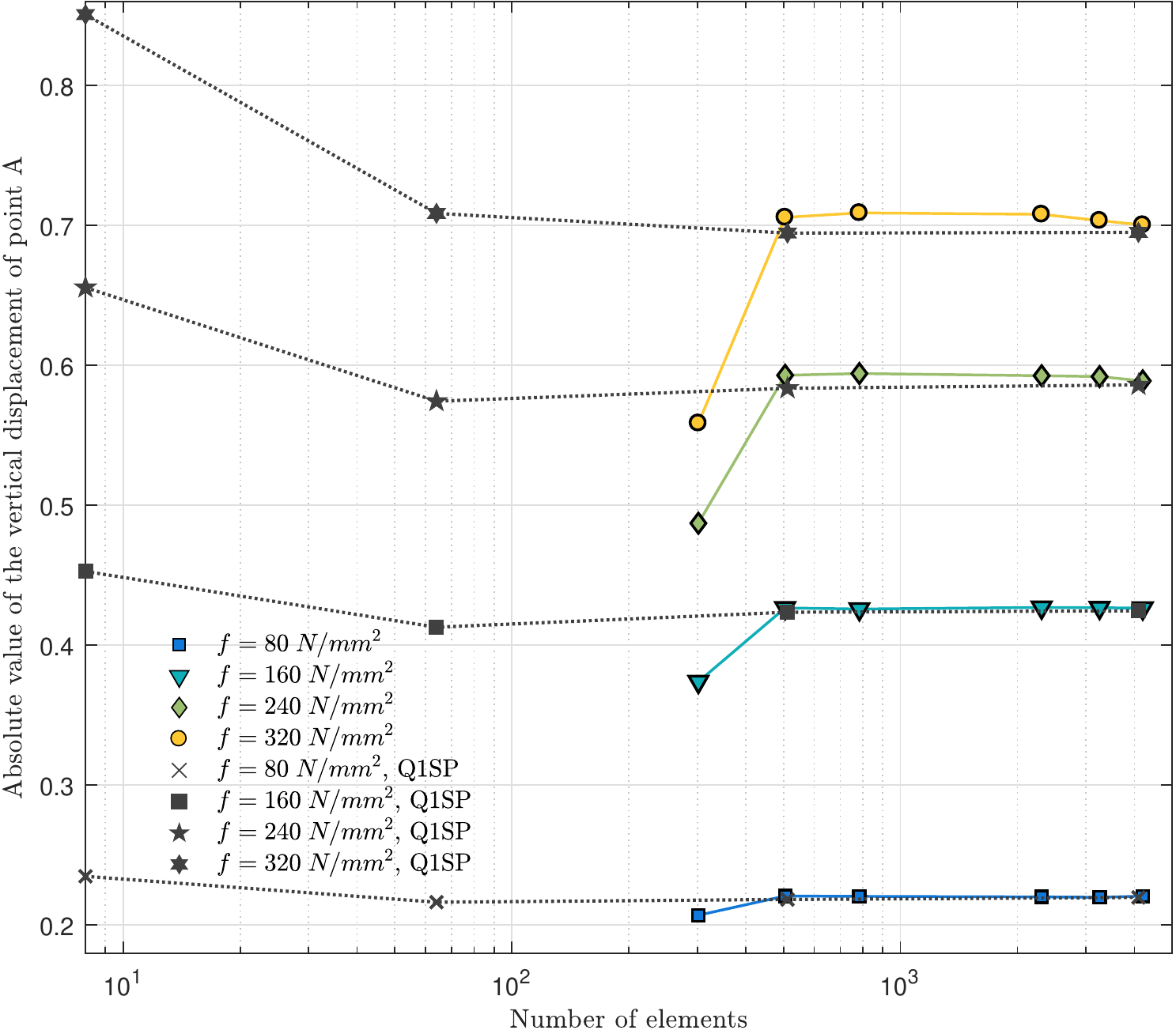}
\end{center}
\caption{\footnotesize  Compression of a near-incompressible block: Absolute value of the vertical displacement of point $A$ in Figure \ref{EX3Mesh} for different values of traction $\overline{\boldsymbol{T}}=(0,0,f)$ versus the number of elements using (\ref{FEh_2}). Q1SP indicates the results obtained by a reduced-integration stabilized brick element given in \citep{reese2000new}.}
\label{EX3Valid}{}
\end{figure}

\begin{figure}[H]
\begin{center}
\vspace*{0.2in}
\includegraphics[width = \textwidth]{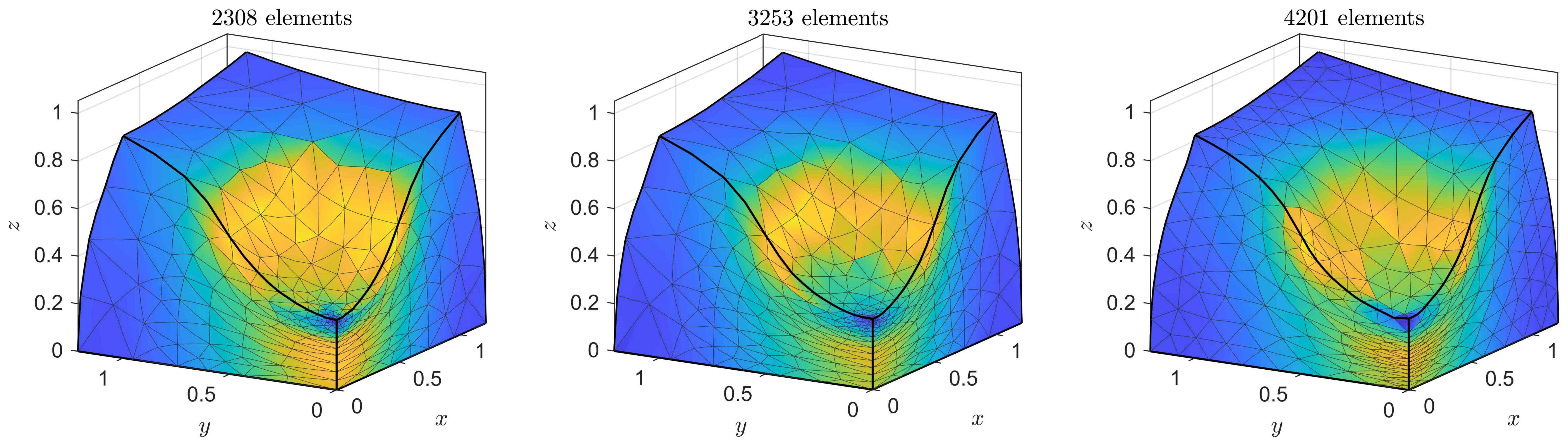} 
\end{center}
\caption{\footnotesize The deformed configurations of the near-incompressible block for traction $\overline{\boldsymbol{T}}=(0,0,320)\,\mathrm{N}/\mathrm{mm}^{2}$ using (\ref{FEh_2}). Colors indicate values of $\|\boldsymbol{K}_h\|$ with lighter colors correspond to larger values.}
\label{EX3ConfStrain}
\end{figure}

\paragraph{Example 4. Stretching a Heterogeneous Block.}
As was mentioned in Remark \ref{feature}, CSFEMs (\ref{FEh}) and (\ref{FEh_2}), by construction, satisfy the Hadamard jump condition and the continuity of traction on all the internal faces in a given mesh. This provides an efficient framework to model heterogeneous solids provided that the constituent materials do not slide at their interfaces, i.e., the displacement field is continuous at the material interfaces. One can generate a 3D mesh such that some of the internal faces of the mesh closely approximate the given material interfaces and assign the material model of each inhomogeneity to its corresponding region in the mesh. Using CSFEMs guarantees that the necessary kinematic and kinetic conditions are automatically satisfied at the material interfaces.

We consider an incompressible cubic block of edge $1\,\mathrm{mm}$ with a spherical inhomogeneity of diameter $0.5\,\mathrm{mm}$ at its center as shown in \ref{EX4Mesh}. The bottom of the block at $Z=-0.5\,\mathrm{mm}$ and the top of the block at $Z=0.5\,\mathrm{mm}$ are subjected to displacement boundaries $(0,0,-0.5)\,\mathrm{mm}$ and $(0,0,0.5)\,\mathrm{mm}$, respectively (stretch = $2$), and the other four faces are traction free. Using symmetry, we model only $1/8$ of the block as shown in Figure \ref{EX4Mesh}. 
The energy function (\ref{Wicmp}) is considered for the block with $\mu = 1\,\mathrm{N}/\mathrm{mm}^{2}$ for the matrix, and $\mu = \bar{\mu}$ for the spherical inhomogeneity. $C(J) = J-1$ is used for imposing the incompressibility constraint. We study four different cases: (i) a homogeneous block with $\bar{\mu} = 1\,\mathrm{N}/\mathrm{mm}^{2}$, (ii) a very soft inhomogeneity with $\bar{\mu} = 1e-5\,\mathrm{N}/\mathrm{mm}^{2}$, (iii) a reinforced block with $\bar{\mu} = 4\,\mathrm{N}/\mathrm{mm}^{2}$, and (iv) a rigid inhomogeneity with $\bar{\mu} = 1e5\,\mathrm{N}/\mathrm{mm}^{2}$.
Figure \ref{EX4Conv} illustrates the convergence of the $L^2$-norm of the field variables $(\boldsymbol{U}_h,\boldsymbol{K}_h, \boldsymbol{P}_h, p_h)$ calculated in the matrix for all the four cases (the values of $p_h$ become disproportionately large in the inhomogeneity for case (iv)). One can see that a significant change in the material properties of the inhomogeneity only slightly changes the convergence of the method.

Figure \ref{EX2Conf} shows the deformed configurations of $1/8$ of the block for all the four cases for a mesh consisting of $5450$ elements. This corresponds to the last points on the convergence graphs given in Figure \ref{EX4Conv}. Colors indicate the values of $\|\boldsymbol{K}_h\|$, $\|\boldsymbol{p}_h\|$, and $p_h$ in the first, second, and third row, respectively, where the lighter colors are associated with larger values. As expected, the values of $\|\boldsymbol{K}_h\|$ ($\|\boldsymbol{P}_h\|$) in the inhomogeneity decrease (increase) as the inhomogeneity becomes stiffer.  In contrast to case (i), one can see a discontinuous change of color from the matrix to the inhomogeneity in cases (ii)-(iv). As expected, the values of $\boldsymbol{K}_h$, $\boldsymbol{P}_h$, and $p_h$ are continuous at the interface of the two regions in case (i) (homogeneous block) but they are discontinuous in cases (ii)-(iv) (heterogeneous blocks). Nevertheless, in all the four cases, the interface conditions are satisfied, i.e., $\boldsymbol{K}_h\boldsymbol{T}$ and $\boldsymbol{P}_h\boldsymbol{N}$ are continuous at the interface of the two regions, where $\boldsymbol{T}$ and $\boldsymbol{N}$ are respectively a tangent vector field and a normal vector field on the interface.
For case (ii), one observes that $\|\boldsymbol{p}_h\|$ is almost uniformly zero in the spherical inhomogeneity. Hence, the traction field on the interface of the two regions is zero as well, which must be the case as a very soft inhomogeneity behaves like a hole. We solved another example by considering a block with the same geometry and the same boundary conditions but with an actual hole. It was observed that the $L^2$-norm of all the four field variables are equal to those calculated in the matrix for the case (ii). 

\begin{figure}[H]
\vspace*{0.3in}
\begin{center}
\includegraphics[width = 0.5\textwidth]{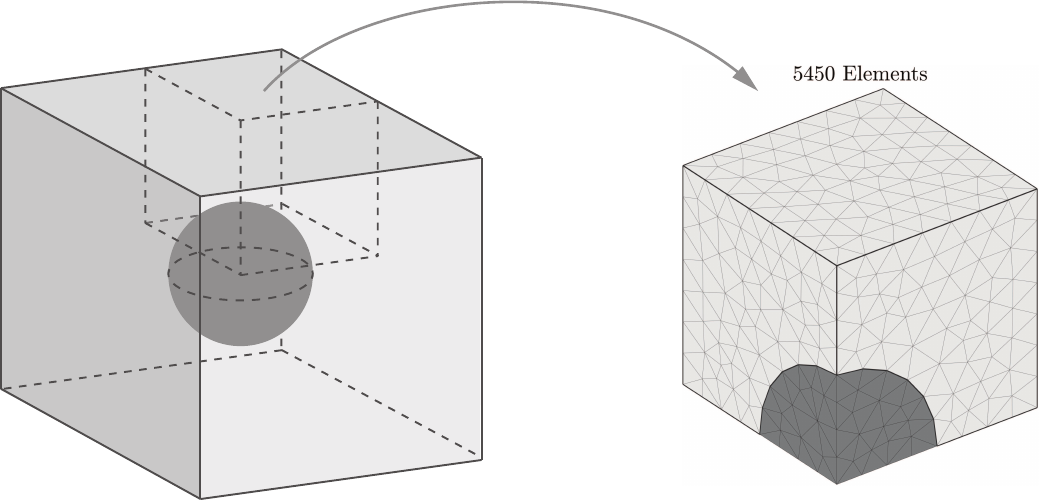}
\end{center}
\caption{\footnotesize Stretching a heterogeneous block: Geometry and an unstructured mesh. The block has an edge of $1\,\mathrm{mm}$ and the sphere at the center has a diameter of $0.5\,\mathrm{mm}$. The bottom and the top faces of the block are subjected to equal and opposite vertical displacements resulting in stretch of the block. The other four faces are traction free.}
\label{EX4Mesh} 
\end{figure}

\begin{figure}[H]
\vspace*{0.2in}
\begin{center}
\includegraphics[width = \textwidth]{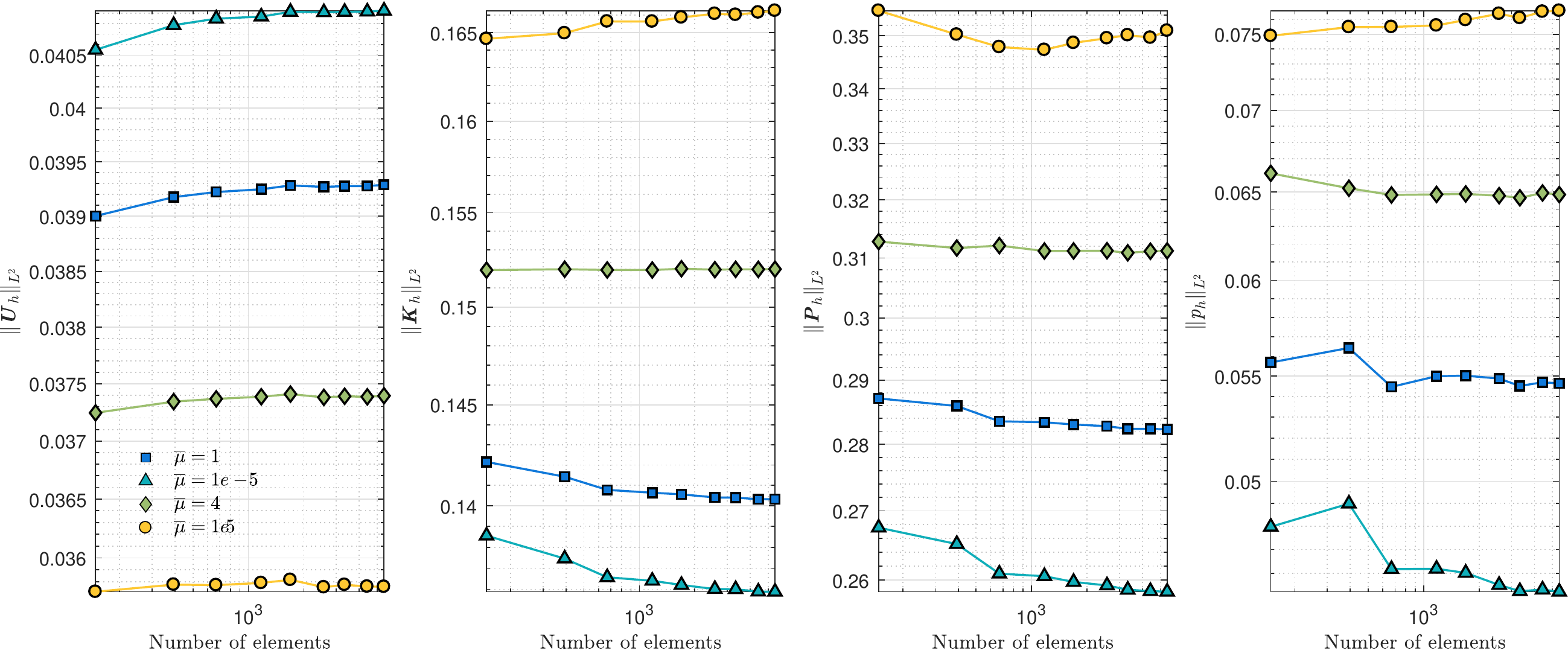} 
\end{center}
\caption{\footnotesize Stretching a heterogeneous block: $L^{2}$-norms of displacement, displacement gradient, and stress versus the number of elements in the mesh using (\ref{FEh}). The shear modulus of the incompressible matrix is $\mu = 1\, ~\mathrm{N}/\mathrm{mm}^{2}$ and $\overline{\mu}$ stands for the shear modulus of the incompressible spherical inhomogeneity.}
\label{EX4Conv}
\end{figure}

\begin{figure}[H]
\begin{center}
\includegraphics[width = \textwidth]{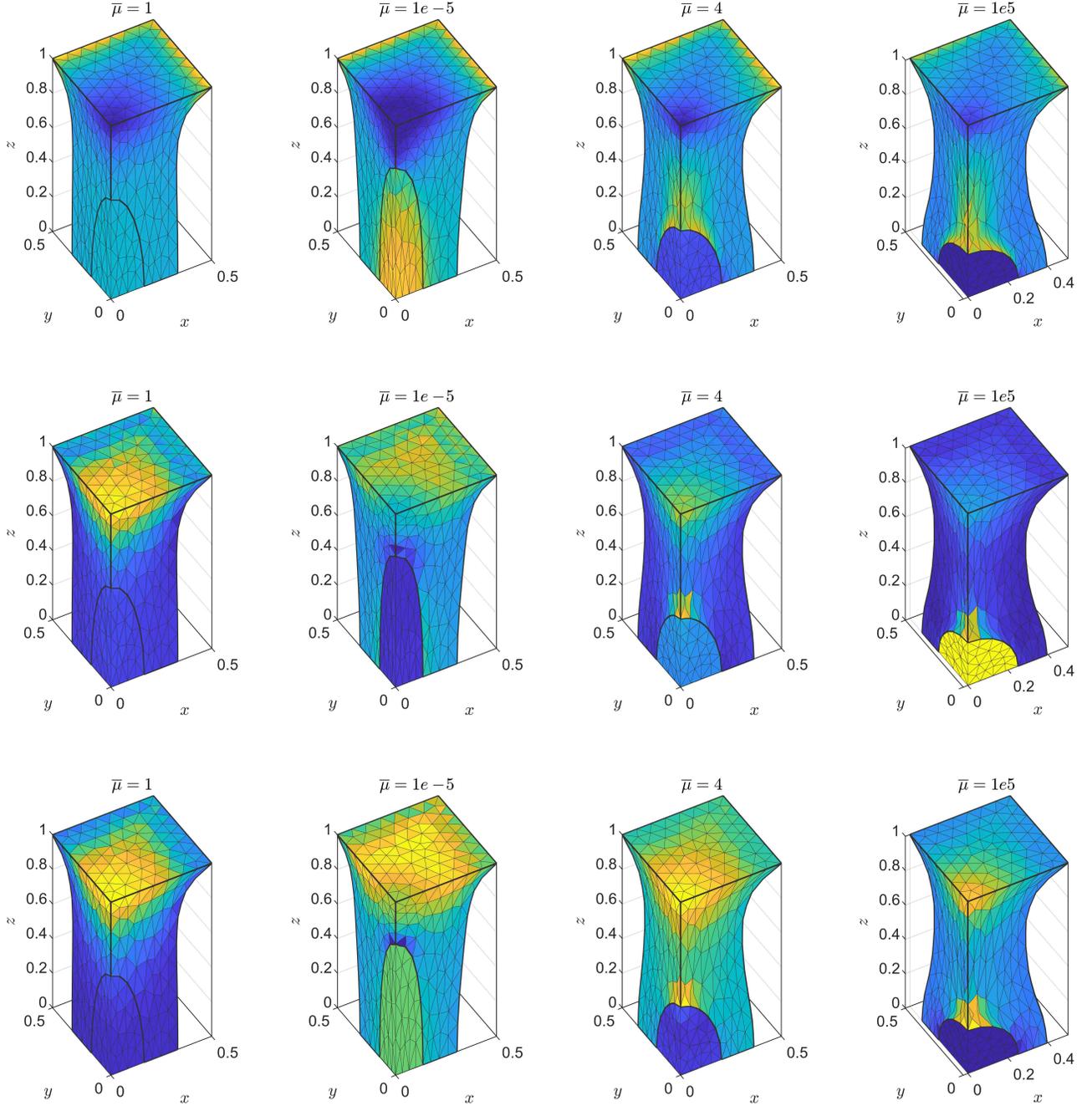} 
\end{center}
\caption{\footnotesize The deformed configurations of a block with a spherical inhomogeneity for $\lambda =2$ and considering different spherical inhomogeneities using (\ref{FEh}). The shear modulus of the incompressible matrix is $\mu = 1\, ~\mathrm{N}/\mathrm{mm}^{2}$ and $\overline{\mu}$ stands for the shear modulus of the incompressible spherical inhomogeneity in each column. Colors indicate values of  $\|\boldsymbol{K}_h\|$ in the first row, $\|\boldsymbol{P}_h\|$ in the second row, and pressure $p_h$ in the third row, where lighter colors correspond to larger values.}
\label{EX2Conf}
\end{figure}

\paragraph{Example 5: Stretching a Block with Randomly Distributed Holes.}
Next, we assess the performance of CSFEM for very large strains in a complex geometry. Let us consider an incompressible cubic block of edge $1\,\mathrm{mm}$ with $6$ spherical holes as shown in Figure \ref{EX5}. The coordinates of the centers of the holes are $(0.25, 0.6, 0.6)$, $(0.7, 0.5, 0.3)$, $(0.6, 0.2, 0.7)$, $(0.2, 0.2, 0.2)$, $(0.3, 0.8, 0.2)$, $(0.8, 0.75, 0.7)$ and their diameters are respectively $0.4$,  $0.4$,  $0.3$,  $0.3$,  $0.3$,  $0.3$. 
The left face of the block is fixed, the right face is subjected to a uniform displacement boundary $(u,0,0)$, and the other four faces are traction free. 
We use the energy function (\ref{Wicmp}) with $\mu=1\,\mathrm{N}/\mathrm{mm}^{2}$, and $C(J)=J-1$ to impose the incompressibility constraint.
The reference and the deformed configurations of the block obtained using (\ref{FEh}) for $u=2\,\mathrm{mm}$ are shown in Figure \ref{EX5}. The mesh consists of $11756$ elements and colors indicate the values of $\|\boldsymbol{K}_h\|$ with lighter colors corresponding to larger values. Note that this result corresponds to the last points on the convergence graphs given in Figure \ref{EX5Conv}. One can see that all the holes are stretched severely along the $x$-axis. Hence, relative to the $x$-axis, the beginning and the end portions of the boundary of each hole have the lower values of $\|\boldsymbol{K}_h\|$ while the middle portion has the larger values of $\|\boldsymbol{K}_h\|$.
Figure \ref{EX5Conv} illustrates the convergence of (\ref{FEh})  for different values of the displacement boundary condition $(u,0,0)$ imposed on the right face of the block. For all values of $u$, one observes that CSFEM given in (\ref{FEh}) has good convergence considering all the four independent variables $(\boldsymbol{U}_h,\boldsymbol{K}_h, \boldsymbol{P}_h, p_h)$.

\begin{figure}[H] 
\begin{center}
\vspace*{0.3in}
\includegraphics[width = 1.\textwidth]{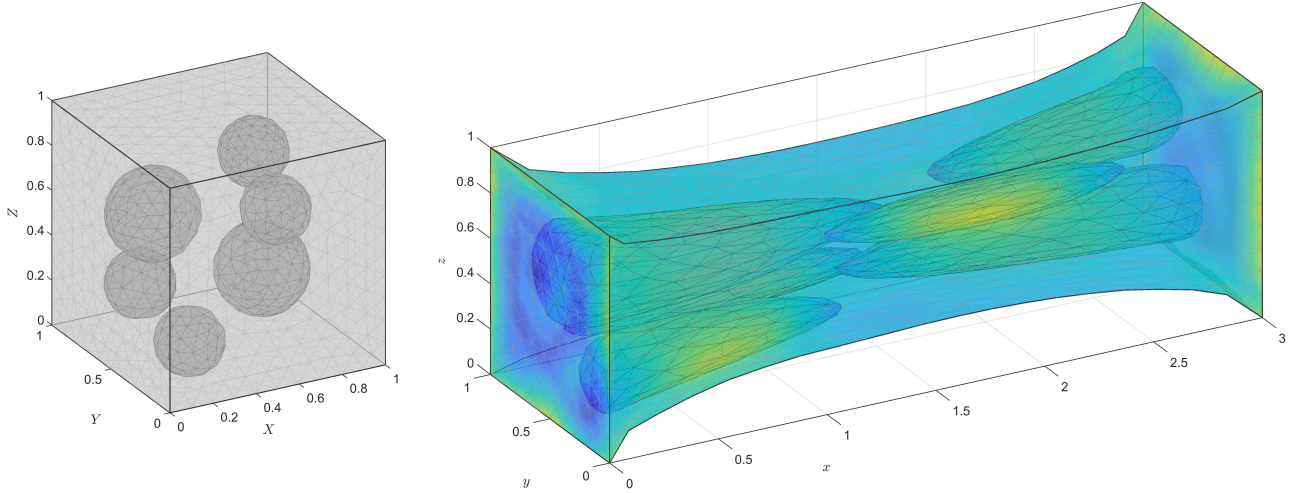}
\end{center}
\caption{\footnotesize The reference (left) and the deformed (right) configurations of a block with randomly distributed holes.
The left face of the block is fixed, the right face is subjected to a displacement $(2,0,0)\,\mathrm{mm}$ ($stretch=3$), and the other four faces are traction free. The mesh consists of $11756$ elements and the deformed configuration is obtained using (\ref{FEh}). Colors indicate values of $\|\boldsymbol{K}_h\|$, where lighter colors correspond to larger values.} 
\label{EX5}
\end{figure}

\begin{figure}[H]
\begin{center}
\vspace*{0.3in}
 \includegraphics[width = \textwidth]{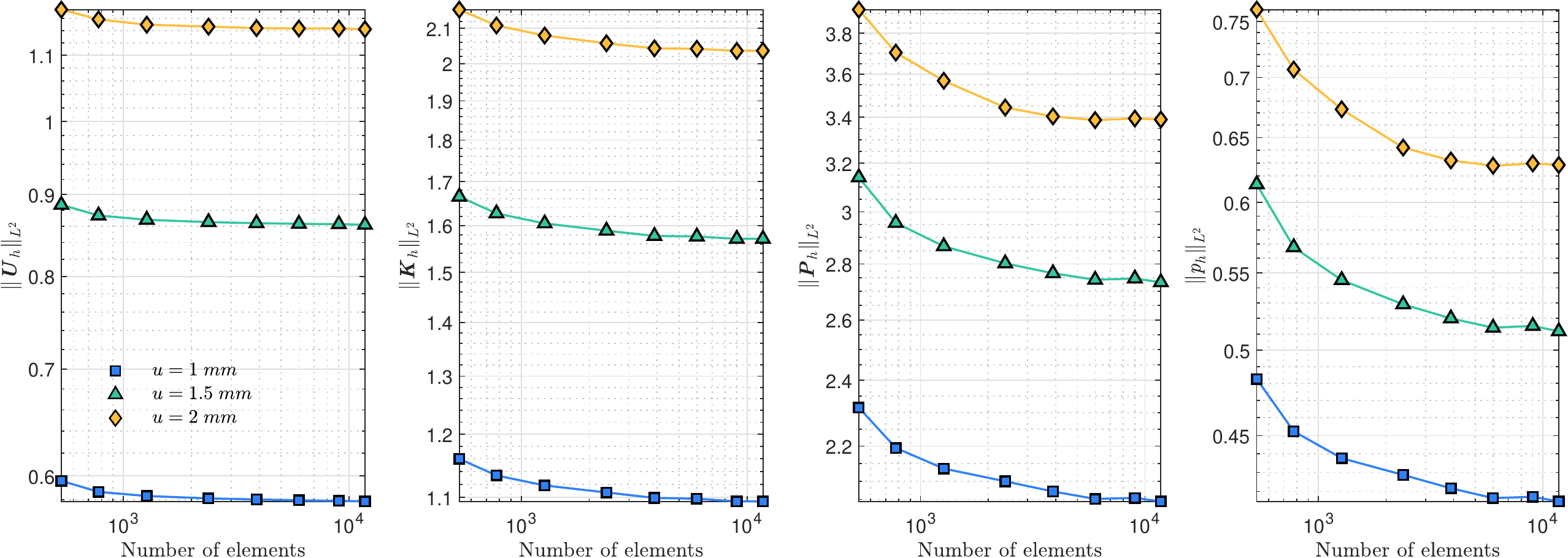}
\end{center}
\caption{\footnotesize Stretching a block with randomly distributed holes: $L^{2}$-norms of displacement, displacement gradient, stress, and pressure versus the number of elements in the mesh for different values of the displacement boundary $(u,0,0)$ using (\ref{FEh}).} 
\label{EX5Conv}
\end{figure}

\section{Concluding Remarks}

A new mixed finite element method for 3D compressible and incompressible nonlinear elasticity was introduced. This work is an extension of \citep{AnFSYa2017} and \citep{FaYa2018} to three-dimensional nonlinear elasticity problems. 
We proposed a new four-field mixed formulation for incompressible nonlinear elasticity in terms of the displacement $\boldsymbol{U}$, the displacement gradient $\boldsymbol{K}$, the first Piola-Kirchhoff stress $\boldsymbol{P}$, and a pressure-like field $p$. By setting $p=0$ in this formulation, one can readily obtain a three-filed mixed formulation for compressible solids. In the present formulation it is assumed that $(\boldsymbol{U},\boldsymbol{K},\boldsymbol{P},p)\in H^{1}(T\mathcal{B})\times H^{\mathbf{c}}(\mathcal{B}) \times H^{\mathbf{d}}(\mathcal{B})\times L^{2}(\mathcal{B})$. The new formulation has some additional terms compared with those used for 2D finite elements in \citep{FaYa2018} that vanish for the exact solutions. Provided with a proper discretization, the extra terms improve the stability of the resulting mixed finite element methods without compromising consistency. 
To obtain the mixed finite element methods, first four conforming finite element spaces were defined and then were used for approximating the four field variables. The discrete fields of the CSDEMs are: $\boldsymbol{U}_h \in V_{h,2}^{1} \subset H^{1}(T\mathcal{B}_h)$, $\boldsymbol{K}_h \in \overline{V}_{h,3}^{\mathbf{c}}\subset H^{\mathbf{c}}(\mathcal{B}_h)$, $\boldsymbol{P}_h \in V_{h,1}^{\mathbf{d}-} \subset H^{\mathbf{d}}(\mathcal{B}_h)$, and $p_h \in V_{h,0}^{\ell}\subset L^2(\mathcal{B}_h)$. 
The discrete spaces $V_{h,2}^{1}$, $V_{h,1}^{\mathbf{d}-}$, and $V_{h,0}^{\ell}$ are constructed using the second-order Lagrange elements, the first-order N\'{e}d\'{e}lec $1^{\mathrm{st}}$-kind \emph{face} elements, and the piecewise constant elements, respectively.
The discrete space $\overline{V}_{h,3}^{\mathbf{c}}$ is constructed using the first-order N\'{e}d\'{e}lec $2^{\mathrm{nd}}$-kind \emph{edge} elements and is enriched by  volume-based third-order shape functions of N\'{e}d\'{e}lec $1^{\mathrm{st}}$-kind \emph{edge} elements. 
Due to interelement continuities of these conforming spaces, our proposed mixed methods by construction provide a continuous approximation of the displacement field and satisfy both the Hadamard jump condition and the continuity of traction at the discrete level.
We solved several 3D numerical examples using CSFEMs. Our observations indicate that CSFEMs have a robust performance for bending, tension, and compression problems, and in the near-incompressible and the incompressible regimes. They are also capable of modeling problems with very large strains and accurately approximating stresses. Moreover, they seem to be free from numerical artifacts such as checkerboarding of pressure, hourglass instability, and locking.

\paragraph{Acknowledgments.} This research was supported by AFOSR -- Grant No. FA9550-12-1-0290.

\bibliographystyle{unsrtnat}
\bibliography{Ref3D}

\end{document}